\documentclass[11pt]{amsart}
\usepackage[top=1.0in, bottom=1.0in, left=1.0in, right=1.0in]{geometry}

\usepackage{amsfonts, amssymb, amsmath, enumerate}
\usepackage{hyperref}
\usepackage[foot]{amsaddr}
\usepackage{bbm}
\usepackage{xcolor}
\usepackage[style=nature]{biblatex}
\usepackage{wasysym}
\hypersetup{
	colorlinks,
	linkcolor={red!40!gray},
	citecolor={blue!40!gray},
	urlcolor={blue!70!gray}
}

\numberwithin{equation}{section}
\usepackage{amsfonts, amssymb, amsmath, enumerate}

\usepackage{mathrsfs}

\usepackage{mathtools}
\usepackage[normalem]{ulem}
\usepackage{comment}

\usepackage{bbm}
\usepackage{algorithm}
\usepackage{algorithmic}
\usepackage{tikz}

\numberwithin{equation}{section}

\DeclareUnicodeCharacter{2113}{l}

\definecolor{darkgreen}{cmyk}{0.6,0,0.8,0}

\DeclareMathOperator{\E}{\mathbb{E}}
\DeclareMathOperator{\R}{\mathbb{R}}
\DeclareMathOperator{\N}{\mathbb{N}}
\DeclareMathOperator{\Z}{\mathbb{Z}}

\DeclareMathOperator{\Pb}{\mathbb{P}}

\DeclareMathOperator*{\tr}{tr}
\DeclareMathOperator*{\Tr}{Tr}

\DeclareMathOperator{\polylog}{polylog}
\DeclareMathOperator{\poly}{poly}

\DeclareMathOperator{\nnz}{nnz}
\DeclareMathOperator{\sym}{sym}
\DeclareMathOperator{\cov}{Cov}

\DeclareMathOperator{\spec}{spec}

\def \etc {,\ldots,}
\def \P {\mathbb{P}}

\newcommand{\cS}{\mathcal{S}}

\DeclarePairedDelimiter{\norm}{\lVert}{\rVert}
\DeclarePairedDelimiter{\abs}{\lvert}{\rvert}

\DeclarePairedDelimiter{\paren}{(}{)}
\DeclarePairedDelimiter{\sqbr}{[}{]}

\usepackage{amsthm}
\usepackage{thmtools, thm-restate}
\declaretheorem[numberwithin=section]{theorem}
\newtheorem{proposition}[theorem]{Proposition}
\newtheorem{corollary}[theorem]{Corollary}
\newtheorem{lemma}[theorem]{Lemma}
\newtheorem*{lemmanon}{Lemma}
\newtheorem{conjecture}[theorem]{Conjecture}

\newtheorem*{osnaptracemom}{Lemma ~\ref{prop:momestdecoupmatrix}}
\newtheorem*{diffineqlemma}{Lemma ~\ref{lem:diffineq}}

\theoremstyle{remark}
\newtheorem{remark}[theorem]{Remark}

\theoremstyle{definition}
\newtheorem{definition}[theorem]{Definition}

\newtheorem*{def:OSE}{Definition \ref{def:OSE}}

\begin{document}

\title[Optimal Oblivious Subspace Embeddings with Near-optimal Sparsity]{
Optimal Oblivious Subspace Embeddings \\with Near-optimal Sparsity
}
\author{Shabarish Chenakkod}
\author{Micha{\l} Derezi\'nski}
\author{Xiaoyu Dong}
\thanks{Partially supported by DMS 2054408 and CCF 2338655. The authors are very grateful for the generous help and support of Mark Rudelson throughout the duration of this work.}
\address{University of Michigan, Ann Arbor, MI, USA}
\email{shabari@umich.edu, derezin@umich.edu}
\address{National University of Singapore, Singapore}
\email{xdong@nus.edu.sg}

\begin{abstract}
An oblivious subspace embedding is a random $m\times n$ matrix $\Pi$ such that, for any $d$-dimensional subspace, with high probability $\Pi$ preserves the norms of all vectors in that subspace within a $1\pm\epsilon$ factor. In this work, we give an oblivious subspace embedding with the optimal dimension $m=\Theta(d/\epsilon^2)$ that has a near-optimal sparsity of $\tilde O(1/\epsilon)$ non-zero entries per column~of~$\Pi$. This is the first result to nearly match the conjecture of Nelson and Nguyen [FOCS 2013] in terms of the best sparsity attainable by an optimal oblivious subspace embedding, improving on a prior bound of $\tilde O(1/\epsilon^6)$ non-zeros per column [Chenakkod et al., STOC 2024]. We further extend our approach to the non-oblivious setting, proposing a new family of Leverage Score Sparsified embeddings with Independent Columns, which yield faster runtimes for matrix approximation and regression tasks.

In our analysis, we develop a new method which uses a decoupling argument together with the cumulant method for bounding the edge universality error of isotropic random matrices. To achieve near-optimal sparsity, we combine this general-purpose approach with new traces inequalities that leverage the specific structure of our subspace embedding construction.
\end{abstract}

\maketitle
\thispagestyle{empty}
\newpage
\setcounter{page}{1}

\section{Introduction}
Subspace embeddings are one of the most fundamental techniques in dimensionality reduction, with applications in linear regression \cite{sarlos2006improved}, low-rank approximation \cite{clarkson2013low}, clustering \cite{cohen2015dimensionality}, and many more (see \cite{woodruff2014sketching} for an overview). The key idea is to construct a random linear transformation $\Pi\in\R^{m\times n}$ which maps from a large dimension $n$ to a small dimension $m$, while approximately preserving the geometry of all vectors in a low-dimensional subspace. In many applications, such embeddings must be constructed without the knowledge of the subspace they are supposed to preserve, in which case they are called \emph{oblivious subspace embeddings}.
\begin{definition} \label{def:OSE}
    Random matrix $\Pi\in\R^{m\times n}$ is an $(\varepsilon,\delta,d)$-oblivious subspace embedding (OSE) if for any $d$-dimensional subspace $T\subseteq \R^n$, it holds that
    \begin{align*}
        \Pb\Big(\forall x\in T,\quad (1-\varepsilon)\|x\|\leq \|\Pi x\|\leq (1+\varepsilon)\|x\|\Big)\geq 1-\delta.
    \end{align*}
\end{definition}
The two central concerns in constructing OSEs are: 1) how small can we make the embedding dimension $m$, and 2) how quickly can we apply $\Pi$ to a vector or a matrix. A popular way to address the latter is to use a sparse embedding matrix: If $\Pi$ has at most $s\ll m$ non-zero entries per column, then the cost of computing $\Pi x$ equals $O(s\cdot \nnz(x))$, where $\nnz(x)$ denotes the number of non-zero coordinates in $x$. Designing oblivious subspace embeddings that simultaneously optimize the embedding dimension $m$ and the sparsity $s$ has been the subject of a long line of works \cite{clarkson2013low,meng2013low,nelson2013osnap,bourgain2015toward,cohen2016nearly,chenakkod2024optimal}, aimed towards resolving the following conjecture of Nelson and Nguyen~\cite{nelson2013osnap}, which is supported by nearly-matching lower bounds \cite{nelson2014lower,li2022lower}.
\begin{conjecture}[Nelson and Nguyen, FOCS 2013 \cite{nelson2013osnap}]\label{c:nn}
    For any $n\geq d$ and $\varepsilon,\delta\in(0,1)$, there is an $(\varepsilon,\delta,d)$-oblivious subspace embedding $\Pi\in\R^{m\times n}$ with dimension $m=O((d+\log1/\delta)/\varepsilon^2)$ having $s=O(\log(d/\delta)/\varepsilon)$ non-zeros per column.
\end{conjecture}
Nelson and Nguyen gave a simple construction that they conjectured would achieve these guarantees: For each column of $\Pi$, place scaled random signs $\pm1/\sqrt s$ in $s$ random locations. They showed that this construction achieves dimension $m=O(d\polylog(d)/\varepsilon^2)$ and sparsity $s=O(\polylog(d)/\varepsilon)$. A number of follow-up works \cite{bourgain2015toward,cohen2016nearly} improved on this; most notably, Cohen \cite{cohen2016nearly} showed that a sparse OSE can achieve $m=O(d\log(d)/\varepsilon^2)$ with $s=O(\log(d)/\varepsilon)$. However, none of these guarantees recover the optimal embedding dimension $m=\Theta(d/\varepsilon^2)$, with the extraneous $\log(d)$ factor arising due to a long-standing limitation in existing matrix concentration techniques \cite{troppmatrixconc}. 

This sub-optimality in dimension $m$ was finally addressed in a recent work of Chenakkod, Derezi\'nski, Dong and Rudelson~\cite{chenakkod2024optimal}, relying on a breakthrough in random matrix universality theory by Brailovskaya and van Handel \cite{brailovskaya2022universality}. They achieved $m=\Theta(d/\varepsilon^2)$, but only with a significantly sub-optimal sparsity $s = \tilde O(1/\varepsilon^6)$, which is a consequence of how the universality error is measured and analyzed in \cite{brailovskaya2022universality} (here, $\tilde O$ hides polylogarithmic factors in $d/\varepsilon\delta$). This raises the following natural question: 

\smallskip

\textit{Can the optimal dimension $m=\Theta(d/\varepsilon^2)$ be achieved with the conjectured $\tilde O(1/\varepsilon)$ sparsity?}

\smallskip

We give a positive answer to this question, thus matching Conjecture \ref{c:nn} in dimension $m$ and nearly-matching it in sparsity $s$. To achieve this, we must substantially depart from the approach of Brailovskaya and van Handel, and as a by-product, develop a new set of tools for matrix universality
which are likely of independent interest (see Section \ref{s:overview} for an overview). Remarkably, our result is attained by one of the simple constructions that were originally suggested by Nelson and Nguyen in their conjecture.
\begin{theorem}[Oblivious Subspace Embedding]\label{t:ose}
    For any $n\geq d$ and $\varepsilon,\delta\in(0,1)$ such that $1/\epsilon\delta \leq\poly(d)$, there is an $(\varepsilon,\delta,d)$-oblivious subspace embedding $\Pi\!\in\!\R^{m\times n}$ with $m=O(d/\varepsilon^2)$ having $s=\tilde O(1/\varepsilon)$ non-zeros per column.
\end{theorem}

Many applications of subspace embeddings arise in matrix approximation \cite{woodruff2014sketching} where, given a large tall matrix $A\in\R^{n\times d}$, we seek a smaller $\tilde A\in\R^{m\times d}$ such that $\|\tilde A x\|=(1\pm\varepsilon)\|Ax\|$ for all $x\in\R^d$. 
Naturally, this can be accomplished with an $(\varepsilon,\delta,d)$-OSE matrix $\Pi\in\R^{m\times n}$, by computing $\tilde A=\Pi A$ in time $\tilde O(\nnz(A)/\varepsilon)$ and considering the column subspace of $A$. However, given direct access to $A$, one may hope to get true input sparsity time $O(\nnz(A))$ by leveraging the fact that the embedding need not be oblivious.

To that end, we adapt our subspace embedding construction, so that it can be made even sparser given additional information about the leverage scores of matrix $A$. The $i$th leverage score of $A$ is defined as the squared norm of the $i$th row of the matrix obtained by orthonormalizing the columns of $A$~\cite{drineas2006sampling}. We show that if the $i$th leverage score of $A$ is bounded by $l_i\in[0,1]$, then the $i$th column of $\Pi$ needs only $\max\{1,\tilde O(l_i/\varepsilon)\}$ non-zero entries. Since the leverage scores of $A$ can be approximated quickly \cite{drineas2012fast}, this leads to our new algorithm, Leverage Score Sparsified embedding with Independent Columns (LESS-IC), which is inspired by related constructions that use LESS with independent rows ~\cite{less-embeddings,newton-less,gaussianization}.

Just like recent prior works \cite{chepurko2022near,cherapanamjeri2023optimal,chenakkod2024optimal}, our algorithm for constructing a subspace embedding from a matrix~$A$ incurs a preprocessing cost of $O(\nnz(A)+d^\omega)$ required for approximating the leverage scores (here, $\omega$ is the matrix multiplication exponent). However, our approach significantly improves on these prior works in the $\poly(d/\varepsilon)$ embedding cost, leading to matching speedups in downstream applications such as constrained/regularized least squares~\cite{chenakkod2024optimal}.
\begin{theorem}[Fast Subspace Embedding]\label{t:fast}
    Given $A\!\in\!\R^{n\times d}$, $\varepsilon,\gamma\!\in\!(0,1)$ and $1/\varepsilon\!\leq\!\poly(d)$,~in
    \begin{align*}
        O\big(\gamma^{-1}\nnz(A) + d^\omega + \varepsilon^{-1} d^{2+\gamma}\polylog(d)\big)\quad\text{time}
    \end{align*}
    we can compute
    $\tilde A\in\R^{m\times d}$ such that $m=O(d/\varepsilon^2)$ and with probability $\geq 0.99$
  \begin{align*}
    (1-\varepsilon)\|Ax\|\leq \|\tilde A x\| \leq
    (1+\varepsilon)\|Ax\|\qquad\forall x\in\R^d.
  \end{align*}
\end{theorem}
This is a direct improvement over the previous best known runtime for constructing an optimal subspace embedding \cite{chenakkod2024optimal}, which suffers an additional $\tilde O(d^{2+\gamma}/\varepsilon^6)$ cost due to their sub-optimal sparsity. Remarkably, our result is also the first to achieve $\tilde O(d^{2+\gamma}/\varepsilon)$ dependence even if we allow a sub-optimal dimension, i.e., $m=O(d\log(d)/\varepsilon^2)$. Here, the previous best time \cite{chepurko2022near,cherapanamjeri2023optimal} has an additional $\tilde O(d^{2+\gamma}/\varepsilon^2)$ cost, due to using a two-stage leverage score sampling scheme in place of a sparse embedding matrix. Our new LESS-IC embedding is crucial in achieving the right dependence on $\varepsilon$, as neither of the previous constructions appear capable of overcoming the $\Omega(d^{2+\gamma}/\varepsilon^2)$ barrier.

As an example application of our results, we show how our fast subspace embedding construction can be used to speed up reductions for a wide class of optimization problems based on constrained or regularized least squares regression, including Lasso regression \cite{bourgain2015toward}. The following corollary follows immediately from Theorem \ref{t:fast}, and is a direct improvement over Theorem 1.8 of \cite{chenakkod2024optimal} in terms of the runtime dependence on $\epsilon$ from $\tilde O(d^{2+\gamma}/\epsilon^6)$ to $\tilde O(d^{2+\gamma}/\epsilon)$, while achieving a matching $O(d/\epsilon^2)\times d$ reduction.
\begin{corollary}[Fast reduction for constrained least squares]\label{t:reduction}
   Given $A\in\R^{n\times d}$, $b\in\R^n$, $\epsilon>0$, function $g:\R^d\rightarrow\R_{\geq 0}$ and set
    $\mathcal C\subseteq\R^d$ consider an $n\times d$  problem $\mathrm{LS}_{\mathcal C,g}(A,b,\epsilon)$:
    \begin{align*}
    \text{Find $\tilde x$ such that}\quad f(\tilde x)\leq (1+\epsilon)\min_{x\in\mathcal C}
      f(x),\quad\text{where}\quad
      f(x) = \|Ax-b\|_2^2+ g(x).
    \end{align*}
    There is an algorithm that reduces this problem to an
    $O(d/\epsilon^2)\times d$ instance $\mathrm{LS}_{\mathcal C,g}(\tilde A,\tilde
    b,0.1\epsilon)$ in $O(\gamma^{-1}\nnz(A) + d^\omega +
    \epsilon^{-1}d^{2+\gamma}\polylog(d))$  time.
  \end{corollary}

\section{Related Work}

Subspace embeddings have played a central role in the area of randomized linear algebra ever since the work of Sarlos \cite{sarlos2006improved} (for an overview, see the following surveys and monographs \cite{woodruff2014sketching,drineas2016randnla,martinsson2020randomized,derezinski2024recent}). Initially, these approaches focused on leveraging fast Hadamard transforms \cite{ailon2009fast,tropp2011improved} to achieve improved time complexity for linear algebraic tasks such as linear regression and low-rank approximation. Clarkson and Woodruff \cite{clarkson2013low} were the first to propose a sparse subspace embedding matrix, the CountSketch, which has exactly one non-zero entry per column but does not recover the optimal embedding dimension guarantee. Before this, the idea of using a sparse random matrix for dimensionality reduction was successfully employed in the context of Johnson-Lindenstrauss embeddings \cite{dasgupta2010sparse,kane2014sparser}, which seek to preserve the geometry of a finite set, as opposed to an entire subspace. 

In addition to the aforementioned efforts in improving sparse subspace embeddings \cite{clarkson2013low,meng2013low,nelson2013osnap,bourgain2015toward,cohen2016nearly,chenakkod2024optimal}, some works have aimed to develop fast subspace embeddings that achieve optimal embedding dimension either without sparsity \cite{chepurko2022near,cherapanamjeri2023optimal}, under additional assumptions \cite{cartis2021hashing}, or with one-sided embedding bounds \cite{tropp2025comparison}. Our time complexity result, Theorem \ref{t:fast}, improves on all of these in terms of the dependence on $\varepsilon$, thanks to a combination of our new analysis techniques and the new LESS-IC construction.

\section{Main Results}

In this section, we define the subspace embedding constructions used in our results, and provide detailed statements of our theorems.

As is customary in the literature, we shall work with an equivalent form of the subspace embedding guarantee from Definition \ref{def:OSE}, which frames this problem as a characterization of the extreme singular values of a class of random matrices. Namely, consider a deterministic $n\times d$ matrix $U$ with orthonormal columns that form the basis of a $d$-dimensional subspace~$T$. Then, a random matrix $\Pi\in\R^{m\times n}$ is an $(\varepsilon,\delta,d)$-subspace embedding for $T$ if and only if all of the singular values of the matrix $\Pi U$ lie in $[1-\varepsilon,1+\varepsilon]$ with probability $1-\delta$, i.e., 
\begin{align}
\Pr(1-\varepsilon\leq s_{\min}(\Pi U)\leq s_{\max}(\Pi U)\leq 1+\varepsilon)\geq 1-\delta,\label{eq:ose-equiv}
\end{align}
where $s_{\min}$ and $s_{\max}$ denote the smallest and largest singular values. To ensure that $\Pi$ is an oblivious subspace embedding, we must therefore ensure \eqref{eq:ose-equiv} for the family of all random matrices of the form $\Pi U$, where $U$ is any $n\times d$ matrix with orthonormal columns.

\subsection{Oblivious Subspace Embeddings}

Our subspace embedding guarantees are achieved by a family of OSEs which have a fixed number of non-zero entries in each column, a key property that was also required of sparse OSE distributions called OSNAP described by Nelson and Nguyen \cite{nelson2013osnap}. As we explain later, our analysis techniques apply to other natural families of sparse embedding distributions, including those with i.i.d.~entries \cite{achlioptas2003database}, however the OSNAP-style construction is crucial for achieving the near-optimal sparsity $s=\tilde O(1/\varepsilon)$.

In our construction of the $m\times n$ OSE matrix $\Pi$, we start by defining an unscaled version of the matrix, called $S$, which has entries in $\{-1,0,1\}$. We then scale $S$ to appropriately normalize the entry-wise variances, obtaining $\Pi$.
Concretely, we wish to obtain an $m \times n$ sparse random matrix $S$ which has exactly $s$ non-zero $\pm1$ entries in each column. Assume $s$ exactly divides $m$. Then we can divide each column of $S$ into $s$ subcolumns and randomly populate one entry in each subcolumn by a Rademacher random variable (see Figure \ref{fig:osnap}). We call this family of distributions (unscaled) OSNAP, carrying over Nelson and Nguyen's terminology (technically, their definition is somewhat broader than ours).

\begin{figure}[h!]
    \centering
    \begin{tikzpicture}[scale=1]
      \pgfmathsetmacro{\a}{1.4};
      \pgfmathsetmacro{\b}{2.8};
      \pgfmathsetmacro{\c}{3.2};
      \pgfmathsetmacro{\d}{4};
      \pgfmathsetmacro{\aa}{0.2};
      \pgfmathsetmacro{\aaa}{0};
      \pgfmathsetmacro{\bbb}{0.8};
      \pgfmathsetmacro{\ccc}{1.6};
      \pgfmathsetmacro{\ddd}{2.6};
      \pgfmathsetmacro{\eee}{3.8};
      \pgfmathsetmacro{\fff}{4.8};
      \draw (0, -1) rectangle (5,2);
      \draw (2.5,2.25) node {\mbox{\scriptsize Embedding matrix $\Pi$}};

    \draw (\ccc,-1) rectangle (\ccc + 0.2,0); 
    \draw (\ccc+0.1,-0.3) node {$*$};
    \draw (\ccc,0) rectangle (\ccc + 0.2,1);
    \draw (\ccc+0.1,0.5) node {$*$};
    \draw (\ccc,1) rectangle (\ccc + 0.2,2);
    \draw (\ccc+0.1,1.9) node {$*$};

\end{tikzpicture}
    \caption{An example of a column divided into $s=3$ subcolumns with each subcolumn having exactly one non-zero entry in a random position.}
    \label{fig:osnap}
\end{figure}

Each non-zero entry in the matrix $S$ can be identified by a tuple $(l, \gamma) \in [n]\times [s]$ where $l$ identifies the column of the non-zero entry and $\gamma$ is the index of the entry in that column. Thus the $(l,\gamma)\textsuperscript{th}$ non-zero entry in $S$ is located in column $l$ and row $\mu_{(l,\gamma)}$, where $\mu_{(l,\gamma)}$ is a uniformly chosen integer from the interval $[(m/s)(\gamma-1)+1:(m/s)\gamma]$. For example, the $(1,1)\textsuperscript{th}$ non-zero entry in $S$ is located in column $1$ and some row in the interval $[1:m/s]$. An $m \times n$ matrix with a non-zero entry in column $l$ and row $\mu_{(l,\gamma)}$ is given by $e_{\mu_{(l,\gamma)}}e_l^T$, where $e_{\mu_{(l, \gamma)}}$ and $e_l$ represent standard basis vectors in $\R^m$ and $\R^n$ respectively, and for $S$ we wish to place a random sign $\xi_{(l,\gamma)}$ at this position. This motivates our formal definition for OSNAP,

\begin{definition}[OSNAP]\label{def:osnap} 
An $m \times n$ random matrix $S$ is called an unscaled oblivious sparse norm-approximating projection with $K$-wise independent subcolumns ($K$-wise independent unscaled OSNAP) with parameters $p, \varepsilon, \delta \in (0,1]$ such that $s=pm$ divides $m$ if,
\[ S = \sum_{l=1}^n \sum_{\gamma=1}^s \xi_{(l,\gamma)} e_{\mu_{(l, \gamma)}} e_l ^\top \]
where,
\begin{itemize}
    \item $\{ \xi_{(l,\gamma)} \}_{l \in [n], \gamma \in [s]}$ is a collection of $K$-wise independent Rademacher random variables.
    \item $\{ \mu_{(l,\gamma)} \}_{l \in [n], \gamma \in [s]}$ is a collection of $K$-wise independent random variables such that each $\mu_{(l,\gamma)}$ is uniformly distributed in $[(m/s)(\gamma-1)+1:(m/s)\gamma]$.
    \item The collection $\{ \xi_{(l,\gamma)} \}_{l \in [n], \gamma \in [s]}$ is independent from the collection $\{ \mu_{(l,\gamma)} \}_{l \in [n], \gamma \in [s]}$.
\end{itemize}

In this case, $\Pi = (1/\sqrt{pm})S$ is called a $K$-wise independent OSNAP with parameters $p, \varepsilon, \delta$. In addition, if all the random variables in the collections $\{ \xi_{(l,\gamma)} \}_{l \in [n], \gamma \in [s]}$ and $\{ \mu_{(l,\gamma)} \}_{l \in [n], \gamma \in [s]}$ are fully independent, then $S$ is called a fully independent unscaled OSNAP and $\Pi$ is called a fully independent OSNAP.
\end{definition}

Thus, each column of the OSNAP matrix $\Pi$ has $s = pm$ many non-zero entries, and the sparsity level can be varied by setting the parameter $p \in [0,1]$ appropriately. With the distribution formally defined, we now provide the full statement of our subspace embedding guarantee for OSNAP,

\begin{restatable}[Subspace Embedding Guarantee for OSNAP]{theorem}{osnapmainthm}
\label{t:ose-full}
Let $\Pi = (1/\sqrt{pm})S$ be an $m \times n$ matrix distributed according to the $8 \lceil\log (\frac{d}{\varepsilon \delta})\rceil$-wise independent OSNAP distribution with parameter $p$. Let $U$ be an arbitrary $n \times d$ deterministic matrix such that $U^\top U=I$. Then, there exist positive constants $c_{\ref*{t:ose-full}.1}$ and $c_{\ref*{t:ose-full}.2}$ such that for any $0 < \delta, \varepsilon < 1$ and $d>10$,  we have 
\begin{align*}
\Pb \left( 1 - \varepsilon  \leq s_{\min}(\Pi U)   \leq s_{\max}(\Pi U) \leq 1 + \varepsilon \right) \geq 1-\delta
\end{align*}
if the embedding dimension satisfies $m \ge c_{\ref*{t:ose-full}.1}  (d + \log(1/\delta\varepsilon))/\varepsilon^2$ and the sparsity $s=pm$ satisfies $s \ge \min\{c_{\ref*{t:ose-full}.2} (\log ^2(\frac{d}{\varepsilon\delta})/\varepsilon+\log^3(\frac{d}{\varepsilon\delta})),m\}$ non-zeros per column.
\end{restatable}

\begin{remark}
We note that if $1/\varepsilon$ is polynomial in $d/\delta$, i.e., $\varepsilon \ge \frac{1}{(d/\delta)^K}$ for some absolute constant $K\geq 1$, then the $\log(1/\varepsilon)$ term in $\log(d/\varepsilon\delta) = \log(d/\delta) + \log(1/\varepsilon)$ is dominated by $\log(d/\delta)$. In this case, our requirement will become
\begin{align*}pm \ge \min \left\{C(K) \paren*{\frac{(\log (d/\delta))^2}{\varepsilon}+(\log (d/\delta))^3},m \right\}
\end{align*}
for some constant $C(K)$ depending only on $K$.
A weaker lower bound on $\varepsilon$, $\varepsilon > 1/e^d$ is sufficient to reduce the requirement on $m$ to:
\[ m \ge 2c_{\ref*{t:ose-full}.1}  \frac{d + \log(1/\delta)}{\varepsilon^2}.   \]

\end{remark}

This is a direct improvement over Theorem 1.2 of  \cite{chenakkod2024optimal}, which requires sparsity $s\geq c\log^4(d/\delta)/\varepsilon^6$ where $c$ is an absolute constant, with the same condition on $m$. The primary gain lies in the polynomial dependence on $1/\varepsilon$, but we note that our result also achieves a better logarithmic dependence on $d$, which means that an improvement is obtained even for $\varepsilon=\Theta(1)$.

Our techniques can be used to obtain a similar result for a simple OSE model with i.i.d.~sparse Rademacher entries \cite{achlioptas2003database}, which was also considered by \cite{chenakkod2024optimal}. However, in this case, we need an additional requirement of $s=pm \ge c\log (\frac{d}{\varepsilon\delta})/\varepsilon^2$ for the sparsity (see Section \ref{sec:oseieproof} for details; this is again a direct improvement over a result of \cite{chenakkod2024optimal}).
\begin{remark}
The $1/\varepsilon^{2}$ factor in the column–sparsity of an OSE model with i.i.d. entries is unavoidable.  
To see why, let
\[
  U=\begin{bmatrix}I_{d}\\[2pt] 0\end{bmatrix},
  \qquad
  \Pi=\frac{1}{\sqrt{pm}}\;S ,
\]
and note that
$\sigma_{\min}(\Pi U),\sigma_{\max}(\Pi U)\in[1-\varepsilon,\,1+\varepsilon]$
forces, for every $j\le d$,
\begin{equation}\label{oseielower}
  \bigl|\|\Pi e_j\|_{2}^{2}-1\bigr|
  =\Bigl|\tfrac{N_j}{pm}-1\Bigr|
  \le\varepsilon,
  \qquad
  N_j:=\operatorname{nnz}(S e_j)\sim\operatorname{Binomial}(m,p).
\end{equation}
 
Set
\[
  Z:=\frac{N_j-pm}{\sqrt{mp(1-p)}} ,
  \qquad
  a:=\varepsilon\,\sqrt{\frac{mp}{1-p}}\le\sqrt2\,\varepsilon\sqrt{mp}\quad( \text{for } p\le\tfrac12).
\]
Condition \ref{oseielower} is equivalent to $|Z|\le a$.  
With $F_Z(x)=\Pr[Z\le x]$ and $\Phi$ the standard normal cumulative distribution function,
the Berry Esseen theorem gives
\[
  \sup_{x\in\mathbb R}|F_Z(x)-\Phi(x)|\le\frac{6}{\sqrt{mp}}.
\]
Hence
\[
  \Pr\bigl[|Z|\le a\bigr]
  \;=\;F_Z(a)-F_Z(-a)
  \;\le\;\bigl(\Phi(a)-\Phi(-a)\bigr)+\frac{12}{\sqrt{mp}}.
\]
Using $\Phi(a)-\Phi(-a)=2\int_{0}^{a}\!\phi(t)\,dt\le a/\sqrt{\pi}$ and the bound on $a$, we have
\[
  \Pr\left( \bigl|\|\Pi e_j\|_{2}^{2}-1\bigr| \le \varepsilon \right)
  \;\le\;
  \frac{a}{\sqrt{\pi}}+\frac{12}{\sqrt{mp}}
  \;\le\;
  \frac{\sqrt2}{\sqrt{\pi}}\varepsilon\sqrt{mp}+
  \frac{12}{\sqrt{mp}}
\]

By general lower bounds for OSE, we know that, when $\varepsilon \to 0$, we need $pm \to \infty$ and therefore so $\frac{12}{\sqrt{mp}} \to 0$.

Therefore, for small enough $\varepsilon$, if $pm<c/\varepsilon^{2}$ with $c:=\frac{1}{81}$, the right–hand side is $<\tfrac13$.
Thus any OSE-IE that succeeds with constant probability must satisfy $pm=\Omega\!\bigl(\varepsilon^{-2}\bigr)$.
\end{remark}

\subsection{Characterization via a Moment Property}
\label{s:osemoments}

Our proof techniques for Theorem \ref{t:ose-full} are based on the moment method, and thus, they naturally imply the following slightly stronger moment-based characterization of an oblivious subspace embedding, which was proposed by \cite{cohen2016optimal} as an extension of the corresponding moment-based characterization of a Johnson-Lindenstrauss embedding \cite{kane2014sparser}.
\begin{definition}\label{d:osemoments}
    A distribution $\mathcal{D}$ over $\mathbb{R}^{m \times n}$ has
$(\varepsilon,\delta,d,\ell)$-OSE moments if, for all matrices
$U \in \mathbb{R}^{n \times d}$ with orthonormal columns,
\begin{align*}
  \E_{\Pi \sim \mathcal{D}}
  \bigl\|(\Pi U)^{T}(\Pi U) - I\bigr\|^{\ell}
  < \varepsilon^{\ell}\delta.
\end{align*}
\end{definition}
Note that a simple application of Markov's inequality recovers the guarantee in Definition~\ref{def:OSE} from the $(\varepsilon,\delta,d,\ell)$-OSE moments property with any $\ell\geq 1$. Moreover, \cite{cohen2016optimal} showed that this moment-based OSE characterization implies several other desirable guarantees of embedding matrices in the context of approximate matrix multiplication, generalized regression and low-rank approximation.

As an immediate consequence of our analysis, we obtain the following OSE moment guarantee for the OSNAP distribution.

\begin{corollary}\label{cor:osemoments}
    Let $\Pi$ be an $m \times n$ matrix with an OSNAP distribution having sparsity~$s$. Let $0 < \delta, \varepsilon < 1$ and $d>10$. Then $\Pi$ has $(\varepsilon, \delta, d, \ell)$-OSE moments with $\ell = 16\log(\frac{d}{\varepsilon \delta})$ when $m \ge c_{\ref*{cor:osemoments}.1}  (d + \log(1/\delta\varepsilon))/\varepsilon^2$ and  $s \geq \min\{c_{\ref*{cor:osemoments}.2} (\log ^2(\frac{d}{\varepsilon\delta})/\varepsilon+\log^3(\frac{d}{\varepsilon\delta})),m\}$.
\end{corollary}
\begin{remark}
$\Pi$ can be applied to a matrix $A$ in time $O(\nnz(A)(\log^2(\frac{d}{\varepsilon\delta})/\varepsilon + \log^3(\frac{d}{\varepsilon\delta})))$.  
As noted by \cite[Remark 3]{cohen2016optimal}, such runtimes can be further refined by chaining together several embeddings with an OSE moment property. For example, \cite{cohen2016nearly} showed that  OSNAP with $m=O(d\log(d/\delta)/\varepsilon^2)$ and $s=O(\log(d/\delta)/\varepsilon)$ has $(\varepsilon,\delta,d,\log(d/\delta))$-OSE moments. Thus, letting $\varepsilon=\Theta(1)$ for simplicity, we can combine a $O(d\log (d/\delta))\times n$ OSNAP matrix $\Pi_1$ having sparsity $s=O(\log(d/\delta))$ together with a $O(d+\log(1/\delta))\times O(d\log(d/\delta))$ OSNAP matrix having sparsity $s=O(\log^3(d/\delta))$ to obtain $\Pi=\Pi_2\Pi_1$ with $(\Theta(1),\delta,d,\log(d/\delta))$-OSE moments which can be applied to a matrix $A\in\R^{n\times d}$ in time $O(\nnz(A)\log(d/\delta) + d^2\log^4(d/\delta))$.
\end{remark}

\subsection{Leverage Score Sparsified Embedding with Independent Columns}
In a related problem, we seek to embed a subspace given by a fixed $U \in \R^{n \times d}$, with information about the squared row norms of $U$ being used to define the distribution of non-zero entries in $\Pi$. Such distributions for $\Pi$ are called non-oblivious (a.k.a.~data-aware) subspace embeddings. Previous work \cite{chenakkod2024optimal} has dealt with one such family of distributions termed LESS embeddings \cite{less-embeddings,newton-less,gaussianization}, showing that they require $\tilde O(1/\varepsilon^4)$ non-zero entries \textit{per row} of $\Pi$ to obtain an $\varepsilon$-embedding guarantee. Since the embedding matrix is very wide, this leads to a much sparser embedding (sparser than any OSE) that can be applied in time sublinear in the input size, leading to fast subspace embedding algorithms. 

In this work, we show that our new techniques  also extend to LESS embeddings and enable us to prove sharper sparsity estimates than~\cite{chenakkod2024optimal}. To fully leverage our approach, we define a new type of sparse embedding (LESS-IC), which can be viewed as a cross between CountSketch and LESS. Here, IC stands for independent columns. At a high level, the CountSketch part ensures that we can use our decoupling method to achieve optimal dependence on $1/\varepsilon$, while the LESS part enables adaptivity to a fixed subspace.

Specifically, a LESS-IC embedding matrix $\Pi$ has a fixed number of non-zero entries in each column, chosen so that it is proportional to the leverage score (i.e. the squared row norm) of the corresponding row of $U$. This is achieved by modifying the OSNAP distribution such that the number of subcolumns is no longer the same in each column. For columns corresponding to very small leverage scores, we only have one ``subcolumn''. Thus, each column has at least one non-zero~entry. This means that the cost of applying LESS-IC to an $n\times d$ matrix $A$ can no longer be sublinear (like it can in the existing LESS embedding constructions), but rather has a fixed linear term of $O(\nnz(A))$, plus an additional sublinear term. Given that the preprocessing step of approximating the leverage scores has to take at least $\nnz(A)$ time, the linear term in the cost of applying LESS-IC is negligible.

To generate an embedding matrix with the LESS-IC distribution, it suffices to have a good enough approximation for the leverage scores of the matrix $U$, in the following sense.

\begin{definition}[Approximate Leverage Scores]\label{def:apprls}
    Given a matrix $U\in\R^{n\times d}$ with orthonormal columns and $\beta_1 \ge 1, \beta_2 \ge 1$, a tuple $(l_1, \ldots, l_n) \in [0,1]^n$ of numbers are $(\beta_1,\beta_2)$-approximate leverage scores for $U$ if, for $1\leq i\leq n$,
    \begin{align*}
        \frac{\norm{e_i^\top U}^2}{\beta_1} \leq l_i \qquad\text{and}\qquad
    \sum_{i=1}^n l_i \leq \beta_2 \sum_{i=1}^n \norm{e_i^\top U}^2 = \beta_2 d.
    \end{align*}
    We say that the numbers $(l_1, \ldots, l_n) \in [0,1]^n$ are $\beta$-approximations of the leverage scores (i.e. squared row norms) of $U$ with $\beta=\beta_1\beta_2$.
\end{definition}

To see how approximate leverage scores determine the distribution of entries in the LESS-IC distribution, let us first consider a simpler distribution, LESS-IE from \cite{chenakkod2024optimal}, based on a similar construction first proposed by \cite{less-embeddings}. Here, we once again start by defining an unscaled matrix $S$, which is then normalized to obtain the subspace embedding matrix $\Pi$.

\begin{definition}[LESS-IE]\label{def:lessindent}
An $m \times n$ random matrix $S$ is called an unscaled leverage score sparsified embedding with independent entries (unscaled LESS-IE), and also $\Pi = (1/\sqrt{pm})S$ is called a LESS-IE, corresponding to $(\beta_1, \beta_2)$-approximate leverage scores $(z_1,...,z_n)$ with parameter $p$, if $S$ has entries $s_{i,j}=\frac{1}{\sqrt{\beta_1  z_j}} \delta_{i,j} \xi_{i,j}$ where $\delta_{i,j}$ are independent Bernoulli random variables taking value 1 with probability $p_{ij}= \beta_1 z_j p$, whereas $\xi_{i,j}$ are i.i.d.~Rademacher random~variables. 
\end{definition}

In the LESS-IE model, we have $\beta_1 pm z_j $ many non-zero entries in  column $j$ in expectation. However, to achieve $1/\varepsilon$ dependency of the sparsity, we need to have \emph{exactly} $\beta_1 pm z_j $ many non-zero entries in the column in the LESS-IC model to fully take advantage of the error cancellation that occurs in our decoupling argument (See Section \ref{subsec:decouposnap} and Section \ref{subsec:oseiediag}). Though these sections deal with oblivious subspace embeddings, the same arguments still apply in the LESS case). This is done by modifying the OSNAP construction so that the size (and consequently, the number) of subcolumns is different across~columns.

Notice that to have $\beta_1 pm z_j$ many non-zero entries in column $j$, we would need $\beta_1 pm z_j$ many subcolumns in column $j$ each with one non-zero entry in a random position. This means that the size of each subcolumn needs to be $m/(\beta_1 pm z_j) = 1/(\beta_1pz_j)$. However, since $1/(\beta_1pz_j)$ may not be an integer, we consider subcolumns of size $ b_j := \max \{ \lfloor 1/(\beta_1pz_j)\rfloor, 1 \}$. 

In column $j$, we stack subcolumns of size $b_j$ until we fill up all the rows up to $m$. Let $s_j$ be the smallest number of subcolumns to do this. Then, it may happen that the row indices of the bottom-most subcolumn exceed $m$. For example, consider the distribution on the first column of $\Pi$ when $m=70$, and $b_1 = 15$. In this case $s_1 = 5$, so we can stack four subcolumns of size 15 and the $5\textsuperscript{th}$ subcolumn only spans row indices $[61:70]$. In each subcolumn, we randomly choose a row to place a non-zero entry, which would be a Rademacher random variable. (See Figure \ref{fig:lessic}). The non-zero entries are appropriately scaled so that all entries of the matrix have the same variance (See Section \ref{sec:lessproofs} for the full definition).

\begin{figure}[h!]
    \centering
    \begin{tikzpicture}[scale=1]
      \pgfmathsetmacro{\a}{1.4};
      \pgfmathsetmacro{\b}{2.8};
      \pgfmathsetmacro{\c}{3.2};
      \pgfmathsetmacro{\d}{4};
      \pgfmathsetmacro{\aa}{0.2};
      \pgfmathsetmacro{\aaa}{0};
      \pgfmathsetmacro{\bbb}{0.8};
      \pgfmathsetmacro{\ccc}{1.6};
      \pgfmathsetmacro{\ddd}{2.6};
      \pgfmathsetmacro{\eee}{3.8};
      \pgfmathsetmacro{\fff}{4.8};
      \draw (0,-0.8) rectangle (5,2);
      \draw (2.5,2.25) node {\mbox{\scriptsize Embedding matrix $\Pi$}};

  \draw (5.5,1) node {\mbox{$\times$}};
  \draw (7,2.25) node {\mbox{\scriptsize Matrix $U$}};  
  \draw[fill=red!20] (6,-3) rectangle (8,2);

    \draw (\aaa,-0.8) rectangle (\aaa + 0.2,-0.4); 
    \draw (\aaa+0.1,-0.5) node {$*$};
  \draw (\aaa,-0.4) rectangle (\aaa + 0.2,0.2); 
  \draw (\aaa+0.1,-0.3) node {$*$};
  \draw (\aaa,0.2) rectangle (\aaa + 0.2,0.8);
  \draw (\aaa+0.1,0.7) node {$*$};
  \draw (\aaa,0.8) rectangle (\aaa + 0.2,1.4);
  \draw (\aaa+0.1,1.3) node {$*$};
  \draw (\aaa,1.4) rectangle (\aaa + 0.2,2);
  \draw (\aaa+0.1,1.7) node {$*$};
    \draw[fill=blue!30] (6, 1.8 - \aaa) rectangle (8, 2 - \aaa);

    \draw (\bbb,-0.8) rectangle (\bbb + 0.2,-0.4); 
    \draw (\bbb+0.1,-0.7) node {$*$};
    \draw (\bbb,-0.4) rectangle (\bbb + 0.2,0.4);
    \draw (\bbb+0.1,0.3) node {$*$};
    \draw (\bbb,0.4) rectangle (\bbb + 0.2,1.2);
    \draw (\bbb+0.1,0.7) node {$*$};
    \draw (\bbb,1.2) rectangle (\bbb + 0.2,2);
    \draw (\bbb+0.1,1.9) node {$*$};
    \draw[fill=blue!30] (6, 1.8 - \bbb) rectangle (8, 2 - \bbb);

    \draw (\ccc,-0.8) rectangle (\ccc + 0.2,0); 
    \draw (\ccc+0.1,-0.3) node {$*$};
    \draw (\ccc,0) rectangle (\ccc + 0.2,1);
    \draw (\ccc+0.1,0.5) node {$*$};
    \draw (\ccc,1) rectangle (\ccc + 0.2,2);
    \draw (\ccc+0.1,1.9) node {$*$};
    \draw[fill=blue!30] (6, 1.8 - \ccc) rectangle (8, 2 - \ccc);

    \draw (\ddd,-0.8) rectangle (\ddd + 0.2,0.6);
    \draw (\ddd+0.1,-0.1) node {$*$};
    \draw (\ddd,0.6) rectangle (\ddd + 0.2,2);
    \draw(\ddd + 0.1 ,1.7) node {$*$};
    \draw[fill=blue!30] (6, 1.8 - \ddd) rectangle (8, 2 - \ddd);

    \draw (\eee,-0.8) rectangle (\eee + 0.2,-0.2); 
    \draw (\eee+0.1,-0.3) node {$*$};
    \draw (\eee,-0.2) rectangle (\eee + 0.2,2);
    \draw (\eee+ 0.1,-0.1) node {$*$};
    \draw[fill=blue!30] (6, 1.8 - \eee) rectangle (8, 2 - \eee);

    \draw (\fff,-0.8) rectangle (\fff + 0.2,2); 
    \draw (\fff + 0.1, 1.1) node {$*$};
    \draw[fill=blue!30] (6, 1.8 - \fff) rectangle (8, 2 - \fff);
    
  \draw (8.95,2.25) node {\mbox{\fontsize{7}{7}\selectfont leverage scores}};
  \foreach \i in {0,...,24}
  {
    \pgfmathtruncatemacro{\x}{\i^2};
    \draw[fill=darkgreen!30] (8.1, 0.2 * \i - 3)
    rectangle (8.1 + 0.0025 * \x, 0.2 * \i - 2.8);
  }
\end{tikzpicture}
    \caption{In the LESS-IC distribution, column $j$ is filled with $s_j$ many subcolumns, with the bottom-most subcolumn truncated to fit the size of $\Pi$. Each subcolumn has one non-zero entry. Notice that as the leverage scores decrease, the number of subcolumns decreases and the matrix becomes sparser. However, each column always has at least one non-zero entry.}
    \label{fig:lessic}
\end{figure}

For the LESS-IC distribution, we show the following subspace embedding guarantee. The structure of the proof is similar to the case of OSNAP, and only the specific expressions change due to the different distribution.

\begin{restatable}[Subspace Embedding Guarantee for LESS-IC]{theorem}{lessicmainthm}
\label{t:less-ic}
    Let $\Pi = (1/\sqrt{pm}) S$ be an $m \times n$ matrix distributed according to the $8 \lceil\log (\frac{d}{\varepsilon \delta})\rceil$-wise independent  LESS-IC distribution with parameter $p$ for some fixed $n\times d$ matrix $U$ satisfying $U^\top U=I$ with given $(\beta_1, \beta_2)$-approximate leverage scores.
Then, there exist positive constants $c_{\ref*{t:less-ic}.1}$ and $c_{\ref*{t:less-ic}.2}$ such that for any $0 < \varepsilon, \delta < 1$, and $d>10$, we have 
\begin{align*}
\Pb \left( 1 - \varepsilon  \leq s_{\min}(\Pi U)   \leq s_{\max}(\Pi U) \leq 1 + \varepsilon \right) \geq 1-\delta
\end{align*}
when $m \ge c_{\ref*{t:less-ic}.1}  \left( \frac{d +  \log^2(d/\delta)+\log(1/\varepsilon)}{\varepsilon^2} + \log^3(d/\delta)/\varepsilon \right) $ and 
$$ c_{\ref*{t:less-ic}.2} \max \left \{  \frac{(\log (d/{\varepsilon\delta}))^{2.5}}{\varepsilon}, (\log (d/{\varepsilon\delta}))^{3} \right \} \le pm \le m.$$ The matrix $\Pi$ has $O(n + \beta pmd)$ many non-zero entries and can be applied to an $n\times d$ matrix~$A$ in $O(\nnz(A) + \beta pmd^2)$ time, where $\beta=\beta_1\beta_2$ is the leverage score approximation factor.
\end{restatable}

\begin{remark}
    When $\delta = d^{-O(1)}$, we recover the optimal dimension $m=\Theta(d/\varepsilon^2)$ while showing that one can apply the LESS-IC embedding in time $O(\nnz(A)) + \tilde O(\beta d^2/\varepsilon)$. In comparison, \cite{chenakkod2024optimal} showed that a corresponding LESS-IE embedding can be applied in $\tilde O(\beta d^2/\varepsilon^6)$ time. Using our techniques, one could improve the runtime of LESS-IE to $\tilde O(\beta d^2/\varepsilon^2)$, but our new LESS-IC construction appears necessary to recover the best dependence on $1/\varepsilon$.
\end{remark}

\subsection{Fast Subspace Embedding (Proof of Theorem \ref{t:fast})}

Here, we briefly outline how our LESS-IC embedding yields a fast subspace embedding construction to recover the time complexity claimed in Theorem \ref{t:fast}. This follows analogously to the construction from Theorem~1.6 of \cite{chenakkod2024optimal}, and our improvement in the dependence on $1/\varepsilon$ compared to their result (from $1/\varepsilon^6$ to $1/\varepsilon$) stems from the improved sparsity of our LESS-IC embedding. 

The key preprocessing step for applying the LESS-IC embedding is approximating the leverage scores of the matrix $A$. Using Lemma 5.1 in \cite{chenakkod2024optimal} (adapted from Lemma 7.2 in~\cite{chepurko2022near}), we can construct coarse approximations of all leverage scores so that $\beta_1=O(n^\gamma)$ and $\beta_2=O(1)$ in time $O(\gamma^{-1}(\nnz(A)+d^2) + d^\omega)$. Applying LESS-IC (Theorem \ref{t:less-ic}) with these leverage scores and parameters $\beta_1,\beta_2$, computing $\Pi A$ takes $O(\nnz(A) + n^{\gamma} d^2\log^{3}(d/\varepsilon\delta)/\varepsilon)$, where $\nnz(A)$ comes from the fact that every column  of $\Pi$ has at least one non-zero, while the second term accounts for the additional $O(\beta d\log^{3}(d/\varepsilon\delta)/\varepsilon)$ non-zeros. 

Thus, if $d\geq n^c$ for, say, $c=0.1$, then we conclude the claim by appropriately scaling $\gamma$ by a constant factor. Now, suppose otherwise. First, note that without loss of generality we can assume that $\gamma<0.1$ (through scaling the time complexity by a constant factor), $\nnz(A)\geq n$ (by removing empty rows) and $\varepsilon\geq \sqrt{d/n}$ (because otherwise $m\geq n$ and we could use $\tilde A= A$). Thus, under our assumption that $d<n^c$, we have $n^{\gamma} d^2/\varepsilon \leq n^{0.5+\gamma+2c}\leq n^{0.8} \ll \nnz(A)$, and the time complexity is dominated by the $O(\gamma^{-1}\nnz(A))$ term.

Finally, we note that Corollary \ref{t:reduction} follows simply by constructing a subspace embedding $\Pi$ via Theorem \ref{t:fast} with respect to matrix $[A\mid b]$, and computing $\tilde A=\Pi A$, $\tilde b = \Pi b$. The proof of the claim is identical to the proof of Theorem 1.8 in \cite{chenakkod2024optimal}. Our improvement comes directly from the faster runtime of our subspace embedding construction.

\subsection{Outline of the Paper}

Section \ref{s:overview} provides a high level overview of the ideas used in the proofs of our main results, Theorem \ref{t:ose-full} and Theorem \ref{t:less-ic}. Section \ref{sec:proofsketch} provides a sketch of the proof of Theorem \ref{t:ose-full}, listing the main technical steps, leaving the full proof with all technical details to Section \ref{sec:osnap-proof}. The proof of Theorem \ref{t:less-ic} follows similarly and is covered in Section \ref{sec:lessproofs}. The subspace embedding guarantee for a sparse matrix with independent entries is proved in Section \ref{sec:oseieproof}. Section \ref{sec:prelim} contains some  basic facts from the existing literature that are used throughout the paper.

\subsection{Notation}
The following notation and terminology will be used in the paper. The notation $[n]$ is used for the set $\{1,2,...,n\}$ and the notation $\operatorname{P}([n])$ denotes the set of all partitions of $[n]$. Also, for two integers $a$ and $b$ with $a \le b$, we use the notation $[a:b]$ for the set $\{k \in \Z:a \le k \le b\}$. For $x \in \R$, we use the notation $\lfloor x \rfloor$ to denote the greatest integer less than or equal to $x$ and $\lceil x \rceil$ to denote the least integer greater than or equal to  $x$. In $\R^n$ (or $\R^m$ or $\R^d$), the $l$th coordinate vector is denoted by $e_l$. All matrices considered in this paper are real valued and the space of $m \times n$ matrices with real valued entries is denoted by $M_{m \times n}(\mathbb{R})$. Also, for a matrix $X \in M_{d \times d}(\mathbb{R})$, the notation $\Tr (X)$ denotes the trace of the matrix $X$, and $\tr (X) = \frac{1}{d} \Tr (X)$ denotes the normalized trace. We write the operator norm of a matrix $X$ as $\norm{X}$, and it is also denoted by $\norm{X}_{op}$ in some places where other norms appear for clarity. The spectrum of a matrix $X$ is denoted by $\spec(X)$.  The standard probability measure is denoted by $\mathbb{P}$, and the symbol $\mathbb{E}$ means taking the expectation with respect to this standard probability measure. To simplify the notation, we follow the convention from \cite{brailovskaya2022universality} and use the notation $\E [X]^{\alpha}$ for $(\E(X))^{\alpha}$, i.e., when a functional is followed by square brackets, it is applied before any other operations. The covariance of two random variables $X$ and $Y$ is denoted by $\cov(X,Y)$. The standard $L_q$ norm of a random variable $\xi$ is denoted by $\norm{\xi}_q$, for $1 \le q \le \infty$. Throughout the paper, the symbols $c_1, c_2, ...$, and $Const, Const', ...$ denote absolute constants.

\section{Main Ideas}
\label{s:overview}

We next outline our new techniques which are needed to establish the main results, Theorems \ref{t:ose-full} and \ref{t:less-ic}. Here, for notational convenience, we will refer to the unscaled random matrix $S$, as opposed to the subspace embedding matrix $\Pi=(1/\sqrt{pm})S$ (see Definition \ref{def:osnap}).

Note that due to the equivalent characterization of the OSE property in \eqref{eq:ose-equiv}, all we need to show is that singular values of $SU$ are clustered around $\sqrt{pm}$ at distance $O(\sqrt{pm} \varepsilon)$. In other words, we need to show that the difference between the spectrum of $SU$ and the spectrum of $\sqrt{pm}I_{d}$ is small, of the order $O(\sqrt{pm} \varepsilon)$. 

In all our models, the entries of $S$ are uncorrelated with mean $0$ and variance $p$, and therefore the entries of $SU$ are uncorrelated with uniform variance. If we consider a random matrix $G$ with Gaussian entries which keeps the covariance profile of the entries of $SU$, then this Gaussian random matrix $G$ has independent Gaussian entries with variance $p$. Using classical results about singular values of Gaussian random matrices, it can be shown that the singular values of $G$ are sufficiently clustered around $\sqrt{pm}$ with high probability for $m = \Omega(d/\varepsilon^2)$.  Thus, it suffices to find conditions under which the singular values of $SU$ are sufficiently close to the singular values of $G$. This is the phenomenon of universality whereby random systems show predictable (in this case Gaussian) behavior under certain limits.

\textbf{Failure of black-box universality.}
Recent work by Brailovskaya-van Handel \cite{brailovskaya2022universality} on universality for certain random matrix models developed tools to bound the distance between the spectrum of a random matrix model obtained as a sum of independent random matrices and the spectrum of a Gaussian random matrix with the same covariance profile. Using these tools, \cite{chenakkod2024optimal} achieved optimal embedding dimension $m=O(d/\varepsilon^2)$ for OSEs by using the bound in \cite[Theorem 2.6]{brailovskaya2022universality} to estimate the Hausdorff distance (a concept of distance between two subsets of $\R$; $A,B \subset \R$ are said to be $\varepsilon$-close in Hausdorff distance if $A$ is in the $\varepsilon$-neighborhood of $B$ and $B$ is in the $\varepsilon$ neighborhood of $A$) between the spectra of 
\begin{align*}
    \sym(SU)=\left[ {\begin{array}{*{20}{c}}
  {}&{{(SU)^T}} \\ 
  SU&{} 
\end{array}} \right] \quad \text{ and } \quad \sym(G)=\left[ {\begin{array}{*{20}{c}}
  {}&{{G^T}} \\ 
  G&{} 
\end{array}} \right].
\end{align*}

This distance is shown to be $(O(\sqrt{pm}))^{2/3}$, which is of order $\sqrt{pm}\varepsilon$ only when $pm$ has $1/\varepsilon^6$ dependence. Thus, \cite{chenakkod2024optimal} did not obtain the conjectured dependency of the sparsity on $\varepsilon$, which requires $pm$ to only have $1/\varepsilon$ dependency. To get better $\varepsilon$ dependency, we would either need a sharper bound on the Hausdorff distance, or have the distance decrease with $\varepsilon$. For example, if the $(O(\sqrt{pm}))^{2/3}$ bound was improved to $(O(\sqrt{pm}))^{1/2}$, we would only need $(\sqrt{pm})^{1/2} \le \sqrt{pm} \varepsilon$ which can be achieved when $pm$ has $1/\varepsilon^4 $ dependence. On the other hand, if the $(O(\sqrt{pm}))^{2/3}$ bound was improved to $(O(\sqrt{pm}))^{2/3} \varepsilon^{1/2}$, we would only need $pm$ to have $1/\varepsilon^3$ dependence.

\textbf{Key idea: Universality of centered moments.}
One can instead look at a different approach to characterize the clustering of singular values. To show that the singular values of $\Pi U$ are between $1 \pm \varepsilon$, it is enough to show that $\norm{(\Pi U)^T\Pi U - I_d} \le \varepsilon$ or $\norm{(S U)^TSU - pm\cdot I_d} \le pm \varepsilon$ (Note that $S=\sqrt{pm} \Pi$). One way to achieve this bound with high probability is to use the moment method, i.e., to show that (see proof of Theorem \ref{t:ose-full} in Section \ref{sec:proofsketch}):
$$\E\Big[\tr \big((SU)^T(SU) - pm \cdot I_{d}\big)^{2q}\Big]^{\frac{1}{2q}}=O(pm \varepsilon).$$  

In this case, standard calculations on Gaussian random matrices(see Lemma \ref{cor:gaussianmom}) show that $(\E[\tr (G^TG - pmI_{d \times d})^{2q}])^{\frac{1}{2q}} \le cpm\sqrt{\frac{d}{m}}=O(pm \varepsilon)$ when $m = \Omega(d/\varepsilon^2)$ and $G$ has the covariance profile of $SU$. So it is enough to show that 
$$\E\Big[\tr \big((SU)^T(SU) - pmI_{d}\big)^{2q}\Big]^{\frac{1}{2q}}-\E\Big[\tr \big(G^TG - pmI_{d}\big)^{2q}\Big]^{\frac{1}{2q}}=O(pm\varepsilon).$$
where we recall the notation $\E\Big[\tr \big(X\big)^{2q}\Big]^{\frac{1}{2q}}=\left (\E\tr\left ( X \right )^{2q}\right )^{1/(2q)}
$.

Now, \cite[Proposition 9.12]{brailovskaya2022universality} does take a similar approach of comparing $(SU)^T(SU) - pmI_{d}$ and $G^TG - pmI_{d}$, by relying on an interpolation argument, where one defines a mixture $S(t) = \sqrt{t} S+\sqrt{1-t}G$ and controls the change in the moments along the trajectory specified by $t\in[0,1]$.  Unfortunately, using that result gives a larger power of $pm$ in the bound than desired, resulting again in a worse $\varepsilon$ dependence.

One can also, by viewing $(SU)^T(SU) - pmI_{d} =\sum_{i=1}^m (U^Ts_is_i^TU-pI_{d})$, get a random matrix model which is a sum of independent random matrices (this is not true for OSNAP, but some other models of OSEs), and then compare $\E[\tr ((SU)^T(SU) - pm \cdot I_{d})^{2q}]^{\frac{1}{2q}}$ with $\E[\tr (H)^{2q}]^{\frac{1}{2q}}$ where $H$ is the Gaussian model for $(SU)^T(SU) - pmI_{d}$. This is the approach of  \cite[Proposition 9.15]{brailovskaya2022universality}, but it fails in obtaining the optimal embedding dimension $m = d/\varepsilon^2$.

\textbf{Key technique: Decoupling.} To overcome these obstacles, we develop a fresh analysis while still using the ideas of \cite{brailovskaya2022universality}. Our first step is to observe that due to the property of $S$ having a fixed number of non-zero entries in a column for the OSNAP distribution, all quadratic terms in $(SU)^T(SU) - pm\cdot I_d$ are square-free, and this allows us to use the decoupling technique to reduce the problem of controlling the moments of $(SU)^T(SU) - pm\cdot I_d$ to controlling the moments of $(S_1U)^T(S_2U) + (S_2U)^T(S_1U)$ where $S_1$ and $S_2$ are independent copies of $S$ (See proof of Lemma \ref{prop:momestdecoupmatrix} in Section \ref{sec:proofsketch}).

We still have to separate bounding $\E[\tr((S_1U)^T(S_2U)+ (S_2U)^T(S_1U) )^{2q}]^\frac{1}{2q}$ into two parts, bounding $\E[\tr(G_1^TG_2+G_2^TG_1 )^{2q}]^\frac{1}{2q}$ for the Gaussian model, and the difference 
\begin{align*}\E[\tr((S_1U)^T(S_2U)+ (S_2U)^T(S_1U))^{2q}]^\frac{1}{2q}-\E[\tr(G_1^TG_2+G_2^TG_1 )^{2q}]^\frac{1}{2q},\end{align*}
which is called the universality error.

By standard calculations, we have $\E[\tr(G_1^TG_2+G_2^TG_1)^{2q}]^\frac{1}{2q} \le c \sqrt{pm}\sqrt{pd} = O(pm \varepsilon)$ and the main task is still to bound the universality error. The advantage of the decoupling idea is that, informally speaking, since $S_1$ and $S_2$ are independent, we can condition on one of them, e.g., $S_1$. For fixed $S_1$, the random matrix $(S_1U)^T(S_2U)$ (where all randomness comes from $S_2$) can be viewed as a sum of independent random matrices, with the individual summands having moments of smaller order than the previous approach. We can then use an interpolation argument to bound the trace universality error for $q=\log(\frac{d}{\varepsilon \delta})$ as follows:
\begin{equation}\label{eq:decouplepolylogunivererror}
     \begin{aligned}
         \Big|\E\!\big[\tr((S_1U)^T(S_2U)+ (S_2U)^T(S_1U))^{2q}\big]^\frac{1}{2q}-\E\!\big[\tr(G_1^TG_2+G_2^TG_1 )^{2q}\big]^\frac{1}{2q}\Big|\le \polylog(\tfrac{d}{\varepsilon \delta}).\hspace{-3mm}
     \end{aligned}
 \end{equation}
 Notice that there is no $pm$ dependence on the right hand side. So our requirement that this quantity be bounded by $pm\varepsilon$ is satisfied when $pm \ge \polylog(\frac{d}{\varepsilon \delta})/\varepsilon$, achieving the conjectured $1/\varepsilon$ dependence.

Nevertheless, the conditioning argument cannot be done directly because
\begin{align*}
\hspace{-1mm}\E\!\Big[\tr\big((S_1U)^T(S_2U)+ (S_2U)^T(S_1U)\big)^{2q}\Big]^\frac{1}{2q}
\ne \E_{S_1}\!\!\bigg[\E_{S_2}\!\!\Big[\tr\big((S_1U)^T(S_2U)+ (S_2U)^T(S_1U)\big)^{2q}\Big]^\frac{1}{2q}\bigg].
\end{align*}

\textbf{Key technique: 2D interpolation via chain rule.}
So, instead we develop a new approach which incorporates the conditioning step directly into a two-dimensional interpolation argument, through the use of the chain rule (see Figure \ref{fig:two-parameter-interpolation}). Define
\[
S_1(t_1) = \sqrt{t_1}\,S_1 + \sqrt{1-t_1}\,G_1, 
\quad
S_2(t_2) = \sqrt{t_2}\,S_2 + \sqrt{1-t_2}\,G_2.
\]
We start from  $(G_1, G_2)$ at $(t_1,t_2) = (0,0)$ and move to $(S_1, S_2)$ at $(1,1)$, interpolating between the easier-to-analyze Gaussian matrices $(G_1,G_2)$ and the true random matrices $(S_1,S_2)$ of interest and controlling the changes in their moments (or the error terms) step by step.

\begin{figure}[]\label{fig:2dinterp}
		\centering
		\begin{tikzpicture}[>=stealth, scale=3.6, font=\small]  
			
			\draw[->] (-0.1,0) -- (1.2,0) node[right] {$t_1$};
			\draw[->] (0,-0.1) -- (0,1.2) node[above] {$t_2$};
			
			\fill[green!10] (0,0) rectangle (1,1);
			
			\draw (0,0) rectangle (1,1);
			
			\node[below left]  at (0,0)  {$(G_1,\; G_2)$};
			\node[below right] at (1,0)  {$(S_1,\; G_2)$};
			\node[above left]  at (0,1)  {$(G_1,\; S_2)$};
			\node[above right] at (1,1)  {$(S_1,\; S_2)$};
			
			\draw[dashed, thick] (0,0) -- (1,1);
			\node[above left, xshift=6pt, yshift=-6pt] 
			at (0.7,0.8) {\scriptsize diagonal $t_1 = t_2$};
			
			\coordinate (smallPoint) at (0.3,0.3);
			\fill (smallPoint) circle (0.02);
			\draw[->] (smallPoint) -- ++(0.09,0) node[right, xshift=-4pt] {\scriptsize $\frac{\partial}{\partial t_1}$};
			\draw[->] (smallPoint) -- ++(0,0.09) node[above, yshift=-4pt] {\scriptsize $\frac{\partial}{\partial t_2}$};
			
		\end{tikzpicture}
		\caption{
			Two-dimensional interpolation in $(t_1,t_2)\in[0,1]^2$, decomposed using the chain rule.}
		\label{fig:two-parameter-interpolation}
	\end{figure}

Defining $f(M_1,M_2)=\tr((M_1U)^T(M_2U)+(M_2U)^T(M_1U))^{2q}$, and applying the chain rule on the diagonal $t_1=t_2=t$, we obtain:
\begin{align*}
\frac{d}{dt} \mathbb{E}\bigl[f\bigl(S_1(t),S_2(t)\bigr)\bigr]
&=\frac{\partial}{\partial t_1}\E\bigl[ f\bigl(S_1(t_1),S_2(t_2)\bigr)\bigr]\Big|_{t_1=t,\,t_2=t}
+\frac{\partial}{\partial t_2}\E\bigl[ f\bigl(S_1(t_1),S_2(t_2)\bigr)\bigr]\Big|_{t_1=t,\,t_2=t}.
\end{align*}
By independence of $S_1$ and $S_2$, we can condition on $S_2$ when we bound the partial derivative $\frac{\partial}{\partial t_1}\E\bigl[f\bigl(S_1(t_1),S_2(t_2)\bigr)\bigr]\big|_{t_1, t_2=t}$, and do similar calculations for the other term. The benefit of doing this is that we can now fine tune the techniques of \cite{brailovskaya2022universality} to get a differential inequality (Lemma \ref{lem:diffineq}) that leads to inequality (\ref{eq:decouplepolylogunivererror}). In doing so, we are able to find the optimal bounds and exponents in the differential inequality.  

\section{Proof Sketch for the Oblivious Subspace Embedding} \label{sec:proofsketch}

We now sketch the proof of our main subspace embedding guarantee, Theorem \ref{t:ose-full} for OSNAP. The full proof can be found in Section \ref{sec:osnap-proof}. The proof of the subspace embedding guarantee for LESS-IC, Theorem \ref{t:less-ic} is similar and can be found in Section \ref{sec:lessproofs}.

\begin{proof}[Proof sketch of Theorem \ref{t:ose-full}]
    Let $X:= \frac{1}{\sqrt{pm}}SU$. We first assume that the collection of all the random variables 
$\{ \xi_{(l,\gamma)}, \mu_{(l,\gamma)} \}_{l \in [n], \gamma \in [s]}$ 
    in the unscaled OSNAP construction are fully independent, and later we will check what is the minimum independence needed.
    
    We observe that to prove the theorem, it is enough to show that
    \begin{align*} \Pb \left( \| X^TX - I_d \| \le \varepsilon \right) \ge 1-\delta. \end{align*}
We call the quantity $X^TX - I_d$ the embedding error. By Markov's inequality, we have
\begin{align*}
        \Pb \left( \| X^TX - I_d \| \ge \delta^{-\frac{1}{2q}} \E [d \tr (X^TX - I_d)^{2q}]^\frac{1}{2q} \right) &\le \delta
    \end{align*}
which, after simplification, becomes
\begin{align*}
        \Pb \left( \| X^TX - I_d \| \ge (d/\delta)^{\frac{1}{2q}} \E [ \tr (X^TX - I_d)^{2q}]^\frac{1}{2q} \right) &\le \delta.
    \end{align*}    
For $q> \log (d/\delta) $, we have
\begin{align*}
    (d/\delta)^{\frac{1}{2q}}=\exp\Big(\log(d/\delta)\frac{1}{2q}\Big) \le \exp\Big(\log(d/\delta)\frac{1}{2\log (\frac{d}{\delta})}\Big) \le {\sqrt{e}}.
\end{align*}
Therefore, we have
\begin{align*}
        \Pb \left( \| X^TX - I_d \| \ge \sqrt{e} \E [ \tr (X^TX - I_d)^{2q}]^\frac{1}{2q} \right) &\le \delta.
    \end{align*}

Thus we need to control moments of order $2q$ of the embedding error for $q>\log(d/\delta)$, and this is done in the following lemma.

\begin{osnaptracemom}[Trace Moments of Embedding Error for OSNAP]
    For $X$ as above, there exist constants $c_{\ref*{prop:momestdecoupmatrix}.1}, c_{\ref*{prop:momestdecoupmatrix}.2}, c_{\ref*{prop:momestdecoupmatrix}.3} > 0$ such that for $q \in \N$ satisfying $2 \le q \le m$, we have
\begin{align}
    \E[\tr(X^TX - I_d)^{2q}]^\frac{1}{2q} &\leq  \varepsilon,\nonumber\\
\text{when }\quad    m &\geq c_{\ref*{prop:momestdecoupmatrix}.1} \frac{d+q}{\varepsilon^2}  \label{ose-cond1}\\
    \text{and } \quad pm &\ge \bigg(\max\Big\{\frac{c_{\ref*{prop:momestdecoupmatrix}.2} q^{2}}{\varepsilon}, c_{\ref*{prop:momestdecoupmatrix}.3} q^3\Big\}\bigg)^{1+\frac{2}{q-2}}. \label{ose-cond2}    
\end{align} 
\end{osnaptracemom}

Applying Lemma \ref{prop:momestdecoupmatrix} (with appropriately adjusted $\varepsilon$) implies
\begin{align*}
        \Pb \left( \| X^TX - I_d \| \ge \varepsilon \right) &\le \delta
    \end{align*}
when combined with the previous calculations.

It remains to check that conditions \eqref{ose-cond1} and \eqref{ose-cond2} are satisfied for $q = \lceil 2\log (\frac{d}{\varepsilon \delta} )\rceil+2$ by requiring $m \ge c_{\ref*{t:ose-full}.1}  (d + \log(1/\delta\varepsilon))/\varepsilon^2$ and $s = pm \geq c_{\ref*{t:ose-full}.2} (\log ^2(\frac{d}{\varepsilon\delta})/\varepsilon+\log^3(\frac{d}{\varepsilon\delta}))$, and this is done in the full version of the proof in Section \ref{subsec:osnapfinal}. 

Note that the expression for $\E[\tr(X^TX - I_d)^{2q}]$ depends only on $2q$ fold products of the entries of $X$. So, the quantity $\E[\tr(X^TX - I_d)^{2q}]$ remains unchanged if we only assume that subsets of the entries of $X$ of size $2q$ are independent instead of arbitrary subsets of the entries of $X$ being independent. Since it suffices to choose $q=\lceil 2\log (\frac{d}{\varepsilon \delta} )\rceil+2$, we only need $S$ to be an $O(\log(d/\varepsilon\delta))$-wise independent unscaled OSNAP.
\end{proof}
Finally, we show how the above arguments also imply the OSE moment property (Definition~\ref{d:osemoments}).
\begin{proof}[Proof of Corollary \ref{cor:osemoments}]
    By the proof of Theorem \ref{t:ose-full}, we see that when  $m \ge c_{\ref*{t:ose-full}.1}  (d + \log(1/\delta\varepsilon))/\varepsilon^2$ and  $s \ge \min\{c_{\ref*{t:ose-full}.2} (\log ^2(\frac{d}{\varepsilon\delta})/\varepsilon+\log^3(\frac{d}{\varepsilon\delta})),m\}$, then $\E[\tr(X^TX - I_d)^{2q}] \le \varepsilon^{2q}$, for $q=8\log(\frac{d}{\varepsilon \delta})$ (The proof originally has $q = \lceil 2\log (\frac{d}{\varepsilon \delta} )\rceil+2$, but upon going through the proof we see that $q=8\log(\frac{d}{\varepsilon \delta})$ also works). To get $\E[\tr(X^TX - I_d)^{2q}] \le \varepsilon^{2q}\delta/d$, it suffices for $m$ and $s$ to satisfy the same lower bounds, but with $\varepsilon$ replaced by $\varepsilon (\delta/d)^\frac{1}{2q} \ge c \varepsilon$ for some $c>0$ since $q \ge \log(d/\delta)$. These new lower bounds can be achieved by lower bounds of the same form as Theorem \ref{t:ose-full}, but with different constants. The claim follows, since $\|X^TX- I_d\|^{2q}\leq d\tr(X^TX-I_d)^{2q}$.
\end{proof}

\subsection{Controlling Trace Moments of the Embedding Error}

We now sketch the proof of Lemma \ref{prop:momestdecoupmatrix}, which obtains the moment bound for $X^TX-I$ used in the previous proof. The full proof can be found in Section~\ref{subsec:osnaptracemom}. 

\begin{proof}[Proof sketch of Lemma \ref{prop:momestdecoupmatrix}]
    Our first step is to observe that due to the property of $S$ having a fixed number of non-zero entries in a column, all quadratic terms in $(SU)^T(SU) - pm\cdot I_d$ are square-free, and this allows us to use the decoupling technique to reduce the problem of controlling the moments of $(SU)^T(SU) - pm\cdot I_d$ to controlling the moments of $(S_1U)^T(S_2U) + (S_2U)^T(S_1U)$ where $S_1$ and $S_2$ are independent copies. This is shown in the following claim, with the proof deferred to Section~\ref{subsec:decouposnap}.
\begin{lemmanon}[Decoupling, Lemma \ref{lem:decoup}] 
When $S$ has the fully independent unscaled OSNAP distribution, we have
\begin{align*}
    \E \Big[ \tr \big(U^TS^TSU - pm\cdot I_d\big)^{2q} \Big] &= \E \left[ \tr \left( \sum_{i=1}^m \sum_{j,j' =1, j \neq j'}^n s_{ij}s_{ij'} u_ju_{j'}^T \right)^{2q} \right]
\end{align*}
where $\{ u_j^T \}_{j \in [n]}$ denote the rows of $U$.
Consequently, we have
\begin{align*}
    \E \Big[ \tr \big(U^TS^TSU - pm\cdot I_d\big)^{2q} \Big] &\le \E_{S_1,S_2} \left[ \tr \left(  2\paren*{(S_1U)^TS_2U + (S_2U)^TS_1U} \right)^{2q} \right]
\end{align*}
where $S_2$ is an independent copy of $S_1$.
\end{lemmanon}

To estimate the moments of $(S_1U)^TS_2U$, we compare them to moments from the Gaussian case, i.e. the moments of $(G_1U)^TG_2U$ where the entries of $G_1$ and $G_2$ are independent normal random variables with variance $p$ (since the entries of $S_1$ and $S_2$ are also uncorrelated with mean 0 and variance $p$, see Lemma \ref{lem:osnapvaruncor})). In this case, due to orthogonal invariance of the Gaussian distribution, the matrices $G_1U$ and $G_2U$ are distributed as $\sqrt{p}H_1$ and $\sqrt{p}H_2$ where $H_1$ and $H_2$ are $m \times d$ matrices with independent standard normal entries. Thus, we can rely on the following bound, which uses standard results about the norms of Gaussian random matrices with independent entries. 

\begin{lemmanon}[Trace Moment of Embedding Error for Decoupled Gaussian Model, Lemma \ref{cor:indgaussianmom}]
    Let $H_1$ and $H_2$ be independent $m \times d$ random matrices with i.i.d. Gaussian entries. Then for any positive integer $q$, there exists $c_{\ref*{cor:indgaussianmom}}>0$ such that 
    \[ \E \left[ \tr \left( H_1^TH_2 + H_2^TH_1 \right)^{2q} \right]^\frac{1}{2q} \le c_{\ref*{cor:indgaussianmom}}\sqrt{\max\{d, q\}}\sqrt{\max\{m, q\}}.\]
\end{lemmanon}

To formally compare the moments of $(S_1U)^TS_2U$ and $(G_1U)^TG_2U$, we define the interpolating matrices $S_1(t), S_2(t)$ for $t \in [0,1]$ as described in Section \ref{s:overview}:
\begin{align}
\begin{split} 
    S_1(t) = \sqrt{t}S_1 + \sqrt{1-t}G_1, \\
    S_2(t) = \sqrt{t}S_2 + \sqrt{1-t}G_2.
\end{split}
\end{align}
Let $\Gamma(M_1,M_2)=(M_1U)^T(M_2U)+(M_2U)^T(M_1U)$ and $\Gamma(t)=\Gamma(S_1(t),S_2(t))$. Then, due to the decoupling lemma (Lemma \ref{lem:decoup}), to prove Lemma \ref{prop:momestdecoupmatrix} it is enough to show that $\E[\tr(\Gamma(1))^{2q}]^\frac{1}{2q} \le pm\varepsilon/2$. Now, by Lemma \ref{cor:indgaussianmom}, we know that:
$$\E\big[\tr(\Gamma(0))^{2q}\big]^\frac{1}{2q} = \E\big[\tr(\Gamma(G_1, G_2))^{2q}\big]^\frac{1}{2q} \le c_{\ref*{cor:indgaussianmom}}p\sqrt{\max\{d, q\}}\sqrt{\max\{m, q\}}.$$ 

Since we want to find the conditions for which $\E[\tr(\Gamma(0))^{2q}]^\frac{1}{2q} \le pm\varepsilon/4$, it is enough to ensure that $c_{\ref*{cor:indgaussianmom}}p\sqrt{\max\{d, q\}}\sqrt{\max\{m, q\}} \le pm\varepsilon/4$. Clearly, this can only happen when $q \le m$, and in this case the inequality holds when $m \ge \frac{c(d+q)}{\varepsilon^2}$. Thus, it suffices to show
\begin{align} \label{eq:univbound}
    \E\big[\tr \Gamma(1)^{2q}\big]^{\frac{1}{2q}}-\E\big[\tr \Gamma(0)^{2q}\big]^{\frac{1}{2q}} \le \frac{1}{4} pm \varepsilon.
\end{align}

For this, we look to estimate the derivative $\frac{d}{dt} \E[\tr \Gamma(t)^{2q}]$, and we obtain the following estimate in Lemma \ref{lem:diffineq} using the 2D interpolation idea mentioned in Section \ref{s:overview}.
\begin{diffineqlemma}[Differential Inequality]
For $\Gamma(t)$ as defined above, there exists a constant $c_{\ref*{lem:diffineq}}$ such that, for any $q \ge 2$, we have
\begin{align*}
\frac{d}{dt} \E\big[\tr \Gamma(t)^{2q}\big] \le& \max \limits_{4 \le k \le 2q} (c_{\ref*{lem:diffineq}}q)^k \Big((pm)^{\frac{1}{q}}\sqrt{\max\{pd,q\}}\Big)^{\frac{qk-2q}{q-1}}\E\!\big[\tr \Gamma(t)^{2q}\big]^{1-\frac{k-2}{2q-2}}.
\end{align*}
\end{diffineqlemma}
This differential inequality can be separated into two distinct cases: $pd \le q$ and $pd>q$. When $pd\leq q$, we can simplify the expression on the right using convexity arguments, and use Lemma 6.6 from \cite{brailovskaya2022universality} to solve the differential inequality and obtain the following bound:
\begin{align*}
    \E\!\big[\tr \Gamma(1)^{2q}\big]^{\frac{1}{2q}}-\E\!\big[\tr \Gamma(0)^{2q}\big]^{\frac{1}{2q}} \le c_7 (pm)^{\frac{1}{q}} q^2.
\end{align*}
for some $c_7>0$ (this is done in the full proof of Lemma \ref{prop:momestdecoupmatrix}). Thus, inequality \eqref{eq:univbound} is satisfied when $c_7(pm)^{\frac{1}{q}} q^2 < pm \varepsilon/4$, or
\[ pm \ge \frac{4c_7(pm)^{\frac{1}{q}} q^2}{\varepsilon}. \]
When $pd>q$, the expression on the right of the above differential inequality has some $pd$ factors. We replace these $pd$ factors by terms involving only $pm$ and $\E[\tr \Gamma(t)^{2q}]$ and similarly obtain:
\begin{align*}
    \E\!\big[\tr \Gamma(1)^{2q}\big]^{\frac{1}{2q}}-\E\!\big[\tr \Gamma(0)^{2q}\big]^{\frac{1}{2q}} \le c_{13}  q^3 (pm)^{\frac{2}{q}} \Big(\frac{d}{m}\Big)^{\frac{1}{2}}
\end{align*}
for some $c_{13}>0$. In this case, inequality \eqref{eq:univbound} is satisfied when 
\begin{align*}
    c_{13} q^3 (pm)^{\frac{2}{q}}\Big(\frac{d}{m}\Big)^{\frac{1}{2}} \le \frac{1}{4} pm \varepsilon.
\end{align*}
Since we have $m \ge \frac{c_{14}d}{\varepsilon^2}$ for some constant $c_{14}$, we have $\varepsilon \ge \sqrt{\frac{c_{14}d}{m}}$, so it suffices to require
\begin{align*}
    c_{13} q^3 (pm)^{\frac{2}{q}}\Big(\frac{d}{m}\Big)^{\frac{1}{2}} &\le \frac{1}{4}pm \sqrt{\frac{c_{14}d}{m}} \\
     \text{or, }\quad pm &\ge c_{15} (pm)^{\frac{2}{q}} {q^3}
\end{align*}
for some $c_{15}>0$.

Combining the analysis for the two cases, it suffices to require
\begin{align*}pm \ge (pm)^{\frac{2}{q}}\max\Big\{\frac{c_{16} q^{2}}{\varepsilon}, c_{17} q^3\Big\}
\end{align*}
for some constants $c_{16}>0$ and $c_{17}>0$. This requirement is equivalent to
\begin{align*}pm \ge \bigg(\max\Big\{\frac{c_{16} q^{2}}{\varepsilon}, c_{17} q^3\Big\}\bigg)^{\frac{1}{1-2/q}},
\end{align*}
which concludes the proof of Lemma \ref{prop:momestdecoupmatrix} (see remaining details in Section~\ref{subsec:osnaptracemom}).
\end{proof}

\subsection{Obtaining the differential inequality in Lemma \ref{lem:diffineq}}

We now discuss the proof of the technical part of our argument in the previous proof, which is to control the derivative of the interpolant. The full proof can be found in Section \ref{subsec:diffineq}.

\begin{proof}[Sketch of proof of Lemma \ref{lem:diffineq}]
    There are two main ideas for obtaining this differential inequality. First, we use the cumulant method as in \cite{brailovskaya2022universality} to transform the derivative in $t$ to matrix directional derivatives. Then, we bound the resulting terms in the expression by delicately using the matrix H\"older's inequality. 
    
    Fix $M_2$ and define $f_{1,M_2}(M_1):=\tr(\Gamma(M_1,M_2)^{2q})$ as a function of $M_1$. We shall first obtain an expression for $\frac{d}{dt}\E[f_{1,M_2}(S_1(t))]$. To see why this is sufficient, note that the derivative we are interested in is the directional derivative along the path $t \to (t,t)$ for the multivariate function $(t_1, t_2) \to \E [\tr(\Gamma(S_1(t_1), S_2(t_2))^{2q})]$ and by the chain rule (as mentioned in Section~\ref{s:overview}),
    \[ \frac{d}{dt} \E[\tr \Gamma(t)^{2q}] = \frac{d}{dt_1}\E [\tr(\Gamma(S_1(t_1), S_2(t_2))^{2q})] \bigg\vert_{t_1, t_2 = t} + \frac{d}{dt_2}\E [\tr(\Gamma(S_1(t_1), S_2(t_2))^{2q})] \bigg\vert_{t_1, t_2 = t}  \]

    Now, recall that $S_1$ can be written in the form $\sum_{(l,\gamma) \in \Xi} Z_{(l,\gamma)}$ where $\Xi=[n] \times [pm]$ and $Z_{(l,\gamma)}=\xi_{(l,\gamma)} e_{\mu_{(l, \gamma)}} e_l ^\top$ (see Definition \ref{def:osnap}). We then have the following lemma.
    \begin{lemmanon}[Based on Corollary 6.1, \cite{brailovskaya2022universality}] 
    For any polynomial $\phi:M_{m\times d}(\mathbb{R})\to\mathbb{R}$, we have
\begin{multline*}
\hspace{-3mm}\frac{d}{dt}\E[\phi(S_1(t))] \\
= \frac{1}{2}\sum_{k=4}^\infty
	\frac{t^{\frac{k}{2}-1}}{(k-1)!}
	\sum_{\pi\in\mathrm{P}([k])}
	(-1)^{|\pi|-1}(|\pi|-1)!\,
	\E\Bigg[ \sum_{(l,\gamma) \in \Xi}\partial_{Z_{(l,\gamma),1|\pi}}\cdots\partial_{Z_{(l,\gamma),k|\pi}}\phi(S_1(t))
	\Bigg],
\end{multline*}
where $\partial_Z\phi$ denotes the directional derivative of
$\phi$ in the direction $Z\in M_{m\times d}(\mathbb{R})$.
\end{lemmanon}
Here, $\mathrm{P}([k])$ denotes the set of all partitions of $[k]$, and $Z_{(l,\gamma),1|\pi} \etc Z_{(l,\gamma),k|\pi}$ are random matrices distributed as $Z_{(l,\gamma)}$. Crucially, those are independent of $S_1, G_1, S_2$ and $G_2$ (but not necessarily from each other). Further details are given in the full proof of Lemma \ref{lem:diffineq} in Section~\ref{subsec:diffineq}.

Applying this lemma to $\frac{d}{dt_1}\E[\tr(\Gamma(S_1(t_1),S_2(t_2))^{2q})] = \frac{d}{dt_1}\E[f_{1,S_2(t_2)}(S_1(t_1))] $, we need to deal with the directional derivatives of $f_{1,S_2(t_2)}$ along $Z_{(l,\gamma),1|\pi}, \ldots, Z_{(l,\gamma),k|\pi}$. Using a general expression for derivatives of multinomials via the product rule, we have, for any deterministic $m \times d$ matrices $B_1, \ldots, B_k, M_1$ and $M_2$,
\begin{align*}
    &\partial_{B_1} \cdots \partial_{B_k}f_{1,M_2}(M_1)
\\&\ \ =\!\!\sum_{\sigma \in \sym (k)}\sum_{\substack{r_1,\ldots,r_{k+1}\ge 0\\
	r_1+\cdots+r_{k+1}=2q-k}}
	\tr \Big(\Gamma(M_1,M_2)^{r_1}((B_{\sigma(1)}U)^TM_2U+(M_2U)^TB_{\sigma(1)}U)
	\Gamma(M_1,M_2)^{r_2}\\[-5mm]&
    \hspace{5.25cm}((B_{\sigma(2)}U)^TM_2U+(M_2U)^TB_{\sigma(2)}U)\cdots
\Gamma(M_1,M_2)^{r_k}\\&
\hspace{5.25cm}((B_{\sigma(k)}U)^TM_2U+(M_2U)^TB_{\sigma(k)}U)\Gamma(M_1,M_2)^{r_{k+1}}\Big).
\end{align*}

In our case, for each fixed $(l, \gamma)$, we have to analyze $\partial_{Z_{(l,\gamma),1|\pi}}\cdots\partial_{Z_{(l,\gamma),k|\pi}} \allowbreak f_{1,S_2(t_2)}(S_1(t_1))$, which means that we have $B_{\lambda} = Z_{(l,\gamma),\lambda|\pi}$ for $\lambda \in [k]$, and $M_2 = S_2(t)$. So terms of the form $(B_{\lambda}U)^TM_2U$ become $(Z_{(l,\gamma),\lambda|\pi}U)^TS_2(t)U$. 
Crucially, $(Z_{(l,\gamma),\lambda|\pi}U)^TS_2(t)U$ is a rank one matrix, so it can be written as an outer product of the form $\Theta_{(l, \gamma), \lambda,1}^T\Theta_{(l, \gamma), \lambda,2}$. 

Then, estimating $$\E\Bigg[ \sum_{(l,\gamma) \in \Xi}\partial_{Z_{(l,\gamma),1|\pi}}\cdots\partial_{Z_{(l,\gamma),k|\pi}}f_{1,S_2(t_2)}(S_1(t_1)) \Bigg]$$ for $t_1=t_2=t$ boils down to estimating terms of the form
\[ \E \Big[ \tr \Gamma(t)^{r_1} \Theta_{(l,\gamma), \sigma(1), \tau_1(1)}^T \Theta_{(l,\gamma), \sigma(1), \tau_1(2)} \cdot
	\Gamma(t)^{r_2} \cdots \Gamma(t)^{r_k}\Theta_{(l,\gamma), \sigma(k), \tau_k(1)}^T \Theta_{(l,\gamma), \sigma(k), \tau_k(2)} \Gamma(t)^{r_{k+1}} \Big] \]
where $\tau_i \in \sym (\{1,2\})$ are permutations of the set $\{1,2\}$.

   For the remainder of the proof, we delicately analyze terms of this form using the matrix H\"older's inequality, and appropriately estimate the terms that arise.
    \begin{lemmanon}[Matrix H\"older's inequality, Lemma 5.3 in \cite{brailovskaya2022universality}]
Let $1\le 
\beta_1,\ldots,\beta_k\le\infty$ satisfy $\sum_{i=1}^k\frac{1}{\beta_i}=1$. Then
$$
	\Big|\E\!\big[\tr Y_1\cdots Y_k\big]\Big| \le
	\|Y_1\|_{\beta_1}\cdots \|Y_k\|_{\beta_k}
$$
for any $d\times d$ random matrices $Y_1,\ldots,Y_k$.
\end{lemmanon}
 This analysis based on matrix H\"older's inequality is done in Lemma Lemma \ref{lem:decouptraceineq}.  One important observation we use in this lemma (among many others) is that $\Theta_{(l, \gamma), \lambda,1}^T\Theta_{(l, \gamma), \lambda,2}$ are rank one matrices, which allows us to bound $\| \Theta_{(l, \gamma), \lambda,1}^T\Theta_{(l, \gamma), \lambda,2} \|_q$ with $\sqrt{pd}$ instead of $\sqrt{pm}$. For further details, please refer to the full proof of  Lemma \ref{lem:decouptraceineq} in Section \ref{subsec:osnaptraceineq}.
\end{proof}

\section{Preliminaries} \label{sec:prelim}

\subsection{Oblivious Subspace Embeddings}

Here, we prove some important properties of the OSNAP distribution that we shall use later.

\begin{lemma}[Variance and Uncorrelatedness] \label{lem:osnapvaruncor}
Let $p = p_{m,n} \in (0,1]$ and $S=\{s_{ij}\}_{i \in [m], j \in [n]}$ be a $m \times n$ random matrix as in the unscaled OSNAP distribution. Then, $\E(s_{ij})=0$ and $\operatorname{Var}(s_{ij})=p$ for all $i \in [m], j \in [n]$, and $\cov(s_{i_1 j_1},s_{i_2 j_2})=0$ for any $\{i_1,i_2\} \subset [m], \{j_1,j_2\} \subset [n]$ and $ (i_1,j_1)\neq (i_2,j_2) $
\end{lemma}
\begin{proof}

Recall that in the unscaled OSNAP distribution, each subcolumn \begin{align*}S_{[(m/s)(\gamma-1)+1:(m/s)\gamma]\times\{ j\}}\end{align*} of $S$ for $\gamma\in [s], j \in [n]$ has the one hot distribution, and all these subcolumns are jointly $\log(mn)$-independent. Therefore, the same argument as above, we directly calculate the variances of the entries
\begin{align*}
    \E (s_{ij}^2) = \Pb(s_{ij} \neq 0) = \frac{1}{\frac{m}{pm}} = p
\end{align*}

For the covariances, we first observe that $\cov(s_{i_1 j_1}, s_{i_2 j_2})=0$ if $(i_1,j_1)$ and $(i_2,j_2)$ belong to two different subcolumns by $\log$-independence. More precisely, $\log(d/\varepsilon\delta)$-independence implies $2$-independence which implies zero covariance. If $(i_1,j_1)$ and $(i_2,j_2)$ belong to the same subcolumn, we have
\begin{align*}
    \cov(s_{i_1 j_1}, s_{i_2 j_2})=\E(s_{i_1 j_2}s_{i_2 j_1})=\E(0)=0
\end{align*}
because each subcolumn has one hot distribution which means at most one of $s_{i_1 j_1}$ and $s_{i_2 j_2}$ can be nonzero.
\end{proof}

\begin{lemma}[Norm of a Random Row in Interpolated OSNAP]\label{lem:rownormbound}
    Let $S(t):= \sqrt{t}S + \sqrt{1-t}G$, where $S$ is as in the fully independent unscaled OSNAP distribution and $G$ is an $m \times n$ matrix with i.i.d. Gausian entries with variance $p$. Let $U$ be an $n \times d$ matrix such that $U^TU=I$. Let $\mu$ be a random variable uniformly distributed in $
    J \subset [m]$ and independent of $S$ and $G$. Then, there exists $c_{\ref*{lem:rownormbound}}> 0$ such that for any positive integer $q>0$, we have
    \begin{align*}
        \E_{\mu, S(t)} [ \norm{e_{\mu}^TS(t)U}^{q} ]^\frac{1}{q} \le c_{\ref*{lem:rownormbound}}\sqrt{\max\{pd,q \}} 
    \end{align*}

\end{lemma}
\begin{proof}
    By Hölder's inequality, it suffices to prove these bounds for moments of the order of the smallest even integer bigger than $q$, so without loss of generality, we may assume that $q$ is an even integer. 
    \begin{align*}
        \E_{\mu, S(t)} [ \norm{e_{\mu}^TS(t)U}^{q}] =& \frac{1}{\abs{J}}\sum_{j \in J} \E_{S(t)} [ \norm{e_{j}^TS(t)U}^{q}] \\
        & = \E_{S(t)} [ \norm{e_{1}^TS(t)U}^{q}]
    \end{align*}
    since rows of $S(t)U$ are identically distributed.

Note that $S_{1,1}$ (which is the $(1,1)$th entry of $S$) is non-zero when the one hot distribution on the column submatrix $S_{[1:m/s]\times{1}}$ has it's non zero entry on the first row. The probability that this happens is $1/(m/s) = s/m = p$. Moreover, the columns of $S$ are independent. Thus, we conclude that for looking at the distribution of $e_{1}^TS(t)U$ we may assume that $S$ has independent 
$\pm 1$ entries sparsified by independent Bernoulli $p$ random variables, i.e., $s_{i,j}=\delta_{i,j} \xi_{i,j}$ where $\delta_{i,j}$ are Bernoulli random variables taking value 1 with probability $p \in (0,1]$, $\xi_{i,j}$ are random variables with $\Pb(\xi_{i,j}=1)=\Pb(\xi_{i,j}=-1)=1/2$ and the collection $\{\delta_{i,j}, \xi_{i,j} \}_{i \in [m], j \in [n]}$ is independent.

Now, conditioned on $G$, $\norm{e_{1}^TS(t)U}$ is a convex function of the entries of  the first row of $S$, and
\begin{align*}
    E_S[\norm{e_{1}^TS(t)U}] &\le  E_S[\norm{e_{1}^TS(t)U}^2]^\frac{1}{2} \\
    &= \sqrt{t\cdot pd + (1-t)\norm{e_{1}^TGU}^2}
\end{align*}
By \cite[Theorem 6.10]{boucheron2013concentration}, conditioned on $G$, we have, for $\lambda > 0$,
\begin{align*}
    \Pb_{S} ( \norm{e_{1}^TS(t)U} \ge \sqrt{t\cdot pd + (1-t)\norm{e_{1}^TGU}^2} + \lambda ) \le \exp(-c_1\lambda^2)
\end{align*}
for some $c_1>0$.
Next, let $G = \sqrt{p}H$, where $H$ is an $m \times d$ matrix with independent standard normal entries. By orthogonal invariance of Gaussians, the random vector $e_1^THU$ is a $d$ dimensional vector with independent standard normal entries. By concentration of norm for Gaussian vectors, \cite[Theorem 3.1.1]{vershyninhdp}, we have
\begin{align*}
    \Pb( \norm{e_1^THU} \ge \sqrt{d} + \lambda ) \le \exp(-c_2\lambda^2)
\end{align*}
which after scaling becomes
\begin{align*}
    \Pb( \norm{e_1^TGU} \ge \sqrt{pd} + \lambda\sqrt{p} )
\end{align*}
Therefore, we have
\begin{align*}
    \Pb( \norm{e_1^TGU} \ge \sqrt{pd} + \lambda ) &\le \exp(-c_2\lambda^2)
\end{align*}
which after taking square becomes
\begin{align*}
    \Pb( \norm{e_1^TGU}^2 \ge 2pd + 2\lambda^2 ) &\le \exp(-c_2\lambda^2)
\end{align*}
for some $c_2 > 0$.

Let $\mathcal{E}$ be the event that $\norm{e_1^TGU}^2 \ge 2pd + 2\lambda^2$. Then, $\Pb(\mathcal{E}) \le \exp(-c_2\lambda^2)$ and we have, 
\begin{align*}
    &\Pb ( \norm{e_{1}^TS(t)U} \ge \sqrt{2pd} + \sqrt{2}\lambda ) \\
    \le &\Pb ( \norm{e_{1}^TS(t)U} \ge \sqrt{2pd + 2\lambda^2} ) \\
    \le &\Pb(\mathcal{E}) + \Pb ( \norm{e_{1}^TS(t)U} \ge \sqrt{2pd + 2\lambda^2} \big\vert \mathcal{E}^C )
\end{align*}

When $\norm{e_{1}^TGU}^2 \le 2pd + 2\lambda^2$, $\sqrt{t\cdot pd + (1-t)\norm{e_{1}^TGU}^2} \le \sqrt{2pd + 2\lambda^2}$, so
\begin{align*}
    &\Pb ( \norm{e_{1}^TS(t)U} \ge \sqrt{2pd} + \sqrt{2}\lambda ) \\
    \le& \Pb(\mathcal{E}) + \Pb ( \norm{e_{1}^TS(t)U} \ge \sqrt{2pd + 2\lambda^2} \big\vert \mathcal{E}^C ) \\
    \le& \Pb(\mathcal{E}) + \Pb ( \norm{e_{1}^TS(t)U} \ge \sqrt{t\cdot pd + (1-t)\norm{e_{1}^TGU}^2} \big\vert \mathcal{E}^C  ) \\
    \le& \exp(-c_2\lambda^2) + \exp(-c_1\lambda^2) \\
    \le& \exp(-c_3\lambda^2)
\end{align*}
for $\lambda \ge \sqrt{pd} \ge 1$ and some constant $c_3>0$, i.e. $\norm{e_{1}^TS(t)U}$ has a subgaussian tail. By standard moment computation of subgaussian random variables, we conclude that $\E[ \norm{e_{1}^TS(t)U}^q ] \le (\sqrt{2pd})^q + (c_4q)^\frac{q}{2}$ for some $c_4\ge 0$. Thus,

\begin{align*}
    \E[ \norm{e_{1}^TS(t)U}^q ] &\le (c_5\sqrt{\max\{pd,q\}})^q 
\end{align*}
for some constant $c_5$.

\end{proof}

\subsection{Basic Facts of the $L_q(S_q^d)$ Space}

To derive the trace inequalities that will be used in the interpolation argument, we need the following tools from \cite{brailovskaya2022universality}. For a $d \times d$ matrix $M$, following \cite{brailovskaya2022universality}, we define the absolute value $|M|=\sqrt{M^*M}$ and normalized trace $\tr(M)=\frac {1}{d}\Tr(M)$. Let $L_q(S_q^d)$ be the normed vector space of $d\times d$ random matrices $M$ with norm
\begin{align*}
	\norm{M}_q = \begin{cases}
	\big(\E[\tr |M|^q]\big)^{\frac{1}{q}} & \text{if }1\le q<\infty\\
	\| \norm{M}_{op} \|_\infty & \text{if }q=\infty
	\end{cases}
\end{align*}
More precisely, the space $L_q(S_q^d)$ consists of random matrices $M$ with $\norm{M}_q$ well defined (which means $\big(\E[\tr |M|^q]\big)<\infty$ for $1\le q<\infty$ and $\| \norm{M}_{op} \|_\infty$ for $q=\infty$).

Next, we state a H\"older inequality for Schatten classes proved in \cite{brailovskaya2022universality}.

\begin{lemma}[Lemma 5.3. in \cite{brailovskaya2022universality}]
\label{lem:holdervanhandel}
Let $1\le 
\beta_1,\ldots,\beta_k\le\infty$ satisfy $\sum_{i=1}^k\frac{1}{\beta_i}=1$. Then
$$
	|\E[\tr Y_1\cdots Y_k]| \le
	\|Y_1\|_{\beta_1}\cdots \|Y_k\|_{\beta_k}
$$
for any $d\times d$ random matrices $Y_1,\ldots,Y_k$.
\end{lemma}

The proof of Lemma \ref{lem:holdervanhandel} (which we do not include here) relies on convexity and interpolation in $L_q(S_q^d)$. More precisely, one can first prove the result for the case when $\beta_1,...,\beta_k$ are extreme exponents and then extend the result to general $\beta_1,...,\beta_k$ by the following convexity result. Later we will also use this method to prove our trace inequalities (Lemma \ref{lem:decouptraceineq}).

\begin{lemma}[Lemma 5.2. in \cite{brailovskaya2022universality}]
\label{lem:calderon}
Let $F:(L_\infty(S_\infty^d))^k\to\mathbb{C}$ be a multilinear functional.
Then the map
$$
	\bigg(\frac{1}{\beta_1},\ldots,\frac{1}{\beta_k}\bigg)
	\mapsto
	\log \sup_{M_1,\ldots,M_k}
	\frac{|F(M_1,\ldots,M_k)|}{\|M_1\|_{\beta_1}\cdots\|M_k\|_{\beta_k}}
$$
is convex on $[0,1]^k$.
\end{lemma}

More generally, we have the following complex interpolation result from \cite{bergh2012interpolation}.
\begin{theorem}[4.4.1 Theorem in \cite{bergh2012interpolation}]\label{thm:compinterp}
    Let $\{(A_{(\nu,0)},A_{(\nu,0)})\}_{\nu=1}^{n}$ and $(B_0,B_1)$ be compatible Banach spaces. Assume that $T: (A_{(1,0)} \cap A_{(1,0)}) \oplus \cdots \oplus (A_{(n,0)} \cap A_{(n,0)}) \to (B_0 \cap B_1)$ is multilinear.
    Also, assume that for any $(a_1,...,a_n) \in (A_{(1,0)} \cap A_{(1,0)}) \oplus \cdots \oplus (A_{(n,0)} \cap A_{(n,0)})$, we have
    \begin{align*}
        \norm{T(a_1,...,a_n)}_{B_0} \le M_0 \prod \limits_{\nu=1}^{n} \norm{a_{\nu}}_{A_{\nu,0}}
    \end{align*}
    and
    \begin{align*}
        \norm{T(a_1,...,a_n)}_{B_1} \le M_1 \prod \limits_{\nu=1}^{n} \norm{a_{\nu}}_{A_{\nu,1}}
    \end{align*}
    Then for any $0 \le \theta \le 1$, the function $T$ may be uniquely extended to a multilinear mapping from $(A_{(1,0)}, A_{(1,0)})_{[\theta]} \oplus \cdots \oplus (A_{(n,0)}, A_{(n,0)})_{[\theta]}$ to $(B_0,B_1)_{[\theta]}$ with norm at most $M_0^{1-\theta}M_1^{\theta}$, where $(A_{(\nu,0)}, A_{(\nu,0)})_{[\theta]}$ is the complex interpolation space of the pair $(A_{(\nu,0)}, A_{(\nu,0)})$ with exponent $\theta$ defined in \cite[p.88]{bergh2012interpolation}.
\end{theorem}

Applying Theorem \ref{thm:compinterp} to $L_q(S_q^d)$ spaces, we have a more general version of Lemma \ref{lem:calderon}.
\begin{corollary}[Interpolation in $L_q(S_q^d)$ Space]\label{cor:spinterpo}
    Let\begin{align*}(\frac{1}{\beta_{1}(0)},...,\frac{1}{\beta_{k}(0)}) \in [0,1]^k \text{ and }(\frac{1}{\beta_{1}(1)},...,\frac{1}{\beta_{k}(1)}) \in [0,1]^k\end{align*} Assume that
    \begin{align*}
        F: (L_{\beta_{1}(0)}(S_{\beta_{1}(0)}) \cap L_{\beta_{1}(1)}(S_{\beta_{1}(1)})) \oplus \cdots \oplus (L_{\beta_{k}(0)}(S_{\beta_{k}(0)}) \cap L_{\beta_{k}(1)}(S_{\beta_{k}(1)})) \to \R
    \end{align*}
    is multilinear with
    \begin{align*}
        F(M_1,...,M_k) \le K \prod \limits_{\nu=1}^{n} \norm{M_{\nu}}_{\beta_{\nu}(0)}
    \end{align*}
    and
    \begin{align*}
        F(M_1,...,M_k) \le K \prod \limits_{\nu=1}^{n} \norm{M_{\nu}}_{\beta_{\nu}(1)}
    \end{align*}
    for all $(M_1,...,M_k) \in (L_{\beta_{1}(0)}(S_{\beta_{1}(0)}) \cap L_{\beta_{1}(1)}(S_{\beta_{1}(1)})) \oplus \cdots \oplus (L_{\beta_{k}(0)}(S_{\beta_{1}(0)}) \cap L_{\beta_{k}(1)}(S_{\beta_{1}(1)}))$.
Define $\beta_{\nu}(\theta)$ for $\nu=1,2,...,k$  and $0 \le \theta \le 1$ such that
\begin{align*}
    \frac{1}{\beta_{1}(\theta)}=(1-\theta)\frac{1}{\beta_{1}(0)}+\theta \frac{1}{\beta_{1}(1)}
\end{align*}
    Then for any $0 \le \theta \le 1$, the multilinear functional $F$ can be uniquely extended to \begin{align*}L_{\beta_{1}(\theta)}(S_{\beta_{1}(\theta)})  \oplus \cdots \oplus L_{\beta_{k}(\theta)}(S_{\beta_{1}(\theta)}) \end{align*} with
    \begin{align*}
        F(M_1,...,M_k) \le K \prod \limits_{\nu=1}^{n} \norm{M_{\nu}}_{\beta_{\nu}(\theta)}
    \end{align*}
    for all $(M_1,...,M_k) \in L_{\beta_{1}(\theta)}(S_{\beta_{1}(\theta)})  \oplus \cdots \oplus L_{\beta_{k}(\theta)}(S_{\beta_{1}(\theta)}) $.
\end{corollary}
\begin{proof}
    By \cite{pisier2003non}, the spaces $L_{\beta_{\nu}(\theta)}(S_{\beta_{\nu}(\theta)})$ are complex interpolation spaces of the compatible Banach spaces $(L_{\beta_{\nu}(0)}(S_{\beta_{\nu}(0)}),L_{\beta_{\nu}(1)}(S_{\beta_{\nu}(1)}))$. The result follows from Theorem \ref{thm:compinterp}.
\end{proof}

\subsection{Spectrum of Gaussian Matrices}
Here we collect results about the spectrum of various Gaussian models that will be used later in conjunction with universality results.

   \begin{lemma}[(2.3), \cite{rudelson2010non}]\label{lem:Gaussianspectrum}
    For $m>d$, let $G$ be an $m \times d$ matrix whose entries are independent standard normal variables. Then,
    \begin{align*}
        \Pb(\sqrt{m}-\sqrt{d}-t \leq s_{\min}(G) \leq s_{\max}(G) \leq \sqrt{m}+\sqrt{d}+t) \geq 1 - 2e^{-t^2/2}
    \end{align*}
\end{lemma}

\begin{corollary}[Trace Moment of Embedding Error for Gaussian Model] \label{cor:gaussianmom}
    Let $G$ be an $m \times d$ matrix whose entries are independent normal random variables with variance $\frac{1}{m}$. Let $\varepsilon<\frac{1}{6}$ and $q \in \mathbb{N} \le m \varepsilon^2 $. Then, there exists $c_{\ref*{cor:gaussianmom}} > 1$ such that for $m \geq \frac{c_{\ref*{cor:gaussianmom}}d}{\varepsilon^2}$,
    \[ \E[\tr(G^TG - I_d)^{2q}]^\frac{1}{2q} \leq  \varepsilon \]
    
\end{corollary}

\begin{proof}

    By lemma \ref{lem:Gaussianspectrum}, we have
\begin{align*}
    \Pb(\sqrt{m}-\sqrt{d}-t \leq s_{\min}(\sqrt{m}G) \leq s_{\max}(\sqrt{m}G) \leq \sqrt{m}+\sqrt{d}+t) &\geq 1 - 2e^{-t^2/2}
\end{align*}
which after scaling becomes
\begin{align*}
    \Pb(1-\sqrt{\frac{d}{m}}- \frac{t}{\sqrt{m}} \leq s_{\min}(G) \leq s_{\max}(G) \leq 1+\sqrt{\frac{d}{m}}+ \frac{t}{\sqrt{m}}) &\geq 1 - 2e^{-t^2/2}
\end{align*}
And therefore, we have
\begin{align*}
    \Pb ( \| G^TG - I_d \| \leq ( 1 + \sqrt{\frac{d}{m}}+ \frac{t}{\sqrt{m}})^2 - 1 ) &\geq 1 - 2e^{-t^2/2}
\end{align*}

Let $s = \sqrt{\frac{d}{m}}+ \frac{t}{\sqrt{m}}$. We have
\begin{align*}
    \Pb ( \| G^TG - I_d \| \ge 2s + s^2 )  &\le 2\exp(-\frac{1}{2}(\sqrt{m}s-\sqrt{d})^2) \\
    &\le \exp(-\frac{m}{2} (s-\sqrt{d/m})^2) \\
\end{align*}

When $m \ge \frac{d}{\varepsilon^2}$, $\sqrt{d/m} \le \varepsilon$, the above result becomes 
\begin{align*}
    \Pb ( \| G^TG - I_d \| \ge 2s + s^2 )  &\le \exp(-\frac{m}{2} (s-\varepsilon)^2) \\
\end{align*} 
when $s \ge \varepsilon$.

In conclusion, we have
\begin{align*}
    \Pb ( \| G^TG - I_d \| \ge 3s)  &\le \exp(-\frac{m}{8} s^2) \qquad \text{ when } 2\varepsilon \le s < 1 
\end{align*}
and
\begin{align*}
     \Pb ( \| G^TG - I_d \| \ge 3s^2)  &\le \exp(-\frac{m}{8} s^2) \qquad \text{ when }  s \ge  1 
\end{align*}

By change of variables, we have
\begin{align*}
    \Pb ( \| G^TG - I_d \| \ge s')  &\le \exp(-\frac{m}{72} s'^2) \qquad \text{ when } 6\varepsilon \le s' < 3 \end{align*} and \begin{align*}
    \Pb ( \| G^TG - I_d \| \ge s')  &\le \exp(-\frac{m}{24} s') \qquad \text{ when }  s' \ge  3 \\
\end{align*}

Denote the random variable $\| G^TG - I_d \|$ by $R$. Then,
\begin{align*}
    \E (R^{2q}) &= 2q \int_0^\infty t^{2q-1}\Pb(R \ge t) dt \\
    &\le 2q(6\varepsilon)^{2q} + 2q \int_{6\varepsilon}^3 t^{2q-1}\exp(-\frac{m}{72} t^2) dt + 2q \int_{3}^\infty t^{2q-1}\exp(-\frac{m}{24} t) dt
\end{align*}
Performing change of variables $x = mt^2/72$ and $y = mt/24$ in the second and third integrals respectively,
\begin{align*}
    \E (R^{2q}) &\le 2q(6\varepsilon)^{2q} + 2q \int_{0}^\infty \left(\frac{72x}{m} \right)^\frac{2q-1}{2}\exp(-x) \frac{72}{2m}\left(\frac{72x}{m} \right)^{-1/2}  dx \\ &\qquad \qquad \qquad + 2q \int_{0}^\infty \left(\frac{24y}{m} \right)^{2q-1}\exp(-y) \frac{24}{m} dy \\
    &\le 2q(6\varepsilon)^{2q} + q \left(\frac{72}{m} \right)^q \int_{0}^\infty x^{q-1}\exp(-x) dx + 2q \left(\frac{24}{m} \right)^{2q} \int_{0}^\infty y^{2q-1}\exp(-y) dy \\
    &\le 2q(6\varepsilon)^{2q} + q \left(\frac{72}{m} \right)^q (q-2)! + 2q \left(\frac{24}{m} \right)^{2q} (2q-2)! \\
    &\le 2q(6\varepsilon)^{2q} +  \left(\frac{72q}{m} \right)^q + \left(\frac{48q}{m} \right)^{2q} \\
\end{align*}
Since $q \le m\varepsilon^2 $,  $\left(\frac{72q}{m} \right)^q \le  \left(72 \varepsilon^2 \right)^{q}$ and $\left(\frac{48q}{m} \right)^{2q} \le  \left(48 \varepsilon^2 \right)^{2q} \le  \left(48 \varepsilon \right)^{2q} $. So,
\begin{align*}
    \E (R^{2q}) &\le (2q+2 )(48\varepsilon)^{2q} \\
    \E (R^{2q})^\frac{1}{2q} &\le (2q+2 )^\frac{1}{2q} 48\varepsilon 
\end{align*}
Now, $48(2q+2 )^\frac{1}{2q} < C$, for some $C>0$. Thus,
\[ \E (R^{2q})^\frac{1}{2q} \le C\varepsilon \]
for some $C > 1$ independent of $\varepsilon$. When $m \ge C^2\frac{d}{\varepsilon^2}$, we can run the same argument with $\varepsilon/C$ in place of $\varepsilon$ to get, 
\[ \E (R^{2q})^\frac{1}{2q} \le \varepsilon \]

\end{proof}

\begin{lemma}[Trace Moment of Embedding Error for Decoupled Gaussian Model]\label{cor:indgaussianmom}
    Let $G_1$ and $G_2$ be independent $m \times d$ random matrices with i.i.d. Gaussian entries. Then for any positive integer $q$, there exists $c_{\ref*{cor:indgaussianmom}}>0$ such that 
    \[ \E \left[ \tr \left( G_1^TG_2 + G_2^TG_1 \right)^{2q} \right]^\frac{1}{2q} \le c_{\ref*{cor:indgaussianmom}}\sqrt{\max\{d, q\}}\sqrt{\max\{m, q\}}\]
    
\end{lemma}

\begin{proof}
    We wish to bound, 
    \begin{align*}
        \E \left[ \tr \left( G_1^TG_2 + G_2^TG_1 \right)^{2q} \right]^\frac{1}{2q} &\le \E \left[ \norm{G_1^TG_2 + G_2^TG_1}^{2q}  \right]^\frac{1}{2q} \\
        &\le \E \left[ \left(\norm{G_1^TG_2} + \norm{G_2^TG_1}\right)^{2q}  \right]^\frac{1}{2q} \\
        &\le \E \left[ \norm{G_2^TG_1}^{2q} \right]^\frac{1}{2q} + \E \left[ \norm{G_1^TG_2}^{2q} \right]^\frac{1}{2q}
    \end{align*}

    Since $G_2^TG_1$ and $G_1^TG_2$ are identically distributed, it is enough to bound $\E \left[ \norm{G_1^TG_2}^{2q} \right]^\frac{1}{2q}$.
    We shall first condition on $G_2$ and compute $\E_{G_1} \left[ \norm{G_1^TG_2}^{2q} \right]$
    Let $G_2 = U \abs{G_2} V^T$ be the SVD of $G_2$ where $U$ is an $m \times d$ matrix with orthogonal columns. Then, 
    \begin{align*}
        \E_{G_1} \left[ \norm{G_1^TG_2}^{2q} \right] &\le \E_{G_1} \left[ \norm{G_1^TU\abs{G_2}V^T}^{2q} \right] \\
        &\le \norm{G_2}^{2q} \E_{G_1} \left[ \norm{G_1^TU}^{2q} \right]
    \end{align*}
    By orthogonal invariance of Gaussians, $G_1^TU$ is a $d \times d$ Gaussian matrix. Thus, 
    \[ \E_{G_1} \left[ \norm{G_1^TU}^{2q} \right] \le (c_1\sqrt{\max\{d, q\}})^{2q}\]
for some constant $c_1>0$, by applying Proposition 2.1 and Lemma 2.2 from \cite{bandeira2016sharp} to 
\begin{align*}\sym(G_1^TU) := \left[ {\begin{array}{*{20}{c}}
0&G_1^TU\\
{{G_1^TU^T}}&0
\end{array}} \right].
\end{align*}
Similarly, we can get
    \[ \E[\norm{G_2}^{2q}] \le (c_2\sqrt{\max\{m, q\}})^{2q} \]
\end{proof}

\section{Full proof of Theorem \ref{t:ose-full}}\label{sec:osnap-proof}

This section contains the full proof of Theorem  \ref{t:ose-full} along with the various lemmas that are used along the way.

\begin{itemize}
    \item Section \ref{subsec:osnapfinal} proves the final subspace embedding guarantee using a bound on the trace moments of the embedding error.
    \item Section \ref{subsec:decouposnap} has details about the decoupling step that reduces the problem of controlling moments of $(SU)^TSU - pm\cdot I_d$ to controlling moments of $(S_1U)^TS_2U+(S_2U)^TS_1U$, for independent $S_1$ and $S_2$.
    \item Section \ref{subsec:osnaptracemom} proves the trace moment bound using 2D interpolation of moments of the form $(S_1(t)U)^TS_2(t)U+(S_2(t)U)^TS_1(t)U$.
    \item Section \ref{subsec:diffineq} obtains the differential inequality for the derivative of the interpolant.
    \item Section \ref{subsec:osnaptraceineq} proves the trace inequality required for obtaining the differential inequality.
\end{itemize}

\subsection{Proving the subspace embedding guarantee for OSNAP} \label{subsec:osnapfinal}
In this section we give the full proof of Theorem \ref{t:ose-full}.
\osnapmainthm*

    \begin{proof}
    Let $X:= \frac{1}{\sqrt{pm}}SU$. We first assume that the collection of all the random variables 
$\{ \xi_{(l,\gamma)}, \mu_{(l,\gamma)} \}_{l \in [n], \gamma \in [s]}$ 
    in the unscaled OSNAP construction are fully independent, and later we will check what is the minimum independence needed.
    
    We observe that to prove the theorem, it is enough to show that
    \begin{align*} \Pb \left( \| X^TX - I_d \| \le \varepsilon \right) \ge 1-\delta \end{align*}
    because
    \begin{align*} \| X^TX - I_d \| < \varepsilon \end{align*}
    will imply \begin{align*}1 - \varepsilon \leq v^TX^TXv \le 1 + \varepsilon \end{align*} and \begin{align*} (1-\varepsilon)^2 \le \| X v \|^2 \le (1+\varepsilon)^2 \end{align*}
    for any $v \in \mathbb{R}^d$, and then the claim about singular values follows from the min-max principle. To show
    \begin{align*} \Pb \left( \| X^TX - I_d \| \le \varepsilon \right) \ge 1-\delta \end{align*} we first use Markov's inequality and obtain
    \begin{align*}
        \Pb \left( \| X^TX - I_d \|^{2q} \ge \delta^{-1} \E [\| X^TX - I_d \|^{2q}] \right) &\le \delta
    \end{align*}
    Taking the $(2q)$th root, we have
    \begin{align*}
        \Pb \left( \| X^TX - I_d \| \ge \delta^{-\frac{1}{2q}} \E [\| X^TX - I_d \|^{2q}]^{\frac{1}{2q}} \right) &\le \delta
    \end{align*}
We observe that
\begin{align*}
    \| X^TX - I_d \|^{2q} =(s_{\max}(X^TX - I_d))^{2q} \le \sum \limits_{i=1}^d (s_{i}(X^TX - I_d))^{2q} = d \tr (X^TX - I_d)^{2q}
\end{align*}
Therefore, we have
\begin{align*}
        \Pb \left( \| X^TX - I_d \| \ge \delta^{-\frac{1}{2q}} \E [d \tr (X^TX - I_d)^{2q}]^\frac{1}{2q} \right) &\le \delta
    \end{align*}
which, after simplification, becomes
\begin{align*}
        \Pb \left( \| X^TX - I_d \| \ge (d/\delta)^{\frac{1}{2q}} \E [ \tr (X^TX - I_d)^{2q}]^\frac{1}{2q} \right) &\le \delta
    \end{align*}    
For $q> \log (\frac{d}{\delta}) $, we have
\begin{align*}
    (d/\delta)^{\frac{1}{2q}}=\exp(\log(d/\delta)\frac{1}{2q}) \le \exp(\log(d/\delta)\frac{1}{2\log (\frac{d}{\delta})}) \le {\sqrt{e}}
\end{align*}
Therefore, we have
\begin{align*}
        \Pb \left( \| X^TX - I_d \| \ge \sqrt{e} \E [ \tr (X^TX - I_d)^{2q}]^\frac{1}{2q} \right) &\le \delta
    \end{align*}
By Theorem \ref{prop:momestdecoupmatrix}, when
$m \geq c_{\ref*{prop:momestdecoupmatrix}.1} \frac{d+q}{(\varepsilon/\sqrt{e})^2}$ and $q \in \N$ satisfying $2 \le q \le m$ and \begin{align*}pm \ge (\max\{\frac{c_{\ref*{prop:momestdecoupmatrix}.2} q^{2}}{(\varepsilon/\sqrt{e})}, c_{\ref*{prop:momestdecoupmatrix}.3} q^3\})^{1+\frac{2}{q-2}} \end{align*} we have
\begin{align*}
    \E[\tr(X^TX - I_d)^{2q}]^\frac{1}{2q} \leq  (\varepsilon/\sqrt{e})
\end{align*}
which will imply
\begin{align*}
        \Pb \left( \| X^TX - I_d \| \ge \varepsilon \right) &\le \delta
    \end{align*}
when combined with the previous calculations.
    To ensure that the assumptions of Theorem \ref{prop:momestdecoupmatrix} are satisfied, we need to choose $q$ such that
    \begin{align*}pm \ge (\max \left\{ \frac{c_{\ref*{prop:momestdecoupmatrix}.2}\sqrt{e} q^2}{ \varepsilon}, c_{\ref*{prop:momestdecoupmatrix}.3} q^3 \right\})^{1+\frac{2}{q-2}} \end{align*}
Without loss of generality, we assume that $c_{\ref*{prop:momestdecoupmatrix}.2}>1$ and $c_{\ref*{prop:momestdecoupmatrix}.3}>1$.
In addition, we recall that we also need the requirement $q> \log (\frac{d}{\delta}) $ for one of the previous steps.

Now, we claim that it suffices to require $m \ge 8c_{\ref*{prop:momestdecoupmatrix}.1} \frac{d+\log(d/\varepsilon\delta)}{(\varepsilon/\sqrt{e})^2}$, and 
\begin{equation}\label{pmassum}
    pm \ge \max \left\{ c_{\ref*{prop:momestdecoupmatrix}.2}{e} 8^2e^3, c_{\ref*{prop:momestdecoupmatrix}.3} 8^3 e^3 \right\}(\max \left\{ \frac{c_{\ref*{prop:momestdecoupmatrix}.2}{e} (8\log(\frac{d}{\varepsilon \delta}))^2}{ \varepsilon}, c_{\ref*{prop:momestdecoupmatrix}.3} (8\log(\frac{d}{\varepsilon \delta}))^3 \right\})
\end{equation}
and choose \begin{align*}q=\lceil 2\log (\frac{d}{\varepsilon \delta} )\rceil+2\end{align*} to obtain the desired high probability bound. Since the quantity $\E[\tr(X^TX - I_d)^{2q}]$ remains unchanged if we only assume that subsets of the entries of $X$ of size $2q$ are independent instead of arbitrary subsets of the entries of $X$ being independent and we choose $q=\lceil 2\log (\frac{d}{\varepsilon \delta} )\rceil+2$, we only need $S$ to be an $O(\log(d/\varepsilon\delta))$-wise independent unscaled OSNAP.

Now we will check that this choice really works. First, by definition, we have $q> \log (d/\delta)$ and the lower bound on $m$ implies $q<m$, so we still have
\begin{align*}
        \Pb \left( \| X^TX - I_d \| \ge \sqrt{e} \E [ \tr (X^TX - I_d)^{2q}]^\frac{1}{2q} \right) &\le \delta
    \end{align*}
So it remains to check 
\begin{align*}pm \ge (\max \left\{ \frac{c_{\ref*{prop:momestdecoupmatrix}.2}{e} q^2}{ \varepsilon}, c_{\ref*{prop:momestdecoupmatrix}.3} q^3 \right\})^{\frac{q_1}{q_1-2}}\end{align*}
Note that $q=2\lceil \log (\frac{d}{\varepsilon \delta})\rceil+2<8\log (\frac{d}{\varepsilon \delta})$ because $d>10$, and therefore the above requirement is implied by
\begin{align*}pm \ge (\max \left\{ \frac{c_{\ref*{prop:momestdecoupmatrix}.2}{e} (8\log(\frac{d}{\varepsilon \delta}))^2}{ \varepsilon}, c_{\ref*{prop:momestdecoupmatrix}.3} (8\log(\frac{d}{\varepsilon \delta}))^3 \right\})^{1+\frac{2}{q-2}}  \end{align*}
Since $q-2 \ge 2\log (\frac{d}{\varepsilon \delta})$, we also claim that
\begin{align*}
    &(\max \left\{ \frac{c_{\ref*{prop:momestdecoupmatrix}.2}{e} (8\log(\frac{d}{\varepsilon \delta}))^2}{ \varepsilon}, c_{\ref*{prop:momestdecoupmatrix}.3} (8\log(\frac{d}{\varepsilon \delta}))^3 \right\})^{\frac{2}{q-2}} \\\le& (\max \left\{ \frac{c_{\ref*{prop:momestdecoupmatrix}.2}{e} (8\log(\frac{d}{\varepsilon \delta}))^2}{ \varepsilon}, c_{\ref*{prop:momestdecoupmatrix}.3} (8\log(\frac{d}{\varepsilon \delta}))^3 \right\})^{\frac{1}{\log (\frac{d}{\varepsilon \delta})}}
\end{align*}
will be bounded by an explicit constant, and this is where we need the extra factor $\frac{1}{\varepsilon}$ inside the $\log$.
We observe that
\begin{align*}
    &(\frac{ (\log(\frac{d}{\varepsilon \delta}))^2}{ \varepsilon})^{\frac{1}{\log (\frac{d}{\varepsilon \delta})}}
    \\ = & \exp(2\log\log(\frac{d}{\varepsilon \delta})/\log (\frac{d}{\varepsilon \delta}))\cdot (\frac{1}{\varepsilon})^{\frac{1}{\log (\frac{d}{\varepsilon \delta})}}
    \\ \le & e^2 \cdot (\frac{1}{\varepsilon})^{\frac{1}{\log (\frac{1}{\varepsilon})}}
    \\\le & e^3
\end{align*}
and
\begin{align*}
    &( (\log(\frac{d}{\varepsilon \delta}))^{\frac{3}{\log (\frac{d}{\varepsilon \delta})}}
    \\ = & \exp(3\log\log(\frac{d}{\varepsilon \delta})/\log (\frac{d}{\varepsilon \delta})) \le e^3
\end{align*}
where we use the fact that $\log\log(\frac{d}{\varepsilon \delta})>0$ when $d>10$.
Therefore, we have
\begin{align*}&(\max \left\{ \frac{c_{\ref*{prop:momestdecoupmatrix}.2}{e} (8\log(\frac{d}{\varepsilon \delta}))^2}{ \varepsilon}, c_{\ref*{prop:momestdecoupmatrix}.3} (8\log(\frac{d}{\varepsilon \delta}))^3 \right\})^{\frac{2}{q-2}} \\\le& (\max \left\{ \frac{c_{\ref*{prop:momestdecoupmatrix}.2}{e} (8\log(\frac{d}{\varepsilon \delta}))^2}{ \varepsilon}, c_{\ref*{prop:momestdecoupmatrix}.3} (8\log(\frac{d}{\varepsilon \delta}))^3 \right\})^{\frac{1}{\log (\frac{d}{\varepsilon \delta})}}
\\\le& \max \left\{ c_{\ref*{prop:momestdecoupmatrix}.2}{e} 8^2e^3, c_{\ref*{prop:momestdecoupmatrix}.3} 8^3 e^3 \right\} \end{align*}
Therefore it is enough to require,
\begin{align*}pm \ge \max \left\{ c_{\ref*{prop:momestdecoupmatrix}.2}{e} 8^2e^3, c_{\ref*{prop:momestdecoupmatrix}.3} 8^3 e^3 \right\}(\max \left\{ \frac{c_{\ref*{prop:momestdecoupmatrix}.2}{e} (8\log(\frac{d}{\varepsilon \delta}))^2}{ \varepsilon}, c_{\ref*{prop:momestdecoupmatrix}.3} (8\log(\frac{d}{\varepsilon \delta}))^3 \right\})  \end{align*}
which is guaranteed by our assumption (\ref{pmassum}) on $pm$.

    Moreover, recall that the definition of OSNAP requires $p \le 1$. Therefore, if the requirement (\ref{pmassum}) already forces $p>1$, then we cannot construct an OSNAP satisfying the requirement. This will generally not be an issue because lower bound in the requirement (\ref{pmassum}) grows much slower than $m$. However, when $m$ is small, this might happen because the constants can be large. For the completeness of the result, we mention that we can always use $p=1$ even when (\ref{pmassum}) requires $p>1$. To prove this claim, it is enough to find values of $m$ for which a dense $m \times n$ matrix populated with random signs satisfies the OSE property for a given $0 < \varepsilon, \delta < 1$, and this follows directly for $m \ge c_4 \cdot \frac{d+\log(1/\delta)}{\varepsilon^2}$ using \cite[Theorem 4.6.1]{vershynin2018high} for some constant $c_4$.
    
     \end{proof}
\subsection{Proving the decoupling lemma for OSNAP} \label{subsec:decouposnap}
Let $S$ have the fully independent unscaled OSNAP distribution as described in Definition \ref{def:osnap}. Then, recall that,
\begin{align*}
    S &= \sum_{l=1}^n \sum_{\gamma=1}^{pm} \xi_{l,\gamma} e_{\mu_{(l, \gamma)}} e_l ^T \\
    &=: \sum_{l=1}^n \sum_{\gamma=1}^{pm} Z_{l,\gamma} 
\end{align*}
where $\{ \xi_{l,\gamma} \}_{l \in [n], \gamma \in [s]}$ is a collection of  independent Rademacher random variables, $\{ \mu_{l,\gamma} \}_{l \in [n], \gamma \in [s]}$ is a collection of independent random variables such that each $\mu_{l,\gamma}$ is uniformly distributed in $[(m/s)(\gamma-1)+1:(m/s)\gamma]$ and $e_{\mu_{(l, \gamma)}}$ and $e_l$ represent basis vectors in $\R^m$ and $\R^n$ respectively. 

Recalling that our goal is to look at the moments of $(SU)^T(SU) - pm I_d$, we observe that,
\begin{align*}
    U^TS^TSU - pm\cdot I_d &= \left( \sum_{i=1}^m U^Ts_is_i^TU \right) - pm\cdot I_d \\
\end{align*}
where $s_i^T$ denotes the $i\textsuperscript{th}$ row of $S$. Then, letting $s_{ij}$ denote the entries of $S$, 
\begin{align*}
    U^TS^TSU - pm\cdot I_d &= \left( \sum_{i=1}^m U^T\left( \sum_{j=1}^n s_{ij}e_{j} \right) \left( \sum_{j'=1}^n s_{ij'}e_{j'}^T \right)U \right) - pm\cdot I_d \\
    &=  \sum_{i=1}^m \left( \sum_{j=1}^n s_{ij}u_{j} \right) \left( \sum_{j'=1}^n s_{ij'}u_{j'}^T \right)  - pm\cdot I_d \\
\end{align*}
where $u_j^T$ denotes the $j\textsuperscript{th}$ row of $U$. Separating the cases where $j=j'$ and $j \neq j'$, 
\begin{equation}\label{eq:diagoffdiag}
\begin{aligned}
    U^TS^TSU - pm\cdot I_d &=  \sum_{i=1}^m \sum_{j=1}^n s_{ij}^2 u_ju_j^T  - pm\cdot I_d + \sum_{i=1}^m \sum_{\substack{j,j' =1 \\ j \neq j'}}^n s_{ij}s_{ij'} u_ju_{j'}^T \\
    &=  \sum_{j=1}^n \left( \sum_{i=1}^m s_{ij}^2 \right) u_ju_j^T  - pm\cdot I_d + \sum_{i=1}^m \sum_{\substack{j,j' =1 \\ j \neq j'}}^n s_{ij}s_{ij'} u_ju_{j'}^T
\end{aligned}
\end{equation}
\begin{remark}
    In Equation \eqref{eq:diagoffdiag}, we decomposed the the embedding error $U^TS^TSU - pm\cdot I_d$ into two parts, the diagonal term
    \begin{align*}
        \sum_{j=1}^n \left( \sum_{i=1}^m s_{ij}^2 \right) u_ju_j^T  - pm\cdot I_d
    \end{align*}
    and the off-diagonal term
    \begin{align*}
        \sum_{i=1}^m \sum_{\substack{j,j' =1 \\ j \neq j'}}^n s_{ij}s_{ij'} u_ju_{j'}^T
    \end{align*}
    This decomposition is the key to understand why OSNAP model only needs $\tilde O(\frac{1}{\varepsilon})$ nonzero entries per column but the i.i.d. entries model might require $\tilde O(\frac{1}{\varepsilon^2})$ nonzero entries per column. In fact, as we will see very soon, the key difference between the OSNAP model and the i.i.d. entries model is that the diagonal term vanishes in OSNAP model but does not vanish in the i.i.d. entries model.
\end{remark}
By construction, $\sum_{i=1}^m s_{ij}^2 = pm$, so the diagonal term becomes
\begin{align*}
        \sum_{j=1}^n \left( \sum_{i=1}^m s_{ij}^2 \right) u_ju_j^T  - pm\cdot I_d=pm \left( \sum_{j=1}^n u_ju_j^T  - I_d \right)=0
    \end{align*}
and therefore we have
\begin{align*}
    U^TS^TSU - pm\cdot I_d &=  pm \left( \sum_{j=1}^n u_ju_j^T  - I_d \right) + \sum_{i=1}^m \sum_{\substack{j,j' =1 \\ j \neq j'}}^n s_{ij}s_{ij'} u_ju_{j'}^T \\
    &= \sum_{i=1}^m \sum_{\substack{j,j' =1 \\ j \neq j'}}^n s_{ij}s_{ij'} u_ju_{j'}^T
\end{align*}

To analyze the off-diagonal term, we use the standard technique of decoupling,

\begin{lemma}[Decoupling] \label{lem:decoup}
When $S$ has the fully independent unscaled OSNAP distribution, we have
\begin{align*}
    \E [ \tr (U^TS^TSU - pm\cdot I_d)^{2q} ] &= \E \left[ \tr \left( \sum_{i=1}^m \sum_{j,j' =1, j \neq j'}^n s_{ij}s_{ij'} u_ju_{j'}^T \right)^{2q} \right] \\
\end{align*}
Consequently, we have
\begin{align*}
    \E [ \tr (U^TS^TSU - pm\cdot I_d)^{2q} ] &\le \E_{S,S'} \left[ \tr \left(  2\paren*{(SU)^TS'U + (SU)^TS'U} \right)^{2q} \right]
\end{align*}
where $S'$ is an independent copy of $S$.
\end{lemma}

\begin{proof}
    Let $\mathbf{w}^T = (w_1, \ldots, w_n)$ be a vector of independent random variables such that $\Pb(w_i=1)=\Pb(w_i=0)=1/2$. Then, whenever $j\neq j'$, $\E[\mathbf{1}_{w_j \neq w_{j'}}] = 1/2$. So we have, 
    \begin{align*}
        \sum_{i=1}^m \sum_{j,j' =1, j \neq j'}^n s_{ij}s_{ij'} u_ju_{j'}^T = 2\E_{\mathbf{w}} \left[ \sum_{i=1}^m \sum_{j,j' =1}^n \mathbf{1}_{w_j \neq w_{j'}} s_{ij}s_{ij'} u_ju_{j'}^T \right]
    \end{align*}
    By Jensen's inequality \cite[Lemma 4.5, p.86]{kallenberg2021foundations}
    \begin{align*}
        \E \left[ \tr \left( \sum_{i=1}^m \sum_{j,j' =1, j \neq j'}^n s_{ij}s_{ij'} u_ju_{j'}^T \right)^{2q} \right] &= \E_S \left[ \tr \left( \E_{\mathbf{w}} \left[ 2 \sum_{i=1}^m \sum_{j,j' =1}^n  \mathbf{1}_{w_j \neq w_{j'}} s_{ij}s_{ij'} u_ju_{j'}^T \right] \right)^{2q} \right] \\
        &\le \E_S \E_{\mathbf{w}} \left[ \tr \left(   2 \sum_{i=1}^m \sum_{j,j' =1}^n  \mathbf{1}_{w_j \neq w_{j'}} s_{ij}s_{ij'} u_ju_{j'}^T  \right)^{2q} \right] \\
        &\le \E_{\mathbf{w}} \E_S \left[ \tr \left(   2 \sum_{i=1}^m \sum_{j,j' =1}^n  \mathbf{1}_{w_j \neq w_{j'}} s_{ij}s_{ij'} u_ju_{j'}^T  \right)^{2q} \right] \\
    \end{align*}
Fix $\mathbf{w}$, and let $J = \{ j \in [n] | w_j=1 \}$. Then,
\begin{align*}
    &\E_S \left[ \tr \left(   2 \sum_{i=1}^m \sum_{j,j' =1}^n  \mathbf{1}_{w_j \neq w_{j'}} s_{ij}s_{ij'} u_ju_{j'}^T  \right)^{2q} \right] \\ &= \E_S \left[ \tr \left(   2 \sum_{i=1}^m \left( \sum_{j\in J, j' \in J^C}   s_{ij}s_{ij'} u_ju_{j'}^T + \sum_{j\in J^C, j' \in J}   s_{ij}s_{ij'} u_ju_{j'}^T \right) \right)^{2q} \right] \\
\end{align*}
The key observation for decoupling is that the collection of random variables $ \mathcal{S}_J := \{s_{ij} \}_{i \in [m], j\in J}$ is independent of the collection $\mathcal{S}_{J^C} := \{s_{ij} \}_{i \in [m], j\in J^C}$, so the above expectation does not change if the collection $ \mathcal{S}_J$ is replaced by an independent copy $\mathcal{S}'_J :=\{s'_{ij} \}_{i \in [m], j\in J}$ thought of as coming from the entries of an independent copy of $S$, say, $S'$, i.e.,
\begin{align*}
    &\E_{\mathcal{S}_J, \mathcal{S}_{J^C} } \Biggl[ \tr \Biggl(   2 \sum_{i=1}^m \Biggl( \sum_{j\in J, j' \in J^C}   s_{ij}s_{ij'} u_ju_{j'}^T + \sum_{j\in J^C, j' \in J}   s_{ij}s_{ij'} u_ju_{j'}^T \Biggr) \Biggr)^{2q} \Biggr] \\
    &= \E_{\mathcal{S}'_J, \mathcal{S}_{J^C}} \Biggl[ \tr \Biggl(   2 \sum_{i=1}^m \Biggl( \sum_{j\in J, j' \in J^C}   s'_{ij}s_{ij'} u_ju_{j'}^T + \sum_{j\in J^C, j' \in J}   s_{ij}s'_{ij'} u_ju_{j'}^T \Biggr) \Biggr)^{2q} \Biggr] \\
    &= \E_{\mathcal{S}'_J, \mathcal{S}_{J^C}} \Biggl[ \tr \Biggl(   2 \sum_{i=1}^m \Biggl( \sum_{j\in J, j' \in J^C}   s'_{ij}s_{ij'} u_ju_{j'}^T + \sum_{j\in J^C, j' \in J}   s_{ij}s'_{ij'} u_ju_{j'}^T  \\
    &\quad + \E_{\mathcal{S}_{J}} \Biggl[ \sum_{j\in J, j' \in J}   s'_{ij}s_{ij'} u_ju_{j'}^T \Biggr] + \E_{\mathcal{S}'_{J^C}, \mathcal{S}_{J}} \Biggl[ \sum_{j\in J^C, j' \in [n]}   s'_{ij}s_{ij'} u_ju_{j'}^T \Biggr] \\
    &\quad + \E_{\mathcal{S}'_{J^C}} \Biggl[ \sum_{j\in J^C, j' \in J^C}   s_{ij}s'_{ij'} u_ju_{j'}^T \Biggr] + \E_{\mathcal{S}_J, \mathcal{S}'_{J^C}} \Biggl[ \sum_{j\in J, j' \in [n]}   s_{ij}s'_{ij'} u_ju_{j'}^T \Biggr] \Biggr) \Biggr)^{2q} \Biggr] \\
    &\le \E_{S,S'} \Biggl[ \tr \Biggl( 2 \sum_{i=1}^m  \sum_{j,j'=1}^n \Biggl( s_{ij}s'_{ij'} u_ju_{j'}^T + s'_{ij}s_{ij'} u_ju_{j'}^T \Biggr) \Biggr)^{2q} \Biggr] \\
    &= \E_{S,S'} \left[ \tr \left( 2\paren*{(SU)^TS'U + (S'U)^TSU} \right)^{2q} \right]
\end{align*}
where in the second equality the expectations we add are all $0$, and the inequality step follows by taking the expectations outside using Jensen's inequality. 
\end{proof}

\subsection{Controlling the trace moments of the embedding error for OSNAP} \label{subsec:osnaptracemom}
In this section we give the full proof of Lemma \ref{prop:momestdecoupmatrix}.

\begin{lemma}[Trace Moments of Embedding Error for OSNAP] \label{prop:momestdecoupmatrix}
Let $S$ be an $m \times n$ matrix distributed according to the fully independent unscaled OSNAP distribution with parameter $p \le 1$. Let $U$ be an arbitrary $n \times d$ deterministic matrix such that $U^TU=I$. Define $X = \frac{1}{\sqrt{pm}}SU$.  Then, there exist constants $c_{\ref*{prop:momestdecoupmatrix}.1}, c_{\ref*{prop:momestdecoupmatrix}.2}, c_{\ref*{prop:momestdecoupmatrix}.3} > 0$ such that for $q \in \N$ satisfying $2 \le q \le m$, when $m \geq c_{\ref*{prop:momestdecoupmatrix}.1} \frac{d+q}{\varepsilon^2}$ and $ pm \ge (\max\{\frac{c_{\ref*{prop:momestdecoupmatrix}.2} q^{2}}{\varepsilon}, c_{\ref*{prop:momestdecoupmatrix}.3} q^3\})^{1+\frac{2}{q-2}} $, we have
\begin{align*}
    \E[\tr(X^TX - I_d)^{2q}]^\frac{1}{2q} \leq  \varepsilon
\end{align*} 

\end{lemma}

\begin{proof}

In Section \ref{subsec:decouposnap}, we saw that when $S$ has the unscaled OSNAP distribution \begin{align*}\E[\tr((SU)^T(SU)-pm\cdot I_d)^{2q}]^\frac{1}{2q} \le 2 \E[\tr(\Gamma(S_1,S_2))^{2q}]^\frac{1}{2q}\end{align*} where $S_1$ and $S_2$ are independent random matrices with the unscaled OSNAP distribution and $\Gamma(M_1,M_2)=(M_1U)^T(M_2U)+(M_2U)^T(M_1U)$. Let $\Gamma(t)$ be as defined in Lemma $\ref{lem:diffineq}$. Then, to show the statement of the theorem, it is enough to show that $\E[\tr(\Gamma(1))^{2q}]^\frac{1}{2q} \le pm\varepsilon/2$. Now, by Lemma \ref{cor:indgaussianmom}, we know that, $$\E[\tr(\Gamma(0))^{2q}]^\frac{1}{2q} = \E[\tr(\Gamma(G_1, G_2))^{2q}]^\frac{1}{2q} \le c_{\ref*{cor:indgaussianmom}}p\sqrt{\max\{d, q\}}\sqrt{\max\{m, q\}}$$ 

We want to find conditions for which $\E[\tr(\Gamma(0))^{2q}]^\frac{1}{2q} \le pm\varepsilon/4$, for which it is enough to ensure $c_{\ref*{cor:indgaussianmom}}p\sqrt{\max\{d, q\}}\sqrt{\max\{m, q\}} \le pm\varepsilon/4$. Clearly, this can only happen when $q \le m$, and in this case the inequality holds when $m \ge \frac{c(d+q)}{\varepsilon^2}$. Thus, it suffices to show
\begin{align*}
    (\E[\tr \Gamma(1)^{2q}])^{\frac{1}{2q}}-(\E[\tr \Gamma(0)^{2q}])^{\frac{1}{2q}} \le \frac{1}{4} pm \varepsilon
\end{align*}

By Lemma \ref{lem:diffineq}, if $q$ satisfies $q \ge 2$, we have
\begin{align*}
\frac{d}{dt} \E[\tr \Gamma(t)^{2q}] \le & \max \limits_{4 \le k \le 2q} (c_5q)^k ((pm)^{\frac{1}{q}}\sqrt{\max\{pd,q\}})^{\frac{qk-2q}{q-1}}(\E[f(S_1(t),S_2(t))])^{1-\frac{k-2}{2q-2}}
\end{align*}

We use different calculations for two different cases.

In the first case, we consider $pd<q$.

By convexity, the maximum must be attained on one of the two end points. Also, we observe that when $k=2q$, we have $\frac{k-2}{2q-2}=1$. Therefore, we have

\begin{align*}
\frac{d}{dt} \E[\tr \Gamma(t)^{2q}] \le&\max \limits_{4 \le k \le 2q} (c_5q)^k ((pm)^{\frac{1}{q}}\sqrt{\max\{pd,q\}})^{\frac{qk-2q}{q-1}}(\E \tr((\Gamma(t))^{2q}))^{1-\frac{k-2}{2q-2}}\\ \le & \max \{(c_5q)^4 ((pm)^{\frac{1}{q}}\sqrt{\max\{pd,q\}})^{\frac{q4-2q}{q-1}}(\E \tr((\Gamma(t))^{2q}))^{1-\frac{2}{2q-2}}, \\& (c_5q)^{2q} ((pm)^{\frac{1}{q}}\sqrt{\max\{pd,q\}})^{\frac{q(2q)-2q}{q-1}}(\E \tr((\Gamma(t))^{2q}))^{1-\frac{2q-2}{2q-2}}\}
\\ = & (c_5q)^4 ((pm)^{\frac{1}{q}}\sqrt{\max\{pd,q\}})^{\frac{2q}{q-1}} \max \{(\E \tr((\Gamma(t))^{2q}))^{1-\frac{2}{2q-2}}, \\& (c_5q)^{2q-4} ((pm)^{\frac{1}{q}}\sqrt{\max\{pd,q\}})^{\frac{2q^2-4q}{q-1}}(\E \tr((\Gamma(t))^{2q}))^{1-\frac{2q-2}{2q-2}}\}
\\ = & (c_5q)^4 ((pm)^{\frac{1}{q}}\sqrt{\max\{pd,q\}})^{\frac{2q}{q-1}} \max \{(\E \tr((\Gamma(t))^{2q}))^{1-\frac{2}{2q-2}}, \\& (c_5q)^{2q-4} ((pm)^{\frac{1}{q}}\sqrt{\max\{pd,q\}})^{\frac{2q^2-4q}{q-1}}\}
\\ = & (c_5q)^4 ((pm)^{\frac{1}{q}}\sqrt{\max\{pd,q\}})^{\frac{2q}{q-1}} \max \{(\E \tr((\Gamma(t))^{2q}))^{1-\frac{1}{q-1}}, \\& (c_5q)^{2q-4} ((pm)^{\frac{1}{q}}\sqrt{\max\{pd,q\}})^{\frac{2q^2-4q}{q-1}}\}
\end{align*}

We will use Lemma 6.6 in \cite{brailovskaya2022universality} with $\alpha=\frac{1}{q-1}$. In this case, we have ${1-\alpha}={\frac{q-2}{q-1}}$, and $\frac{\alpha}{1-\alpha}=\frac{1}{q-2}$. Also, we want $K^{1-\alpha}=(c_5q)^{2q-4} ((pm)^{\frac{1}{q}}\sqrt{\max\{pd,q\}})^{\frac{2q(2q)-8q}{2q-2}}$, so we need $K=((c_5q)^{2q-4} ((pm)^{\frac{1}{q}}\sqrt{\max\{pd,q\}})^{\frac{2q^2-4q}{q-1}})^{\frac{1}{1-\alpha}}$.

By Lemma 6.6 in \cite{brailovskaya2022universality}, we have
\begin{align*}
    &(\E[\tr \Gamma(S_1,S_2)^{2q}])^{\alpha}-(\E[\tr \Gamma(G_1,G_2)^{2q}])^{\alpha} \\\le& (c_5q)^4 ((pm)^{\frac{1}{q}}\sqrt{\max\{pd,q\}})^{\frac{2q}{q-1}} \cdot \frac{1}{q-1}+ ((c_5q)^{2q-4} ((pm)^{\frac{1}{q}}\sqrt{\max\{pd,q\}})^{\frac{2q^2-4q}{q-1}})^{\frac{\alpha}{1-\alpha}}
    \\\le& (c_5q)^4 ((pm)^{\frac{1}{q}}\sqrt{q})^{\frac{2q}{q-1}} \cdot \frac{1}{q-1}+ ((c_5q)^{2q-4} ((pm)^{\frac{1}{q}}\sqrt{q})^{\frac{2q^2-4q}{q-1}})^{\frac{1}{q-2}}
    \\=& (c_5q)^4 (pm)^{\frac{2}{q-1}} q (q)^{\frac{1}{q-1}} \cdot \frac{1}{q-1}+ (c_5q)^{2} (pm)^{\frac{2}{q-1}} q^\frac{q}{q-1}
    \\ \le & c_6 (pm)^{\frac{2}{q-1}} q^4
\end{align*}

Therefore, we have
\begin{align*}
    (\E[\tr \Gamma(S_1,S_2)^{2q}])^{\frac{1}{q-1}}-(\E[\tr \Gamma(G_1,G_2)^{2q}])^{\frac{1}{q-1}} \le c_6 (pm)^{\frac{2}{q-1}} q^4
\end{align*}

By Taylor expansion, we have
\begin{align*}
    (a+b)^{\frac{2q}{q-1}} \ge a^{\frac{2q}{q-1}}+b^{\frac{2q}{q-1}}
\end{align*}

Using change of variable $c=a^{\frac{2q}{q-1}}+b^{\frac{2q}{q-1}}$ and $d=a^{\frac{2q}{q-1}}$, we have
\begin{align*}
    c^{\frac{q-1}{2q}}-d^{\frac{q-1}{2q}} \le (c-d)^{\frac{q-1}{2q}}
\end{align*}

Therefore, we have
\begin{align*}
    (\E[\tr \Gamma(S_1,S_2)^{2q}])^{\frac{1}{2q}}-(\E[\tr \Gamma(G_1,G_2)^{2q}])^{\frac{1}{2q}} \le (c_6 (pm)^{\frac{2}{q-1}} q^4)^{\frac{q-1}{2q}} \le c_7 (pm)^{\frac{1}{q}} q^2
\end{align*}

By previous analysis, the desired requirement $(\E[\tr \Gamma(S_1,S_2)^{2q}])^{\frac{1}{2q}}-(\E[\tr \Gamma(G_1,G_2)^{2q}])^{\frac{1}{2q}} \le \frac{1}{4} pm \varepsilon$ will be satisfied when
$c_7 (pm)^{\frac{1}{q}}  q^2 \le pm \varepsilon$
which means
\begin{align*}
    pm \ge c_7 (pm)^{\frac{1}{q}}  \frac{q^2}{\varepsilon}
\end{align*}

In the second case, we consider $pd>q$. In this case, we have some $pd$ factors in the bound for $\frac{d}{dt} \E[\tr \Gamma(t)^{2q}]$. We will eventually get rid of these $pd$ factors by a replacing it by some powers of $\E[\tr \Gamma(t)^{2q}]$ and some extra factors. To illustrate this idea, we start from following lemma to calculate $\E [\tr \Gamma(t)^2]$.

\begin{lemma}[Second Moments]\label{2moment}
    Let $M$ and $N$ be $m \times d$ independent random matrices such that the entries of $M$ (and $N$) are uncorrelated and have variance $p$. Then, 
    \begin{itemize}
        \item $\E [\tr M^T N M^T N ] = p^2m$.
        \item $\E [\tr M^T N N^T M ] = p^2md$
    \end{itemize}
\end{lemma}
\begin{proof}
    We have, 
    \begin{align*}
        \E [\tr M^T N M^T N ] &= \frac{1}{d} \sum_{i=1}^d \sum_{j,k,l} \E [m_{ji}n_{jk}m_{lk}n_{li}  ] 
    \end{align*}
    By uncorrelatedness of entries, the expectation is non-zero only when $j=l$ and $k=i$. So,
    \begin{align*}
        \E [\tr M^T N M^T N ] &= \frac{1}{d} \sum_{i=1}^d \sum_{j=1}^m \E [m_{ji}^2n_{ji}^2 ] \\
        &= p^2m
    \end{align*}
    Similarly,
    \begin{align*}
        \E [\tr M^T N N^T M ] &= \frac{1}{d} \sum_{i=1}^d \sum_{j,k,l} \E [m_{ji}n_{jk}n_{lk}m_{li}  ] \\
        &= \frac{1}{d} \sum_{i=1}^d \sum_{j,k} \E [m_{ji}^2n_{jk}^2  ] \\
        &= p^2md
    \end{align*}
\end{proof}

By Lemma \ref{2moment}, we have $\E [\tr \Gamma(t)^2] = 2p^2md + 2p^2m \ge 2p^2md$. By Hölder's inequality,
\[ \sqrt{2p^2md} \le  \E [\tr \Gamma(t)^2]^\frac{1}{2} \le \E [\tr \Gamma(t)^{2q}]^\frac{1}{2q}  \]

In this case, we can replace the extra $pd$ factors and obtain
\begin{align*}
\frac{d}{dt} \E[\tr \Gamma(t)^{2q}] \le & \max \limits_{4 \le k \le 2q} (c_5q)^k ((pm)^{\frac{1}{q}} \sqrt{\max\{pd,q\}})^{\frac{2qk-4q}{2q-2}}(\E \tr((\Gamma(t))^{2q}))^{1-\frac{k-2}{2q-2}}
\\ \le  & \max \limits_{4 \le k \le 2q} (c_5q)^k ((pm)^{\frac{1}{q}} \sqrt{pd})^{\frac{2qk-4q}{2q-2}}(\E \tr((\Gamma(t))^{2q}))^{1-\frac{k-2}{2q-2}}
\\ \le  & \max \limits_{4 \le k \le 2q} (c_5q)^k ((pm)^{\frac{1}{q}} (2p^2md)^{1/4}(\frac{1}{2}\frac{d}{m})^{1/4})^{\frac{2qk-4q}{2q-2}}(\E \tr((\Gamma(t))^{2q}))^{1-\frac{k-2}{2q-2}}
\\ \le  & \max \limits_{4 \le k \le 2q} (c_5q)^k ((pm)^{\frac{1}{q}} ( \E [\tr \Gamma(t)^{2q}]^\frac{1}{2q})^{1/2}(\frac{1}{2}\frac{d}{m})^{1/4})^{\frac{2qk-4q}{2q-2}}(\E \tr((\Gamma(t))^{2q}))^{1-\frac{k-2}{2q-2}}
\\ =  & \max \limits_{4 \le k \le 2q} (c_5q)^k (pm)^{\frac{k-2}{q-1}}  ( \E [\tr \Gamma(t)^{2q}])^{\frac{k-2}{4(q-1)}} (\frac{1}{2}\frac{d}{m})^{\frac{qk-2q}{4(q-1)}}(\E \tr((\Gamma(t))^{2q}))^{1-\frac{k-2}{2(q-1)}}
\\ =  & \max \limits_{4 \le k \le 2q} (c_5q)^k (pm)^{\frac{k-2}{q-1}}  (\frac{1}{2}\frac{d}{m})^{\frac{qk-2q}{4(q-1)}}(\E \tr((\Gamma(t))^{2q}))^{1-\frac{2(k-2)}{4(q-1)}+\frac{k-2}{4(q-1)}}
\\ =  & \max \limits_{4 \le k \le 2q} (c_5q)^k (pm)^{\frac{k-2}{q-1}}  (\frac{1}{2}\frac{d}{m})^{\frac{qk-2q}{4(q-1)}}(\E \tr((\Gamma(t))^{2q}))^{1-\frac{(k-2)}{4(q-1)}}
\end{align*}

If we want $1-\frac{(k-2)}{4(q-1)}=0$, then we need $(k-2)=4(q-1)$. Therefore, we can choose $k=4(q-1)+2=4q-2>2q$ when $q>1$.

Using the convexity of the function $a^xb^{1-x}$ in $x$ as before, we have
\begin{align*}
\frac{d}{dt} \E[\tr \Gamma(t)^{2q}] \le & \max \limits_{4 \le k \le 2q} (c_5q)^k (pm)^{\frac{k-2}{q-1}}  (\frac{1}{2}\frac{d}{m})^{\frac{qk-2q}{4(q-1)}}(\E \tr((\Gamma(t))^{2q}))^{1-\frac{(k-2)}{4(q-1)}}
\\  \le   & \max \limits_{4 \le k \le (4q-2)} (c_5q)^k (pm)^{\frac{k-2}{q-1}} (\frac{1}{2}\frac{d}{m})^{\frac{qk-2q}{4(q-1)}}(\E \tr((\Gamma(t))^{2q}))^{1-\frac{(k-2)}{4(q-1)}}
\\ \le & \max\{(c_5q)^4 (pm)^{\frac{4-2}{q-1}}  (\frac{1}{2}\frac{d}{m})^{\frac{q4-2q}{4(q-1)}}(\E \tr((\Gamma(t))^{2q}))^{1-\frac{(4-2)}{4(q-1)}},\\&(c_5q)^{4q-2} (pm)^{\frac{4q-2-2}{q-1}}  (\frac{1}{2}\frac{d}{m})^{\frac{q(4q-2)-2q}{4(q-1)}}\}
\\ = & (c_5q)^4 (pm)^{\frac{2}{q-1}} (\frac{1}{2}\frac{d}{m})^{\frac{q4-2q}{4(q-1)}} \max\{(\E \tr((\Gamma(t))^{2q}))^{1-\frac{(4-2)}{4(q-1)}},\\&(c_5q)^{4q-2-4} (pm)^{\frac{4q-6}{q-1}} (\frac{1}{2}\frac{d}{m})^{\frac{q(4q-2)-2q}{4(q-1)}-\frac{q4-2q}{4(q-1)}}\}
\\ = & (c_5q)^4 (pm)^{\frac{2}{q-1}}  (\frac{1}{2}\frac{d}{m})^{\frac{q}{2(q-1)}} \max\{(\E \tr((\Gamma(t))^{2q}))^{1-\frac{1}{2(q-1)}},\\&(c_5q)^{4q-6} (pm)^{\frac{4q-6}{q-1}} (\frac{1}{2}\frac{d}{m})^{\frac{q(4q-6)}{4(q-1)}}\}
\end{align*}
by convexity.

At this point, we will use Lemma 6.6 in \cite{brailovskaya2022universality} again with $\alpha=\frac{1}{2(q-1)}$. Therefore, we have $1-\alpha=\frac{2q-3}{2(q-1)}$. Therefore, we have $\frac{\alpha}{1-\alpha}=\frac{1}{2q-3}$. So we need to use Lemma 6.6 in \cite{brailovskaya2022universality} with
\begin{align*}
    K=((c_5q)^{4q-6} (pm)^{\frac{4q-6}{q-1}} (\frac{1}{2}\frac{d}{m})^{\frac{q(4q-6)}{4(q-1)}})^{\frac{1}{1-\alpha}}
\end{align*}
in the notation there.

By Lemma 6.6 in \cite{brailovskaya2022universality}, we have
\begin{align*}
    &(\E[\tr \Gamma(S_1,S_2)^{2q}])^{\alpha}-(\E[\tr \Gamma(G_1,G_2)^{2q}])^{\alpha} \\\le& (c_5q)^4 (pm)^{\frac{2}{q-1}} (\frac{1}{2}\frac{d}{m})^{\frac{q}{2(q-1)}} \frac{1}{2(q-1)}+((c_5q)^{4q-6} (pm)^{\frac{4q-6}{q-1}} (\frac{1}{2}\frac{d}{m})^{\frac{q(4q-6)}{4(q-1)}})^{\frac{\alpha}{1-\alpha}}
    \\ \le & c_{10} q^{3} (pm)^{\frac{2}{q-1}} (\frac{d}{m})^{\frac{q}{2(q-1)}}+((c_5q)^{4q-6} (pm)^{\frac{4q-6}{q-1}}(\frac{1}{2}\frac{d}{m})^{\frac{q(4q-6)}{4(q-1)}})^{\frac{1}{2q-3}}
    \\ \le & c_{10} q^{3} (pm)^{\frac{2}{q-1}} (\frac{d}{m})^{\frac{q}{2(q-1)}}+(c_5q)^{2} (pm)^{\frac{2}{q-1}} (\frac{1}{2}\frac{d}{m})^{\frac{2q}{4(q-1)}}
    \\ \le & c_{10} q^{3} (pm)^{\frac{2}{q-1}} (\frac{d}{m})^{\frac{q}{2(q-1)}}+(c_5q)^{2} (pm)^{\frac{2}{q-1}} (\frac{1}{2}\frac{d}{m})^{\frac{q}{2(q-1)}}
    \\ \le & c_{10} q^{3} (pm)^{\frac{2}{q-1}} (\frac{d}{m})^{\frac{q}{2(q-1)}}+c_{11}q^{2} (pm)^{\frac{2}{q-1}}(\frac{d}{m})^{\frac{q}{2(q-1)}}
    \\ \le & c_{12} q^{3} (pm)^{\frac{2}{q-1}} (\frac{d}{m})^{\frac{q}{2(q-1)}}
\end{align*}

In summary, we have
\begin{align*}
    &(\E[\tr \Gamma(S_1,S_2)^{2q}])^{\frac{1}{2(q-1)}}-(\E[\tr \Gamma(G_1,G_2)^{2q}])^{\frac{1}{2(q-1)}} \\\le& c_{12} q^{3} (pm)^{\frac{2}{q-1}} (\frac{d}{m})^{\frac{q}{2(q-1)}}
\end{align*}

Therefore, we have
\begin{align*}
    &(\E[\tr \Gamma(S_1,S_2)^{2q}])^{\frac{1}{2q}}-(\E[\tr \Gamma(G_1,G_2)^{2q}])^{\frac{1}{2q}} \\\le& (c_{12} q^{3} (pm)^{\frac{2}{q-1}} (\frac{d}{m})^{\frac{q}{2(q-1)}})^{\frac{2(q-1)}{(2q)}}
    \\\le& c_{13} q^{\frac{3(q-1)}{q}} (pm)^{\frac{2}{q}} (\frac{d}{m})^{\frac{1}{2}}
    \\\le& c_{13} q^3 (pm)^{\frac{2}{q}} (\frac{d}{m})^{\frac{1}{2}}
\end{align*}
where we use the fact that
\begin{align*}
    (a+b)^{\frac{(2q)}{2(q-1)}} \ge a^{\frac{(2q)}{2(q-1)}}+b^{\frac{(2q)}{2(q-1)}}
\end{align*}
by Taylor expansion.

We need $(\E[\tr \Gamma(S_1,S_2)^{2q}])^{\frac{1}{2q}}-(\E[\tr \Gamma(G_1,G_2)^{2q}])^{\frac{1}{2q}} \le pm \varepsilon/4$, so the requirement is
\begin{align*}
    c_{13} q^3 (pm)^{\frac{2}{q}}(\frac{d}{m})^{\frac{1}{2}} \le \frac{pm \varepsilon}{4}
\end{align*}

Since we have $m \ge \frac{c_{14}d}{\varepsilon^2}$ for some constant $c_{14}$, we have $\varepsilon \ge \sqrt{\frac{c_{14}d}{m}}$, so it suffices to require
\begin{align*}
    c_{13} q^3 (pm)^{\frac{2}{q}}(\frac{d}{m})^{\frac{1}{2}} \le \frac{pm \sqrt{\frac{c_{14}d}{m}}}{4}
\end{align*}

which means
\begin{align*}
    pm \ge c_{15} (pm)^{\frac{2}{q}} {q^3}
\end{align*}

Combining the analysis for the two cases, it suffices to require
\begin{align*}pm \ge (pm)^{\frac{2}{q}}\max\{\frac{c_{16} q^{2}}{\varepsilon}, c_{17} q^3\}
\end{align*}
for some constants $c_{16}>0$ and $c_{17}>0$.

This requirement is equivalent to
\begin{align*}pm \ge (\max\{\frac{c_{16} q^{2}}{\varepsilon}, c_{17} q^3\})^{\frac{1}{1-2/q}}
\end{align*}

\end{proof}

\subsection{Differential inequality for the derivative of the interpolant}\label{subsec:diffineq}
In this section we give the full proof of Lemma \ref{lem:diffineq}.

\begin{lemma}[Differential Inequality]\label{lem:diffineq}
 Let $S_1$ and $S_2$ be independent random matrices such that either both $S_1$ and $S_2$ have the fully independent unscaled OSNAP distribution with parameter $p$ or both $S_1$ and $S_2$ have the unscaled OSE-IE distribution with parameter $p$. Let $G_1$ and $G_2$ be independent random matrices with i.i.d. Guassian entries each with variance $p$, and define the interpolated random matrices,
\begin{align}
\begin{split} 
    S_1(t) = \sqrt{t}S_1 + \sqrt{1-t}G_1 \\
    S_2(t) = \sqrt{t}S_2 + \sqrt{1-t}G_2
\end{split}
\end{align}
Let $f(M_1,M_2)=\tr(((M_1U)^T(M_2U)+(M_2U)^T(M_1U))^{2q})$.
Then, there exists a constant $c_{\ref*{lem:diffineq}}$ such that, for any $q \ge 2$, we have
\begin{align*}
\frac{d}{dt} \E[f(S_1(t),S_2(t))] \le& \max \limits_{4 \le k \le 2q} (c_{\ref*{lem:diffineq}}q)^k ((pm)^{\frac{1}{q}}\sqrt{\max\{pd,q\}})^{\frac{qk-2q}{q-1}}(\E[f(S_1(t),S_2(t))])^{1-\frac{k-2}{2q-2}}
\end{align*}
\end{lemma}
\begin{remark}
The proof of Lemma \ref{lem:diffineq} relies on a technical trace inequality, Lemma \ref{lem:decouptraceineq}. To illustrate the main idea, we present the proof of Lemma \ref{lem:diffineq} using Lemma \ref{lem:decouptraceineq} here first, and then prove Lemma \ref{lem:decouptraceineq} in the next section.
\end{remark}
\begin{proof}[Proof of Lemma \ref{lem:diffineq}]

For convenience, let $\Gamma(M_1,M_2)=(M_1U)^T(M_2U)+(M_2U)^T(M_1U)$ and $\Gamma(t)=\Gamma(S_1(t),S_2(t))$.

Fix $M_2$ and view $f_{1,M_2}(M_1)=\tr(\Gamma(M_1,M_2)^{2q})$ as a function of $M_1$. We shall first obtain an expression for $\frac{d}{dt_1}\E_{S_1(t_1)}[f(S_1(t_1),M_2)]$. For this purpose, we consider the construction from \cite{brailovskaya2022universality}. Recall that $S_1$ can be written in the form $\sum_{(l,\gamma) \in \Xi} Z_{(l,\gamma)}$ where $\Xi=[n] \times [pm]$ and $Z_{(l,\gamma)}=\xi_{(l,\gamma)} e_{\mu_{(l, \gamma)}} e_l ^\top$ (see definition \ref{def:osnap}). We shall assume that the $Z_{(l,\gamma)}$ are independent, with the additional remark that since all quantities involved are moments of order at most $4q$, the same calculations hold even if the $Z_{(l,\gamma)}$ are $4q$-wise independent.

Let $k\in\mathbb{N}$ and $\pi\in\mathrm{P}([k])$ be a partition of $[k]$. We construct the following family of random elements
\begin{align*}
\{Z_{(l,\gamma),j|\pi}\}_{(l,\gamma)\in \Xi, j \in [k]}
\end{align*}
that is independent from the sigma-algebra $\sigma(S_1,S_2,G_1,G_2)$ 
with the
following properties:
\begin{enumerate}[1.]
\item Each $\{Z_{(l,\gamma),j|\pi}\}_{(l,\gamma)\in \Xi}$ has the same distribution as
$\{Z_{(l,\gamma)}\}_{(l,\gamma) \in \Xi}$.
\item $\{Z_{(l,\gamma),j|\pi}\}_{(l,\gamma)\in \Xi}=\{Z_{(l,\gamma),j'|\pi}\}_{(l,\gamma)\in \Xi}$ for indices $j,j'$
that belong to
the same element of $\pi$.
\item $\{Z_{(l,\gamma),j|\pi}\}_{(l,\gamma)\in \Xi}$ are independent for indices $j$ that belong
to distinct elements of $\pi$.
\end{enumerate}
More precisely, we construct these random elements in the following way. First, we choose a partition $\pi\in\mathrm{P}([k])$. Without loss of generality, consider the simple case where $\pi=\{A_1,A_2\}$. Then we construct two i.i.d. copies of $\{Z_{(l,\gamma)}\}_{(l,\gamma) \in \Xi}$, and call them $(Z_{(l,\gamma),A_1})_{(l,\gamma) \in \Xi}$ and $(Z_{(l,\gamma),A_2})_{(l,\gamma) \in \Xi}$. For any $j \in A_1$, we just set $\{Z_{(l,\gamma),j|\pi}\}_{(l,\gamma)\in \Xi}=\{Z_{(l,\gamma),A_1}\}_{(l,\gamma) \in \Xi}$. Similarly, for any $j \in A_2$, we just set $\{Z_{(l,\gamma),j|\pi}\}_{(l,\gamma)\in \Xi}=\{Z_{(l,\gamma),A_2}\}_{(l,\gamma) \in \Xi}$.

In addition, we assume that the random elements
\begin{align*}
\{Z_{(l,\gamma),j|\pi}\}_{(l,\gamma)\in \Xi, j \in [k]}
\end{align*}
for different $k$ and $\pi$ are mutually independent from each other.

Using these random matrices, we can compute an expression for $\frac{d}{dt_1}\E_{S_1(t_1)}[f(S_1(t_1),M_2)]$ using the following lemma,

\begin{lemma}[Corollary 6.1, \cite{brailovskaya2022universality}] \label{lem:derexpansion}
    For any polynomial $\phi:M_{m\times d}(\mathbb{R})\to\mathbb{R}$, we have
\begin{align*}
	&\frac{d}{dt}\E[\phi(S_1(t))] \\=& 
	\frac{1}{2}\sum_{k=4}^\infty
	\frac{t^{\frac{k}{2}-1}}{(k-1)!}
	\sum_{\pi\in\mathrm{P}([k])}
	(-1)^{|\pi|-1}(|\pi|-1)!\,
	\E\Bigg[ \sum_{(l,\gamma) \in \Xi}\partial_{Z_{(l,\gamma)1|\pi}}\cdots\partial_{Z_{(l,\gamma)k|\pi}}\phi(S_1(t))
	\Bigg],
\end{align*}
where $\partial_Z\phi$ denotes the directional derivative of
$\phi$ in the direction $Z\in M_{m\times d}(\mathbb{C})$.
\end{lemma}

We remark that Lemma \ref{lem:derexpansion} is proved in \cite{brailovskaya2022universality} for $S_1(t)$ in the space of $d \times d$ self-adjoint matrices with complex entries, but the result can be seen to hold for arbitrary matrices. Lemma \ref{lem:derexpansion} relies on the fact that the entries of $G_1$ have the same mean and covariance as the entries of $S_1$ for the vanishing of the terms corresponding to $k \le 2$ in the above expansion. Terms corresponding to $k=3$ vanish due to our random variables being symmetric, so we start the series with $k=4$.

\begin{proof}[Proof of Lemma \ref{lem:derexpansion}]
Following the proof of Corollary 6.1 in \cite{brailovskaya2022universality}, we have
\begin{multline*}
	\frac{d}{dt}\E[\phi(X(t))] =
	\frac{1}{2}\sum_{k=3}^\infty
	\frac{t^{\frac{k}{2}-1}}{(k-1)!}\times\mbox{}\\
	\sum_{(l,\gamma) \in \Xi}
	\sum_{\substack{(u_j,v_j)\in[m] \times [d]\\
	j=1,\ldots,k}}
	\kappa((Z_{(l,\gamma)})_{u_1v_1},\ldots,(Z_{(l,\gamma)})_{u_kv_k})\,
	\E\bigg[
	\frac{\partial^k\phi}{\partial  M_{u_1v_1}\cdots
	\partial M_{u_kv_k}}(X(t))
	\bigg],
\end{multline*}
Also, we observe that $\kappa((Z_{(l,\gamma)})_{u_1v_1},\ldots,(Z_{(l,\gamma)})_{u_kv_k})$ is always $0$ when $k=3$ because the entries of $Z_{(l,\gamma)}$ are symmetric.
Therefore, we conclude that
\begin{multline*}
	\frac{d}{dt}\E[\phi(X(t))] =
	\frac{1}{2}\sum_{k=4}^\infty
	\frac{t^{\frac{k}{2}-1}}{(k-1)!}\times\mbox{}\\
	\sum_{(l,\gamma) \in \Xi}
	\sum_{\substack{(u_j,v_j)\in[m] \times [d]\\
	j=1,\ldots,k}}
	\kappa((Z_{(l,\gamma)})_{u_1v_1},\ldots,(Z_{(l,\gamma)})_{u_kv_k})\,
	\E\bigg[
	\frac{\partial^k\phi}{\partial  M_{u_1v_1}\cdots
	\partial M_{u_kv_k}}(X(t))
	\bigg],
\end{multline*}
Then the result follows by applying Lemma 4.1 in \cite{brailovskaya2022universality}.
\end{proof}

Using Lemma \ref{lem:derexpansion} with $\phi=f_{1,M_2}$ and $t_1$ as the variable name, we have
\begin{align*}
	&\frac{d}{dt_1}\E[f(S_1(t_1),M_2)]=\frac{d}{dt_1}\E[f_{1,M_2}(S_1(t_1))]
 \\=& 
	\frac{1}{2}\sum_{k=4}^{2q}
	\frac{t^{\frac{k}{2}-1}}{(k-1)!}
	\sum_{\pi\in\mathrm{P}([k])}
	(-1)^{|\pi|-1}(|\pi|-1)!\,
	\E\Bigg[ \sum_{(l,\gamma) \in [n] \times [pm]} \partial_{Z_{(l,\gamma),1|\pi}}\cdots\partial_{Z_{(l,\gamma),k|\pi}}f_{1,M_2}(S_1(t_1))
	\Bigg]
\end{align*}

For clarity, we introduce the following notations. Let $X,Y$ be two random elements taking values in measurable spaces $\mathscr{S}_1$ and $\mathscr{S}_2$. Let $\phi:\mathscr{S}_1 \times \mathscr{S}_2 \to \R$ be a measurable function. Assume that $\phi(X,Y)$ is integrable. Define the function $\psi_2:\mathscr{S}_2 \to \R$ such that $\psi_2(y)=\E\phi(X,y)$. Then the notation $\E_X(\phi(X,Y))$ stands for $\psi_2(Y)$. Similarly, we define the function $\psi_1:\mathscr{S}_1 \to \R$ such that $\psi_1(x)=\E\phi(x,Y)$ and the notation $\E_Y(\phi(X,Y))$ stands for the random variable $\psi_1(X)$. With these notation and observing that we can plug the random variable $S_2(t_2)$ into $M_2$ in the previous identity (since the identity holds for arbitrary $M_2$), we have
\begin{align*}
	&\frac{d}{dt_1}\E_{S_1(t_1)}[f(S_1(t_1),S_2(t_2))] \\=& 
	\frac{1}{2}\sum_{k=4}^{2q}
	\frac{t^{\frac{k}{2}-1}}{(k-1)!}
	\sum_{\pi\in\mathrm{P}([k])}
	(-1)^{|\pi|-1}(|\pi|-1)!\,
	\\ &\cdot\E_{S_1(t_1)}\Bigg[ \sum_{(l,\gamma) \in [n] \times [pm]} \partial_{Z_{(l,\gamma),1|\pi}}\cdots\partial_{Z_{(l,\gamma),k|\pi}}f_{1,S_2(t_2)}(S_1(t_1))
	\Bigg]
\end{align*}

By differentiation under integral sign \cite[Theorem 2.27]{folland2013real} and iterated integration, we have
\begin{align*}
	&\frac{d}{dt_1}\E_{S_2(t_2)}\E_{S_1(t_1)}[f(S_1(t_1),S_2(t_2))] \\=& \E_{S_2(t_2)}\frac{d}{dt_1}\E_{S_1(t_1)}[f(S_1(t_1),S_2(t_2))]
 \\=&
	\frac{1}{2}\sum_{k=4}^{2q}
	\frac{t^{\frac{k}{2}-1}}{(k-1)!}
	\sum_{\pi\in\mathrm{P}([k])}
	(-1)^{|\pi|-1}(|\pi|-1)!\,
	\\&\E_{S_2(t_2)}\E_{S_1(t_1)}\Bigg[ \sum_{(l,\gamma) \in [n] \times [pm]} \partial_{Z_{(l,\gamma),1|\pi}}\cdots\partial_{Z_{(l,\gamma),k|\pi}}f_{1,S_2(t_2)}(S_1(t_1))
	\Bigg],
\end{align*}

By the independence between $S_1(t_1)$ and $S_2(t_2)$, we can use iterated expectation and conclude that
\begin{align*}
    \E_{S_2(t_2)}\E_{S_1(t_1)}[f(S_1(t_1),S_2(t_2))]=\E[f(S_1(t_1),S_2(t_2))]
\end{align*}
and
\begin{align*}
    &\E_{S_2(t_2)}\E_{S_1(t_1)}\Bigg[ \sum_{(l,\gamma) \in [n] \times [pm]} \partial_{Z_{(l,\gamma),1|\pi}}\cdots\partial_{Z_{(l,\gamma),k|\pi}}f_{1,S_2(t_2)}(S_1(t_1))
	\Bigg]\\=&\E\Bigg[ \sum_{(l,\gamma) \in [n] \times [pm]} \partial_{Z_{(l,\gamma),1|\pi}}\cdots\partial_{Z_{(l,\gamma),k|\pi}}f_{1,S_2(t_2)}(S_1(t_1))
	\Bigg]
\end{align*}

Therefore, we have
\begin{align*}
	&\frac{d}{dt_1}\E[f(S_1(t_1),S_2(t_2))] 
 \\=&
	\frac{1}{2}\sum_{k=4}^{2q}
	\frac{t^{\frac{k}{2}-1}}{(k-1)!}
	\sum_{\pi\in\mathrm{P}([k])}
	(-1)^{|\pi|-1}(|\pi|-1)!\,
	\\&\E\Bigg[ \sum_{(l,\gamma) \in [n] \times [pm]} \partial_{Z_{(l,\gamma),1|\pi}}\cdots\partial_{Z_{(l,\gamma),k|\pi}}f_{1,S_2(t_2)}(S_1(t_1))
	\Bigg],
\end{align*}

Similarly, we can define $f_{2,M_1}(M_2)=\tr(\Gamma(M_1,M_2)^{2q})$ and obtain
\begin{align*}
	&\frac{d}{dt_2}\E[f(S_1(t_1),S_2(t_2))] \\=& 
	\frac{1}{2}\sum_{k=4}^{2q}
	\frac{t^{\frac{k}{2}-1}}{(k-1)!}
	\sum_{\pi\in\mathrm{P}([k])}
	(-1)^{|\pi|-1}(|\pi|-1)!\,\\&
	\E\Bigg[ \sum_{(l,\gamma) \in [n] \times [pm]}\partial_{Z_{(l,\gamma),1|\pi}}\cdots\partial_{Z_{(l,\gamma),k|\pi}}f_{2,S_1(t_1)}(S_2(t_2))
	\Bigg],
\end{align*}

We decompose \begin{align*}\phi(t)=\frac{d}{dt} \E [f(S_1(t), S_2(t))]\end{align*} as $\phi(t)=\psi_1 \circ \psi_2(t)$ where
$\psi_1(t_1,t_2)=\E[f(S_1(t_1),S_2(t_2))]$
and $(t_1,t_2)=\psi_2(t)=(t,t)$.

By chain rule, we have
\begin{align}
\begin{split} \label{eq:derestdecoup} 
	&\frac{d}{dt}\E[f(S_1(t),S_2(t))]
 \\=&
 \Bigg[\begin{array}{cc}
    \frac{d}{dt_1}(\E[f(S_1(t),S_2(t))])  & \frac{d}{dt_2}(\E[f(S_1(t),S_2(t))]) 
 \end{array}\Bigg] \cdot \Bigg[\begin{array}{cc}
      \frac{d}{dt_1}(t_1)  \\ \frac{d}{dt_2}(t_2)
 \end{array}\Bigg]
 \\=& 	\frac{1}{2}\sum_{k=4}^{2q}
	\frac{t^{\frac{k}{2}-1}}{(k-1)!}
	\sum_{\pi\in\mathrm{P}([k])}
	(-1)^{|\pi|-1}(|\pi|-1)!\,
	\\&\E\Bigg[ \sum_{(l,\gamma) \in [n] \times [pm]} \partial_{Z_{(l,\gamma),1|\pi}}\cdots\partial_{Z_{(l,\gamma),k|\pi}}f_{1,S_2(t)}(S_1(t))
	\Bigg]\\&+\frac{1}{2}\sum_{k=4}^{2q}
	\frac{t^{\frac{k}{2}-1}}{(k-1)!}
	\sum_{\pi\in\mathrm{P}([k])}
	(-1)^{|\pi|-1}(|\pi|-1)!\,
	\\&\E\Bigg[ \sum_{(l,\gamma) \in [n] \times [pm]} \partial_{Z_{(l,\gamma),1|\pi}}\cdots\partial_{Z_{(l,\gamma),k|\pi}}f_{2,S_1(t)}(S_2(t))
	\Bigg]\\=:& T_1 + T_2
 \end{split}
\end{align}
where $T_1$ and $T_2$ denote the first and second sum respectively.

The next step is to bound $\frac{d}{dt}\E[f(S_1(t),S_2(t))]$.

From Equation \eqref{eq:derestdecoup}, we notice that the terms inside the expectation in both $T_1$ and $T_2$ are directional derivatives of $f_{1,S_2(t_2)}$ and $f_{2,S_1(t_1)}$ along $Z_{(l,\gamma),1|\pi}, \ldots, Z_{(l,\gamma),k|\pi}$. Using a general expression for derivatives of multinomials using product rule, we have, for any deterministic $m \times d$ matrices $B_1, \ldots, B_k, M_1$ and $M_2$,
\begin{align*}
    &\partial_{B_1} \cdots \partial_{B_k}f_{1,M_2}(M_1)
\\=&\sum_{\sigma \in \sym (k)}\sum_{\substack{r_1,\ldots,r_{k+1}\ge 0\\
	r_1+\cdots+r_{k+1}=2q-k}}
	\tr (\Gamma(M_1,M_2)^{r_1}((B_{\sigma(1)}U)^TM_2U+(M_2U)^TB_{\sigma(1)}U)
	\Gamma(M_1,M_2)^{r_2}\\&((B_{\sigma(2)}U)^TM_2U+(M_2U)^TB_{\sigma(2)}U)\cdots
\Gamma(M_1,M_2)^{r_k}\\&((B_{\sigma(k)}U)^TM_2U+(M_2U)^TB_{\sigma(k)}U)\Gamma(M_1,M_2)^{r_{k+1}})
\end{align*}

In our case, for each fixed $(l, \gamma)$, we have to analyse (in the case of $T_1$) $\partial_{Z_{(l,\gamma),1|\pi}}\cdots\partial_{Z_{(l,\gamma),k|\pi}} \allowbreak f_{1,S_2(t_2)}(S_1(t_1))$, which means we have $B_{\lambda} = Z_{(l,\gamma),\lambda|\pi}$ for $\lambda \in [k]$, and $M_2 = S_2(t)$. So terms of the form $(B_{\lambda}U)^TM_2U$ become, 
\begin{align*}
    (Z_{(l,\gamma),\lambda|\pi}U)^TS_2(t)U
    \allowbreak =\xi_{(l,\gamma),\lambda|\pi} \mathbf{u}_l e_{\mu_{(l,\gamma),\lambda|\pi}}^TS_2(t)U
\end{align*}
where $\mathbf{u}_l$ is the column vector such that $\mathbf{u}_l^T$ is the $l$th row of $U$ and we use the fact that $Z_{(l,\gamma)}=\xi_{(l,\gamma)} e_{\mu_{(l, \gamma)}} e_l ^T$ with $\mu_{(l,\gamma)}$ being a uniformly distributed random variable on $[(m/s)(\gamma-1)+1:(m/s)\gamma]$ and $\xi_{(l,\gamma)}$ being a random sign. 

Let $\Theta_{(l,\gamma), \lambda, 1} = \xi_{(l,\gamma)}\mathbf{u}_l^T$ and $\Theta_{(l,\gamma),\lambda, 2} = e_{\mu_{(l,\gamma),\lambda|\pi}}(S_2(t)U) $. Then, 
\begin{align*}
    (Z_{(l,\gamma),\lambda|\pi}U)^TS_2(t)U
    \allowbreak =\xi_{(l,\gamma)} \mathbf{u}_l e_{\mu_{(l,\gamma),\lambda|\pi}}^TS_2(t)U = \Theta_{(l,\gamma), \lambda, 1}^T \Theta_{(l,\gamma), \lambda, 2}
\end{align*}

Thus, terms of the form $(B_{\lambda}U)^TM_2U+(M_2U)^TB_{\lambda}U$ become
\begin{align*}
    (Z_{(l,\gamma),\lambda|\pi}U)^TS_2(t)U + (S_2(t)U)^TZ_{(l,\gamma),\lambda|\pi}U &= \Theta_{(l,\gamma), \lambda, 1}^T\Theta_{(l,\gamma), \lambda, 2} + \Theta_{(l,\gamma), \lambda, 2}^T\Theta_{(l,\gamma), \lambda, 1} \\
    &= \sum_{\tau \in \sym(\{1,2\})} \Theta_{(l,\gamma), \lambda, \tau(1)}^T \Theta_{(l,\gamma), \lambda, \tau(2)}
\end{align*}

Doing this for all terms \begin{align*}(B_{\sigma(1)}U)^TM_2U+(M_2U)^TB_{\sigma(1)}U, \ldots, (B_{\sigma(1)}U)^TM_2U+(M_2U)^TB_{\sigma(1)}U\end{align*} in the expansion of $\partial_{Z_{(l,\gamma),1|\pi}}\cdots\partial_{Z_{(l,\gamma),k|\pi}} \allowbreak f_{1,S_2(t_2)}(S_1(t_1))$ (for a fixed $(l, \gamma)$ and $\pi \in P([k])$), we get
\begin{align*}
&\partial_{Z_{(l,\gamma),1|\pi}}\cdots\partial_{Z_{(l,\gamma),k|\pi}}  f_{1,S_2(t_2)}(S_1(t_1)) \\ 
    &= \sum_{\sigma \in \sym([k])} \sum_{\substack{r_1,\ldots,r_{k+1}\ge 0\\
	r_1+\cdots+r_{k+1}=2q-k}}
	 \sum_{\tau_1,...,\tau_k \in \sym(\{1,2\})}
	\tr \Gamma(t)^{r_1} \Theta_{(l,\gamma), \sigma(1), \tau_1(1)}^T \Theta_{(l,\gamma), \sigma(1), \tau_1(2)} \\& \cdot
	\Gamma(t)^{r_2}\Theta_{(l,\gamma), \sigma(2), \tau_2(1)}^T \Theta_{(l,\gamma), \sigma(2), \tau_2(2)}\cdots
\Gamma(t)^{r_k}\Theta_{(l,\gamma), \sigma(k), \tau_k(1)}^T \Theta_{(l,\gamma), \sigma(k), \tau_k(2)} \Gamma(t)^{r_{k+1}} \\
\end{align*}

Fixing $k$ and $\pi$, summing over $(l, \gamma) \in [n] \times [pm]$, and taking expectations, we have
\begin{align*}
    & \E \left[ \sum_{(l,\gamma) \in [n] \times [pm]} \partial_{Z_{(l,\gamma),1|\pi}}\cdots\partial_{Z_{(l,\gamma),k|\pi}}  f_{1,S_2(t_2)}(S_1(t_1))  \right] \\ 
    =& \sum_{\sigma \in \sym([k])} \sum_{\substack{r_1,\ldots,r_{k+1}\ge 0\\
	r_1+\cdots+r_{k+1}=2q-k}}
	 \sum_{\tau_1,...,\tau_k \in \sym(\{1,2\})}  
	 \sum_{(l,\gamma) \in [n] \times [pm]} \\& \E [ \tr \Gamma(t)^{r_1} \Theta_{(l,\gamma), \sigma(1), \tau_1(1)}^T \Theta_{(l,\gamma), \sigma(1), \tau_1(2)} \cdot
	\Gamma(t)^{r_2} \\&\cdots \Gamma(t)^{r_k}\Theta_{(l,\gamma), \sigma(k), \tau_k(1)}^T \Theta_{(l,\gamma), \sigma(k), \tau_k(2)} \Gamma(t)^{r_{k+1}} ] \\
\end{align*}

By Lemma \ref{lem:decouptraceineq} (when $S_1$ and $S_2$ have the unscaled OSNAP distribution) and by Lemma \ref{lem:decouptraceineqose} (when $S_1$ and $S_2$ have the unscaled OSE-IE distribution) with $\Upsilon_1=\Gamma(t)^{r_2},...,\Upsilon_k=\Gamma(t)^{r_k+r_1}$, we have
\begin{equation}\label{tracebound}
\begin{aligned}
&\sum_{(l,\gamma) \in [n] \times [pm]} \E [ \tr \Gamma(t)^{r_1} \Theta_{(l,\gamma), \sigma(1), \tau_1(1)}^T\Theta_{(l,\gamma), \sigma(1), \tau_1(2)}\cdot
	\Gamma(t)^{r_2} \\&\cdots \Gamma(t)^{r_k}\Theta_{(l,\gamma), \sigma(k), \tau_k(1)}^T\Theta_{(l,\gamma), \sigma(k), \tau_k(2)}\Gamma(t)^{r_{k+1}} ]\\=
 &\sum_{(l,\gamma) \in [n] \times [pm]} \E [ \tr  \Theta_{(l,\gamma), \sigma(1), \tau_1(1)}^T\Theta_{(l,\gamma), \sigma(1), \tau_1(2)} \cdot
	\Gamma(t)^{r_2} \\&\cdots \Gamma(t)^{r_k}\Theta_{(l,\gamma), \sigma(k), \tau_k(1)}^T\Theta_{(l,\gamma), \sigma(k), \tau_k(2)}\Gamma(t)^{r_{k+1}}\Gamma(t)^{r_1} ]\\=
    &\sum_{(l,\gamma) \in [n] \times [pm]} \E[ \tr \Theta_{(l,\gamma), 1, \tau_1(1)}^T\Theta_{(l,\gamma), 1, \tau_1(2)}
	\Upsilon_1\Theta_{(l,\gamma), 2, \tau_2(1)}^T\Theta_{(l,\gamma), 2, \tau_2(2)} \\&\cdots
	\Upsilon_2\Theta_{(l,\gamma), k, \tau_k(1)}^T\Theta_{(l,\gamma), k, \tau_k(2)}\Upsilon_k]  \\ \le &
       (c_{\ref*{lem:decouptraceineq}}(pm)^{\frac{1}{q}}\sqrt{\max\{pd,q\}})^{\frac{qk-2q}{q-1}} (\E \tr((\Gamma(t))^{2q}))^{\frac{1}{q} \cdot \frac{2q-k}{2q-2}} (\E \tr((\Gamma(t))^{2q}))^{1-\frac{k}{2q}}
\end{aligned}
\end{equation}

Therefore, we have
\begin{align*}
    &\E \left[ \sum_{(l,\gamma) \in [n] \times [pm]} \partial_{Z_{(l,\gamma),1|\pi}}\cdots\partial_{Z_{(l,\gamma),k|\pi}}  f_{1,S_2(t_2)}(S_1(t_1))  \right]
    \\\le& k! \binom{2q}{k} 2^k   (c_{\ref*{lem:decouptraceineq}}(pm)^{\frac{1}{q}}\sqrt{\max\{pd,q\}})^{\frac{qk-2q}{q-1}} (\E \tr((\Gamma(t))^{2q}))^{\frac{1}{q} \cdot \frac{2q-k}{2q-2}} (\E \tr((\Gamma(t))^{2q}))^{1-\frac{k}{2q}} 
    \\\le& (4q)^k   (c_{\ref*{lem:decouptraceineq}}(pm)^{\frac{1}{q}}\sqrt{\max\{pd,q\}})^{\frac{qk-2q}{q-1}}  (\E \tr((\Gamma(t))^{2q}))^{1-\frac{1}{2q}(k-2/r)}
\end{align*}
where $r=\frac{2q-2}{2q-k}$.

Going back to the expression for $T_1$ and using \cite[Lemma 6.4]{brailovskaya2022universality}, we get,
\begin{align*}
    T_1 =& \frac{1}{2}\sum_{k=4}^{2q}
	\frac{t^{\frac{k}{2}-1}}{(k-1)!}
	\sum_{\pi\in\mathrm{P}([k])}
	(-1)^{|\pi|-1}(|\pi|-1)!
	\\ &\cdot\E \Bigg[ \sum_{(l,\gamma) \in [n] \times [pm]} \partial_{Z_{(l,\gamma),1|\pi}}\cdots\partial_{Z_{(l,\gamma),k|\pi}}f_{1,S_2(t_2)}(S_1(t_1))
	\Bigg] \\
    \le& \frac{1}{2} \sum_{k=4}^{2q}
	\frac{1}{(k-1)!}  2^k (k-1)! \\ &\cdot (4q)^k   (c_{\ref*{lem:decouptraceineq}}(pm)^{\frac{1}{q}}\sqrt{\max\{pd,q\}})^{\frac{qk-2q}{q-1}}  (\E \tr((\Gamma(t))^{2q}))^{1-\frac{1}{2q}(k-2/r)} \\
    \le& \frac{1}{2} \sum_{k=4}^{2q} (8q)^k c_{1}^{k}  ((pm)^{\frac{1}{q}}\sqrt{\max\{pd,q\}})^{\frac{qk-2q}{q-1}}  (\E \tr((\Gamma(t))^{2q}))^{1-\frac{1}{2q}(k-2/r)}
\end{align*}

We similarly get $T_2 \le \frac{1}{2} \sum_{k=4}^{2q} (8q)^k c_{1}^{k}  ((pm)^{\frac{1}{q}}\sqrt{\max\{pd,q\}})^{\frac{qk-2q}{q-1}}  (\E \tr((\Gamma(t))^{2q}))^{1-\frac{1}{2q}(k-2/r)}$.

Going back, we have,
\begin{align*}
    \frac{d}{dt}\E[\tr \Gamma(t)^{2q}] \le&  \sum_{k=4}^{2q} (8c_1q)^k   ((pm)^{\frac{1}{q}}\sqrt{\max\{pd,q\}})^{\frac{qk-2q}{q-1}}  (\E \tr((\Gamma(t))^{2q}))^{1-\frac{1}{2q}(k-2/r)}
    \\ = & \sum_{k=4}^{2q} 2^{-k}(16c_1q)^k   ((pm)^{\frac{1}{q}}\sqrt{\max\{pd,q\}})^{\frac{qk-2q}{q-1}}  (\E \tr((\Gamma(t))^{2q}))^{1-\frac{1}{2q}(k-2/r)}
    \\ \le & \max \limits_{4 \le k \le 2q}(16c_1q)^k   ((pm)^{\frac{1}{q}}\sqrt{\max\{pd,q\}})^{\frac{qk-2q}{q-1}}  (\E \tr((\Gamma(t))^{2q}))^{1-\frac{1}{2q}(k-2/r)}
    \\ \le & \max \limits_{4 \le k \le 2q} (c_5q)^k ((pm)^{\frac{1}{q}}\sqrt{\max\{pd,q\}})^{\frac{qk-2q}{q-1}}(\E[f(S_1(t),S_2(t))])^{1-\frac{k-2}{2q-2}}
\end{align*}

\end{proof}

\subsection{Proving the trace inequality needed to obtain the differential inequality for the derivative of the interpolant} \label{subsec:osnaptraceineq}

In this section, we explain the following trace inequality result which is the key to prove the bound (\ref{tracebound}) in the proof of Lemma \ref{lem:diffineq}.

\begin{lemma} [Trace Inequalities for OSNAP]\label{lem:decouptraceineq}
Let $S_1(t)$ and $S_2(t)$ be as in Lemma \ref{lem:diffineq} with both having the fully independent unscaled OSNAP distribution. Let
\begin{align*}\Gamma(t)=(S_1(t)U)^T(S_2(t)U)+(S_2(t)U)^T(S_1(t)U)\end{align*} Let $\mathcal{Z}=\{\xi_{(l,\gamma)}:(l,\gamma) \in [n] \times [pm]\} \cup \{\mu_{(l,\gamma)}:(l,\gamma) \in [n] \times [pm]\}$ be the family of mutually independent random variables generating an instance of $S_1$ with the fully independent  unscaled OSNAP distribution. Let $q \ge 2$ and $3 \le k \le 2q$. Let $\{\mathcal{Z}_{\lambda}\}_{\lambda \in [k]}$ be a family of (possibly dependent) random elements, where for each $\lambda \in [k]$, the random element
\begin{align*}
    \mathcal{Z}_{\lambda}=\{\xi_{(l,\gamma),\lambda}:(l,\gamma) \in [n] \times [pm]\} \cup \{\mu_{(l,\gamma),\lambda}:(l,\gamma) \in [n] \times [pm]\}
\end{align*}
has the same distribution as 
\begin{align*}
    \mathcal{Z}=\{\xi_{(l,\gamma)}:(l,\gamma) \in [n] \times [pm]\} \cup \{\mu_{(l,\gamma)}:(l,\gamma) \in [n] \times [pm]\}
\end{align*}
Let $Z_{(l,\gamma)}=\xi_{(l,\gamma)} e_{\mu_{(l, \gamma)}} e_l ^T$ and $Z_{(l,\gamma),\lambda}=\xi_{(l,\gamma),\lambda} e_{\mu_{(l, \gamma),\lambda}} e_l ^T$.
Let $\{\Upsilon_1,...,\Upsilon_k\}$ be a family of $L_{\infty}(S_{\infty}^d)$ random matrices.
Assume further that the collection $\{\mathcal{Z}_{\lambda}\}_{\lambda \in [k]}$ is independent of $S_1, S_2, G_1, G_2$, and $\{\Upsilon_1,...,\Upsilon_k\}$. (In other words, $\{\Upsilon_1,...,\Upsilon_k\}$ can possibly be dependent with $S_1, S_2, G_1, G_2$.)
For each $(l,\gamma) \in [n] \times [pm]$ and $\lambda \in k$, we define random vectors $\Theta_{(l,\gamma), \lambda, 1}, \Theta_{(l,\gamma), \lambda, 2} \in \R^d$ such that
\begin{align*}
    \Theta_{(l,\gamma), \lambda, 1} = \xi_{(l,\gamma),\lambda} u_l^T \text{ and } \Theta_{(l,\gamma), \lambda, 2} = e_{\mu_{(l,\gamma),\lambda}}^TS_2(t)U
\end{align*}
where $e_{\mu_{(l,\gamma),\lambda}}$ represents the $\mu_{(l,\gamma),\lambda}$th coordinate vector. Then, given $0 \le \beta_1,...,\beta_k \le +\infty$ such that $\sum \limits_{\lambda=1}^k \frac{1}{\beta_{\lambda}}=1-\frac{k}{2q}$, $\tau_1, \ldots, \tau_k \in \sym(\{1,2 \})$, there exists $c_{\ref*{lem:decouptraceineq}}>0$ such that
\begin{equation}\label{eq:traceineq}
\begin{aligned}
    &\sum_{(l,\gamma) \in [n] \times [pm]} \E[ \tr \Theta_{(l,\gamma), 1, \tau_1(1)}^T\Theta_{(l,\gamma), 1, \tau_1(2)}
	\Upsilon_1\Theta_{(l,\gamma), 2, \tau_2(1)}^T\Theta_{(l,\gamma), 2, \tau_2(2)} \\&\cdots
	\Upsilon_2\Theta_{(l,\gamma), k, \tau_k(1)}^T\Theta_{(l,\gamma), k, \tau_k(2)}\Upsilon_k]  \\ \le &
       (c_{\ref*{lem:decouptraceineq}}(pm)^{\frac{1}{q}}\sqrt{\max\{pd,q\}})^{\frac{2qk-4q}{2q-2}} (\E \tr((\Gamma(t))^{2q}))^{\frac{1}{q} \cdot \frac{2q-k}{2q-2}} \prod \limits_{\lambda=1}^k \norm{\Upsilon_{\lambda}}_{\beta_{\lambda}}
\end{aligned}
\end{equation}
\end{lemma}
\begin{remark}
The main idea of this lemma is simple. We first use matrix Holder inequality to transform the left hand side into smaller factors, and then we bound those small factors separately. Following this idea, there could be different variants of this lemma. However, not all of them will eventually lead to the optimal dependency of the sparsity on $\varepsilon$. We will explain the idea on how to choose the correct bound. As explained in Proposition \ref{prop:momestdecoupmatrix}, the sparsity requirement comes from the condition
\begin{align*}
    (\E[\tr \Gamma(1)^{2q}])^{\frac{1}{2q}}-(\E[\tr \Gamma(0)^{2q}])^{\frac{1}{2q}} \le \frac{1}{4} pm \varepsilon
\end{align*}
If we want this condition to be implied by a requirement of the type
\begin{align*}
    pm>\frac{C (\log(d))^{\text{(some power)}}}{\varepsilon}
\end{align*}
then a natural attempt would be to try to bound $(\E[\tr \Gamma(1)^{2q}])^{\frac{1}{2q}}-(\E[\tr \Gamma(0)^{2q}])^{\frac{1}{2q}}$ by just a constant times some power of $\log(d)$.

To bound $(\E[\tr \Gamma(1)^{2q}])^{\frac{1}{2q}}-(\E[\tr \Gamma(0)^{2q}])^{\frac{1}{2q}}$, we combine our differential inequality with Lemma 6.6 in \cite{brailovskaya2022universality}.

By Lemma 6.6 in \cite{brailovskaya2022universality}, we can get the bound of the type
\begin{align*}
    (\E[\tr \Gamma(1)^{2q}])^{\alpha}-(\E[\tr \Gamma(0)^{2q}])^{\alpha} \le C\alpha +K^{1-\alpha}
\end{align*}
if we have a differential inequality of the form
\begin{align*}
    |\frac{d}{dt}(\E[\tr \Gamma(t)^{2q}])| \le C \max\{(\E[\tr \Gamma(t)^{2q}])^{1-\alpha},K^{1-\alpha}\}
\end{align*}
So we mainly want to obtain a differential inequality where $C$ and $K$ only contain factors of $\log(d)$ but do not contain any positive powers of $pm$ or any negative power of $\varepsilon \approx \sqrt{d/m}$. To this end, we need to choose a variant of the bound for the left hand side of (\ref{eq:traceineq}) which does not contain any positive powers of $pm$ or any negative power of $\varepsilon \approx \sqrt{d/m}$.

Therefore, when choosing the different variants of the bound in Lemma \ref{lem:decouptraceineq}, we seek to remove those factors that we do not want. One natural attempt would be to use the method similar to the proof of the second inequality in Theorem 2.9 in \cite{brailovskaya2022universality}, where they first bound the small factors that arise after using Holder by some matrix parameters, and then replace those matrix parameters by the trace moments by Jensen inequality. In our case, this means to replace $\sqrt{pmpd}$ by $(\E[\tr \Gamma(t)^{2q}])^{\frac{1}{2q}}$, thereby removing the factors $pm$ and $pd$.

However, in our case, after bounding the small factors obtained by directly separating the left hand side of (\ref{eq:traceineq}) using Holder, we can only obtain factors of the form $\sqrt{\max\{pd,q\}}$ and $\sqrt{\max\{pm,q\}}$. (Recall that $q \sim \log(d/\varepsilon\delta)$). The product $\sqrt{\max\{pd,q\} \cdot \max\{pm,q\}}$ of these two factors can be interchanged by the factor $\sqrt{pmpd}$ only when $pd \ge \log(d/\varepsilon\delta)$. But this already means that $pm \ge O(\frac{\log(d/\varepsilon\delta)}{\varepsilon^2})$, since $\varepsilon=O(\sqrt{\frac{d}{m}})$.

To get rid of the unbalanced factor $\sqrt{\max\{pd,q\} \cdot \max\{pm,q\}}$, we develop a completely different method. We combine several factors at early stage and then use a key observation, Lemma \ref{zssz}, to replace this combined factor by $(\E[\tr \Gamma(t)^{2q}])^{\frac{1}{2q}}$ directly. Eventually this method allowed us to get a more precise upper estimation, namely the right hand side of (\ref{eq:traceineq}). In this upper estimation, although there are still some factors of $\max\{pd,q\}$, they are not harmful, because $\max\{pd,q\}$ is of smaller order than $\sqrt{pdpm}$ (see the proof of Proposition \ref{prop:momestdecoupmatrix} for details).

\end{remark}
\begin{proof}[Proof of Lemma \ref{lem:decouptraceineq}]

Let
\begin{align*}
&F(\Upsilon_1,...,\Upsilon_k)\\=&\sum_{(l,\gamma) \in [n] \times [pm]} \E[ \tr \Theta_{(l,\gamma), 1, \tau_1(1)}^T\Theta_{(l,\gamma), 1, \tau_1(2)}
	\Upsilon_1\Theta_{(l,\gamma), 2, \tau_2(1)}^T\Theta_{(l,\gamma), 2, \tau_2(2)} \\&\cdots
	\Upsilon_2\Theta_{(l,\gamma), k, \tau_k(1)}^T\Theta_{(l,\gamma), k, \tau_k(2)}\Upsilon_k]
\end{align*}
Then, by Holder and the definition of $\Theta_{(l,\gamma), \lambda, 1}, \Theta_{(l,\gamma), \lambda, 2} \in \R^d$, we can show that $F(\Upsilon_1,...,\Upsilon_k)$ defines a multilinear functional on $L_{\infty}(S_{\infty}^d)$.

By Corollary \ref{cor:spinterpo}, it suffices to prove the claim for $(\beta_1,...,\beta_k)$ such that $(\frac{1}{\beta_1},...,\frac{1}{\beta_k})$ are extreme points of the set $\{(x_1,...,x_k) \in [0,1]^k:\sum \limits_{\lambda=1}^{k} x_{\lambda}=1-\frac{k}{2q}\}$, because all the other $(\frac{1}{\beta_1},...,\frac{1}{\beta_k})$ are convex combinations of the extreme points and the result follows from interpolation. Therefore, we only need to prove the claim for the extreme case when one of the $\beta_{\lambda}$'s is $\frac{q}{q-k}$ and all the others are $\infty$. By symmetry, we only need to consider the case $\beta_1= \cdots =\beta_{k-1}=\infty$ and $\beta_k=\frac{q}{q-k}$.

Let $\eta = (\eta(1), \eta(2))$ be a uniformly distributed random vector in $[n] \times [pm]$ such that $\eta$ is independent with $\{\mathcal{Z}_{\lambda}\}_{\lambda \in [k]}$, $S_1, S_2, G_1, G_2$, and $\{\Upsilon_1,...,\Upsilon_k\}$.
    
For all $\lambda \in [k]$, define random vectors $\Theta_{1, \lambda} = \Theta_{\eta, \lambda, 1}$ and $\Theta_{2, \lambda} = \Theta_{\eta, \lambda, 2}$. Then,
\begin{align*}
        &\sum_{(l,\gamma) \in [n] \times [pm]} \E[ \tr \Theta_{(l,\gamma), 1, \tau_1(1)}^T\Theta_{(l,\gamma), 1, \tau_1(2)}
	\Upsilon_1\Theta_{(l,\gamma), 2, \tau_2(1)}^T\Theta_{(l,\gamma), 2, \tau_2(2)}\\ &\cdots
	\Upsilon_2\Theta_{(l,\gamma), k, \tau_k(1)}^T\Theta_{(l,\gamma), k, \tau_k(2)}\Upsilon_k] \\ =&
    pmn \cdot \E[
	\tr \Theta_{\tau_{1}(1),1}^T\Theta_{\tau_{1}(2),1}
	\Upsilon_1\Theta_{\tau_{2}(1),2}^T\Theta_{\tau_{2}(2),2}\cdots
	\Upsilon_2\Theta_{\tau_{k}(1),k}^T\Theta_{\tau_{k}(2),k}\Upsilon_k] \\ =&
    \cS \cdot \E[
	\tr \Theta_{\tau_{1}(1),1}^T\Theta_{\tau_{1}(2),1}
	\Upsilon_1\Theta_{\tau_{2}(1),2}^T\Theta_{\tau_{2}(2),2}\cdots
	\Upsilon_2\Theta_{\tau_{k}(1),k}^T\Theta_{\tau_{k}(2),k}\Upsilon_k]
\end{align*}
    where $\cS := pmn$ is the number of non-zero entries in $S$.
    Essentially, we have written the sum over indices $(l, \gamma)$ as an expectation over uniformly chosen $(l, \gamma)$ times the number of tuples $(l, \gamma)$.

Let $\Upsilon_k=V_k|\Upsilon_k|$ be the polar decomposition. By matrix Holder inequality (Lemma 5.3. in \cite{brailovskaya2022universality}), we have
\begin{align*}
    &|\E[
	\tr \Theta_{\tau_{1}(1),1}^T\Theta_{\tau_{1}(2),1}
	\Upsilon_1\Theta_{\tau_{2}(1),2}^T\Theta_{\tau_{2}(2),2}\cdots
	\Upsilon_{k-1}\Theta_{\tau_{k}(1),k}^T\Theta_{\tau_{k}(2),k}\Upsilon_k]|
 \\ = & |\E[
	\tr |\Upsilon_k|^{1/2}\Theta_{\tau_{1}(1),1}^T\Theta_{\tau_{1}(2),1}
	\Upsilon_1\Theta_{\tau_{2}(1),2}^T\Theta_{\tau_{2}(2),2}\cdots
	\Upsilon_{k-1}\Theta_{\tau_{k}(1),k}^T\Theta_{\tau_{k}(2),k}V_k|\Upsilon_k|^{1/2}]|
 \\ \le & |\E[
	\tr \Theta_{\tau_{1}(1),1}^T\Theta_{\tau_{1}(2),1}
	\Upsilon_1 \cdots \Theta_{\tau_{k/2}(1),k/2}^T\Theta_{\tau_{k/2}(2),k/2}
	\Upsilon_{k/2} 
 \\& \cdot \Upsilon_{k/2}^T \Theta_{\tau_{k/2}(2),k/2}^T\Theta_{\tau_{k/2}(1),k/2} \cdots \Upsilon_{1}^T \Theta_{\tau_{1}(2),1}^T\Theta_{\tau_{1}(1),1} |\Upsilon_k|]|^{1/2}
\\ &\cdot |\E[
	\tr  \Theta_{\tau_{k}(2),k}^T\Theta_{\tau_{k}(1),k}
	\Upsilon_{k-1} \cdots \Upsilon_{k/2+1}^T \Theta_{\tau_{k/2+1}(2),k/2+1}^T\Theta_{\tau_{k/2+1}(1),k/2+1}
 \\& \cdot
	\Theta_{\tau_{k/2+1}(1),k/2+1}^T\Theta_{\tau_{k/2+1}(2),k/2+1} \Upsilon_{k/2+1} \cdots \Upsilon_{k-1}  \Theta_{\tau_{k}(1),k}^T\Theta_{\tau_{k}(2),k} V_k |\Upsilon_k|V_k^T]|^{1/2}
  \\ =& (\text{factor 1}) \cdot (\text{factor 2})
\end{align*}
where, assuming $k$ is even, in the notation of Proposition 5.1 in \cite{brailovskaya2022universality}, we set
\begin{align*}Y_1=|\Upsilon_k|^{1/2}\Theta_{\tau_{1}(1),1}^T\Theta_{\tau_{1}(2),1}
	\Upsilon_1 \cdots \Theta_{\tau_{k/2}(1),k/2}^T\Theta_{\tau_{k/2}(2),k/2}
	\Upsilon_{k/2}\end{align*}
 and
\begin{align*}Y_2=\Theta_{\tau_{k/2+1}(1),k/2+1}^T\Theta_{\tau_{k/2+1}(2),k/2+1} \Upsilon_{k/2+1} \cdots \Upsilon_{k-1}  \Theta_{\tau_{k}(1),k}^T\Theta_{\tau_{k}(2),k} V_k |\Upsilon_k|^{1/2}\end{align*}
with the exponents being $p_1=p_2=2$. When $k$ is odd, the argument can be modified as suggested in Step 4 of the proof of Proposition 5.1 in \cite{brailovskaya2022universality}.

We focus on the factor
\begin{align*}
    (\text{factor 1})=&|\E[
	\tr \Theta_{\tau_{1}(1),1}^T\Theta_{\tau_{1}(2),1}
	\Upsilon_1 \cdots \Theta_{\tau_{k/2}(1),k/2}^T\Theta_{\tau_{k/2}(2),k/2}
	\Upsilon_{k/2} 
 \\& \cdot \Upsilon_{k/2}^T \Theta_{\tau_{k/2}(2),k/2}^T\Theta_{\tau_{k/2}(1),k/2} \cdots \Upsilon_{1}^T \Theta_{\tau_{1}(2),1}^T\Theta_{\tau_{1}(1),1} |\Upsilon_k|]|^{1/2}
\end{align*}
and the analysis of the other factor will be similar.

Let
\begin{align*}
    \Theta_{\tau_{1}(2),1}^T\Theta_{\tau_{1}(1),1}=\Lambda_{1} |\Theta_{\tau_{1}(2),1}^T\Theta_{\tau_{1}(1),1}|
\end{align*}
be the polar decomposition.

Therefore, we have, for some $r \ge 1$ to be fixed later, 
\begin{equation}\label{eq:factor1}
\begin{aligned}
&(\text{factor 1})^2\\=
    &|\E[
	\tr \Theta_{\tau_{1}(1),1}^T\Theta_{\tau_{1}(2),1}
	\Upsilon_1 \cdots \Theta_{\tau_{k/2}(1),k/2}^T\Theta_{\tau_{k/2}(2),k/2}
	\Upsilon_{k/2} 
 \\& \cdot \Upsilon_{k/2}^T \Theta_{\tau_{k/2}(2),k/2}^T\Theta_{\tau_{k/2}(1),k/2} \cdots \Upsilon_{1}^T \Theta_{\tau_{1}(2),1}^T\Theta_{\tau_{1}(1),1} |\Upsilon_k|]|
 \\=&|\E[
	\tr |\Theta_{\tau_{1}(2),1}^T\Theta_{\tau_{1}(1),1}|^{1-1/r} \Lambda_{1}^T
	\Upsilon_1 \cdots \Theta_{\tau_{k/2}(1),k/2}^T\Theta_{\tau_{k/2}(2),k/2}
	\Upsilon_{k/2} 
 \\& \cdot \Upsilon_{k/2}^T \Theta_{\tau_{k/2}(2),k/2}^T\Theta_{\tau_{k/2}(1),k/2} \cdots \Upsilon_{1}^T \Lambda_{1} |\Theta_{\tau_{1}(2),1}^T\Theta_{\tau_{1}(1),1}|^{1-1/r} \\& \cdot|\Theta_{\tau_{1}(2),1}^T\Theta_{\tau_{1}(1),1}|^{1/r} |\Upsilon_k| |\Theta_{\tau_{1}(2),1}^T\Theta_{\tau_{1}(1),1}|^{1/r}]|
 \\ \le & \norm{\Theta_{\tau_{1}(2),1}^T\Theta_{\tau_{1}(1),1}}_{2q}^{2-2/r} \prod\limits_{\lambda  = 2}^{k/2}\norm{\Theta_{\tau_{\lambda}(1),\lambda}^T\Theta_{\tau_{\lambda}(2),\lambda}}_{2q}^2\prod\limits_{\lambda  = 1}^{k/2} {\norm{\Upsilon_\lambda }_{\infty}^2} \\ &\cdot \norm{|\Theta_{\tau_{1}(2),1}^T\Theta_{\tau_{1}(1),1}|^{1/r} |\Upsilon_k| |\Theta_{\tau_{1}(2),1}^T\Theta_{\tau_{1}(1),1}|^{1/r}}_r
 \end{aligned}
\end{equation}
where the last assertion follows from using the Hölder inequality with the exponents,
\begin{itemize}
    \item $2q$ for $\Theta_{\tau_j(1),j}^T\Theta_{\tau_j(2),j}$ and $\Theta_{\tau_j(2),j}^T\Theta_{\tau_j(1),j}$  for $j=2, \ldots, k/2$. Note that there are $2(\frac{k}{2}-1) = k-2$ many such terms.
    \item $\infty$ for $\Upsilon_j$ or $\Upsilon_j^T$ for $j=1, \ldots, k/2$, $\Lambda_{1}$ and $\Lambda_{1}^T$.
    \item $\frac{2qr}{r-1}$ for each term $|\Theta_{\tau_{1}(2),1}^T\Theta_{\tau_{1}(1),1}|^{1-1/r}$.
    \item $r$ for the term $|\Theta_{\tau_{1}(2),1}^T\Theta_{\tau_{1}(1),1}|^{1/r} |\Upsilon_k| |\Theta_{\tau_{1}(2),1}^T\Theta_{\tau_{1}(1),1}|^{1/r}$.
\end{itemize}
Then, the condition on exponents for Hölder inequality determines the value of $r$, 
\[ \frac{k-2}{2q} +  \frac{2(r-1)}{2rq} + \frac{1}{r} = 1 \]
giving, $r=\frac{2q-2}{2q-k}$.

Starting from here, we bound all the factors that appear in the final line of (\ref{eq:factor1}). First, the factor $\norm{\Theta_{\tau_{\lambda}(1),\lambda}^T\Theta_{\tau_{\lambda}(2),\lambda}}_{2q}=\norm{\Theta_{1,\lambda}^T\Theta_{2,\lambda}}_{2q}$ can be bounded directly by using the following lemma.

\begin{lemma}[Product of Random Rows] \label{theta1theta2}
Let $\Theta_{(l,\gamma), \lambda, 1}, \Theta_{(l,\gamma), \lambda, 2} \in \R^d$ be as in Lemma \ref{lem:decouptraceineq}. Let $\eta = (\eta(1), \eta(2))$ be a uniformly distributed random vector in $[n] \times [pm]$ such that $\eta$ is independent with $\{\mathcal{Z}_{\lambda}\}_{\lambda \in [k]}$, $S_1, S_2, G_1, G_2$. For all $\lambda \in [k]$, define random vectors $\Theta_{1, \lambda} = \Theta_{\eta, \lambda, 1}$ and $\Theta_{2, \lambda} = \Theta_{\eta, \lambda, 2}$. Then there exists a constant $c_{\ref*{theta1theta2}}$ such that, for any integer $q \ge 2$, we have
\begin{align*}
\norm{\Theta_{1,\lambda}^T\Theta_{2,\lambda}}_{q} &\le \frac{c_{\ref*{theta1theta2}}(pm)^{\frac{1}{q}}\sqrt{\max \{ pd, q \}}}{\cS^\frac{1}{q}}
\end{align*}
where $\cS=pmn$.
\end{lemma}

\begin{proof}
Note that for each fixed $1 \le \lambda \le k/2$, 
\begin{align*}
    \norm{\Theta_{1,\lambda}^T\Theta_{2,\lambda}}_{q} &= \norm{\xi_{\eta, \lambda} \mathbf{u}_{\eta(1)} e_{\mu_{\eta,\lambda}}^TS_2(t)U}_{q} \\
    &\le \left( \E [\tr \abs{\xi_{\eta, \lambda} \mathbf{u}_{\eta(1)} e_{\mu_{\eta,\lambda}}^TS_2(t)U}^{q} ] \right)^\frac{1}{q} \\
    &= \left( \E [\tr \abs{ \mathbf{u}_{\eta(1)} e_{\mu_{\eta,\lambda}}^TS_2(t)U}^{q} ] \right)^\frac{1}{q} \\
    &= \left( \E \left[ \frac{1}{d} \norm{ \mathbf{u}_{\eta(1)} }^{q}\norm{e_{\mu_{\eta,\lambda}}^TS_2(t)U}^{q} \right] \right)^\frac{1}{q} \\
\end{align*} 
Conditioning on $\eta$, we take expectation of the second factor over $\mu_{\eta, \lambda}$ and $S_2$,
\begin{align*}
    \norm{\Theta_{1,\lambda}^T\Theta_{2,\lambda}}_{q} &\le \left( \E_{\eta} \left[ \frac{1}{d} \norm{ \mathbf{u}_{\eta(1)} }^{q} \E_{\mu_{\eta, \lambda}, S_2(t)}[ \norm{e_{\mu_{\eta,\lambda}}^TS_2(t)U}^{q}] \right] \right)^\frac{1}{q} \\
\end{align*}

Note that, conditioned on $\eta$, $\mu_{\eta, \lambda}$ is uniformly distributed over a subset of $[m]$. Therefore, by Lemma \ref{lem:rownormbound}, we have  
\begin{align*}
    \norm{\Theta_{1,\lambda}^T\Theta_{2,\lambda}}_{q} &\le \left( \E_{\eta} \left[ \frac{1}{d} \norm{ \mathbf{u}_{\eta(1)} }^{q} \E_{\mu_{\eta, \lambda}, S_2(t)}[ \norm{e_{\mu_{\eta,\lambda}}^TS_2(t)U}^{q}] \right] \right)^\frac{1}{q} \\
    &\le c_{\ref*{lem:rownormbound}}\sqrt{\max \{ pd, q \}} \left( \E_{\eta} \left[ \frac{1}{d} \norm{ \mathbf{u}_{\eta(1)} }^{q} \right] \right)^\frac{1}{q} \\
    &= c_{\ref*{lem:rownormbound}}\sqrt{\max \{ pd, q \}} \left( \frac{1}{\cS  d} \sum_{(l,\gamma) \in [n]\times [pm]} \norm{u_l}^{q} \right)^\frac{1}{q} \\
    &= c_{\ref*{lem:rownormbound}}\sqrt{\max \{ pd, q \}} \left( \frac{pm}{\cS  d} \sum_{l=1}^{n} \norm{u_l}^{q} \right)^\frac{1}{q} \\
    &\le c_{\ref*{lem:rownormbound}}\sqrt{\max \{ pd, q \}} \left( \frac{pm}{\cS  d} \sum_{l=1}^{n} \norm{u_l}^{2} \right)^\frac{1}{q} \\
    &\le \frac{c_{\ref*{theta1theta2}}(pm)^{\frac{1}{q}}\sqrt{\max \{ pd, q \}}}{\cS^\frac{1}{q}}
\end{align*}
where in the last line we use the fact that $\sum_{l=1}^{n} \norm{u_l}^{2} = d$.

\end{proof}

Next, we have the following bound for the factor $\norm{|\Theta_{\tau_{1}(2),1}^T\Theta_{\tau_{1}(1),1}|^{1/r} |\Upsilon_k| |\Theta_{\tau_{1}(2),1}^T\Theta_{\tau_{1}(1),1}|^{1/r}}_{r}$.
\begin{align*}
    &\norm{|\Theta_{\tau_{1}(2),1}^T\Theta_{\tau_{1}(1),1}|^{1/r} |\Upsilon_k| |\Theta_{\tau_{1}(2),1}^T\Theta_{\tau_{1}(1),1}|^{1/r}}_{r}^{r} \\=&\E\tr((| \Theta_{\tau_{1}(2),1}^T\Theta_{\tau_{1}(1),1}|^{1/r} |\Upsilon_k| | \Theta_{\tau_{1}(2),1}^T\Theta_{\tau_{1}(1),1}|^{1/r})^{r})
    \\ \le & \E\tr(| \Theta_{\tau_{1}(2),1}^T\Theta_{\tau_{1}(1),1}| |\Upsilon_k|^{r} | \Theta_{\tau_{1}(2),1}^T\Theta_{\tau_{1}(1),1}|)
    \\=& \E\tr( |\Upsilon_k|^{r} | \Theta_{\tau_{1}(2),1}^T\Theta_{\tau_{1}(1),1}|| \Theta_{\tau_{1}(2),1}^T\Theta_{\tau_{1}(1),1}|)
    \\=& \E\tr( |\Upsilon_k|^{r}  (\Theta_{\tau_{1}(2),1}^T\Theta_{\tau_{1}(1),1})^T \Theta_{\tau_{1}(2),1}^T\Theta_{\tau_{1}(1),1})
    \end{align*}
where the inequality is due to the Lieb Thirring Inequality (\cite[Lemma 5.4]{brailovskaya2022universality}).

Here, we have two possible cases. The first case (CASE I) is $\tau_1(1)=1$. In this case, we have
\begin{align*}
   & (\Theta_{\tau_{1}(2),1}^T\Theta_{\tau_{1}(1),1})^T (\Theta_{\tau_{1}(2),1}^T\Theta_{\tau_{1}(1),1})\\
   =&\Theta_{1,1}^T\Theta_{2,1} (\Theta_{2,1}^T\Theta_{1,1})\\=&(Z_{\eta,1}U)^T(S_2(t)U)(S_2(t)U )^TZ_{\eta,1}U 
\end{align*}

Therefore, we have
\begin{align*}
& \E\tr( |\Upsilon_k|^{r}  (\Theta_{\tau_{1}(2),1}^T\Theta_{\tau_{1}(1),1})^T \Theta_{\tau_{1}(2),1}^T\Theta_{\tau_{1}(1),1})\\=&\E\tr( |\Upsilon_k|^{r}  (Z_{\eta,1}U)^T(S_2(t)U)(S_2(t)U )^TZ_{\eta,1}U )
    \\=&\E\tr( |\Upsilon_k|^{r}  \E_{\eta}((Z_{\eta,1}U)^T(S_2(t)U)(S_2(t)U )^TZ_{\eta,1}U) )
    \\=&\E\tr( |\Upsilon_k|^{r}  \frac{1}{\cS}(\sum_{(l,\gamma)}(Z_{(l,\gamma),1}U)^T(S_2(t)U)(S_2(t)U )^TZ_{(l,\gamma),1}U) )
    \\=&\frac{1}{\cS}\E\tr( |\Upsilon_k|^{r}  \E_{Z}(\sum_{(l,\gamma)}(Z_{(l,\gamma),1}U)^T(S_2(t)U)(S_2(t)U )^TZ_{(l,\gamma),1}U) )
    \\ \le& \frac{1}{\cS} \norm{|\Upsilon_k|^{r}}_{\frac{2q}{2q-2}} \cdot \norm{\E_{Z}(\sum_{(l,\gamma)}(Z_{(l,\gamma),1}U)^T(S_2(t)U)(S_2(t)U )^TZ_{(l,\gamma),1}U)}_{q}
\end{align*}

Plugging this estimate and the previous estimate on $\norm{\Theta_{\tau_j(1),j}^T\Theta_{\tau_j(2),j}}_{q}$ into inequality \ref{eq:factor1}, we have
\begin{equation}\label{eq:factor1, second}
\begin{aligned}
&(\text{factor 1})^2\\ 
\le & \norm{\Theta_{2,1}^T\Theta_{1,1}}_{2q}^{2-2/r} \prod\limits_{\lambda  = 2}^{k/2}\norm{\Theta_{1,\lambda}^T\Theta_{2,\lambda}}_{2q}^2\prod\limits_{\lambda  = 1}^{k/2} {\norm{\Upsilon_\lambda }_{\infty}^2} \norm{|\Theta_{2,1}^T\Theta_{1,1}|^{1/r} |\Upsilon_k| |\Theta_{2,1}^T\Theta_{1,1}|^{1/r}}_r
 \\ \le & \paren*{\frac{c_{\ref*{theta1theta2}}(pm)^{\frac{1}{q}}\sqrt{\max \{ pd, q \}}}{\cS^\frac{1}{q}}}^{2-2/r}\paren*{\frac{c_{\ref*{theta1theta2}}(pm)^{\frac{1}{q}}\sqrt{\max \{ pd, q \}}}{\cS^\frac{1}{q}}}^{k-2} \prod\limits_{\lambda  = 1}^{k/2} {\norm{\Upsilon_\lambda }_{\infty}^2} \\ &\cdot \paren*{\frac{1}{\cS} \norm{|\Upsilon_k|^{r}}_{\frac{2q}{2q-2}} \cdot \norm*{\E_{Z}\paren*{\sum_{(l,\gamma)}(Z_{(l,\gamma),1}U)^T(S_2(t)U)(S_2(t)U )^TZ_{(l,\gamma),1}U} }_{q}}^{1/r}
  \\=&\paren*{\frac{c_{\ref*{theta1theta2}}(pm)^{\frac{1}{q}}\sqrt{\max \{ pd, q \}}}{\cS^\frac{1}{q}}}^{\frac{2qk-4q}{2q-2}}\prod\limits_{\lambda  = 1}^{k/2} {\norm{\Upsilon_\lambda }_{\infty}^2} \\ &\cdot \paren*{\frac{1}{\cS} \norm{|\Upsilon_k|^{r}}_{\frac{2q}{2q-2}} \cdot \norm*{\E_{Z}\paren*{\sum_{(l,\gamma)}(Z_{(l,\gamma),1}U)^T(S_2(t)U)(S_2(t)U )^TZ_{(l,\gamma),1}U} }_{q}}^{1/r}
\end{aligned}
\end{equation}

To bound the term $\norm{\E_{Z}(\sum_{(l,\gamma)}(Z_{(l,\gamma),1}U)^T(S_2(t)U)(S_2(t)U )^TZ_{(l,\gamma),1}U)}_{q}$, we can use the following lemma to relate it to $\Gamma(t)$.

\begin{lemma}[Replacement by the Original Matrix] \label{zssz}
Let $S_1(t)$, $S_2(t)$, $\Gamma(t)$, and $Z_{(l,\gamma),\lambda}$  be as in Lemma \ref{lem:decouptraceineq}.

We have
\begin{align*}
    &\norm{\E_{Z_{(l,\gamma),1}}(\sum_{(l,\gamma)}(Z_{(l,\gamma),1}U)^T(S_2(t)U)(S_2(t)U )^TZ_{(l,\gamma),1}U)}_{q}^q
    \\ \le & \E \tr((\Gamma(t))^{2q})
\end{align*}
and
\begin{align*}
    &\norm{\E_{Z_{(l,\gamma),1}}(\sum_{(l,\gamma)}(S_2(t)U)^T(Z_{(l,\gamma),1}U)(Z_{(l,\gamma),1}U)^T(S_2(t)U )}_{q}^q
    \\ \le & \E \tr((\Gamma(t))^{2q})
\end{align*}
\end{lemma}

\begin{proof}
Let $W_{(l,\gamma),1}$ be the gaussian model for $Z_{(l,\gamma),1}$, i.e., the entries of $W_{(l,\gamma),1}$ are gaussian and the covariance structure between entries of $W_{(l,\gamma),1}$ is the same as $Z_{(l,\gamma),1}$, such that the family $\{W_{(l,\gamma),1}:(l,\gamma) \in \Xi\} \cup \{Z_{(l,\gamma),1}:(l,\gamma) \in \Xi\}$ are mutually independent. Let $Z_{(l,\gamma),1}(t)=\sqrt{t}Z_{(l,\gamma),1}+\sqrt{1-t}W_{(l,\gamma),1}$. With this notation, we know that $\sum \limits_{(l,\gamma) \in \Xi}Z_{(l,\gamma),1}(t)$ has the same distribution with $S_1(t)$.

We observe that
\begin{align*}
    &\E_{Z_{(l,\gamma),1}}(\sum_{(l,\gamma)}(Z_{(l,\gamma),1}U)^T(S_2(t)U)(S_2(t)U )^TZ_{(l,\gamma),1}U)
    \\=&\E_{Z_{(l,\gamma),1}(t)}(\sum_{(l,\gamma)}(Z_{(l,\gamma),1}(t)U)^T(S_2(t)U)(S_2(t)U )^TZ_{(l,\gamma),1}(t)U)
\end{align*}
because for fixed $(S_2(t)U)(S_2(t)U )^T$, the value
\begin{align*}
    \E_{Z_{(l,\gamma),1}}(\sum_{(l,\gamma)}(Z_{(l,\gamma),1}U)^T(S_2(t)U)(S_2(t)U )^TZ_{(l,\gamma),1}U)
\end{align*}
only depends on the covariance structure of $Z_{(l,\gamma),1}$.

Also, we have
\begin{align*}
    &\E_{Z_{(l,\gamma),1}(t)}(\sum_{(l,\gamma)}(Z_{(l,\gamma),1}(t)U)^T(S_2(t)U)(S_2(t)U )^TZ_{(l,\gamma),1}(t)U)
    \\=&\E_{Z_{(l,\gamma),1}(t)}((\sum_{(l,\gamma)}Z_{(l,\gamma),1}(t)U)^T(S_2(t)U)(S_2(t)U )^T(\sum_{(l,\gamma)}Z_{(l,\gamma),1}(t))U)
    \\=&\E_{S_1(t)}((S_1(t)U)^T(S_2(t)U)(S_2(t)U )^TS_1(t)U)
\end{align*}
because the cross terms have zero expectation by independence.

Therefore, we have
\begin{align*}
    &\E_{Z_{(l,\gamma),1}}\paren*{\sum_{(l,\gamma)}(Z_{(l,\gamma),1}U)^T(S_2(t)U)(S_2(t)U )^TZ_{(l,\gamma),1}U}
    \\=&\E_{S_1(t)}((S_1(t)U)^T(S_2(t)U)(S_2(t)U )^TS_1(t)U)
\end{align*}

Similarly, we also have
\begin{align*}
    &\E_{Z_{(l,\gamma),1}}\paren*{\sum_{(l,\gamma)}(S_2(t)U)^T(Z_{(l,\gamma),1}U)(Z_{(l,\gamma),1}U)^T(S_2(t)U )}\\
    =&\E_{S_1(t)}((S_2(t)U)^T(S_1(t)U)(S_1(t)U )^TS_2(t)U)
\end{align*}

Now we look at $\E_{S_1(t)}(\Gamma(t)^2)$. We have
\begin{align*}
    &\E_{S_1(t)}(\Gamma(t)^2)\\=&\E_{S_1(t)}((S_1(t)U)^T(S_2(t)U)(S_2(t)U )^TS_1(t)U)\\&+\E_{S_1(t)}((S_2(t)U)^T(S_1(t)U)(S_1(t)U )^TS_2(t)U)\\&+\E_{S_1(t)}((S_1(t)U)^T(S_2(t)U)(S_1(t)U )^TS_2(t)U)\\&+\E_{S_1(t)}((S_2(t)U)^T(S_1(t)U)(S_2(t)U )^TS_1(t)U)
\end{align*}

To simplify the last two terms, we need the following observation. Consider a random matrix $X$ whose entries have variance $1$ and zero covariances. Let $Y$ be a deterministic matrix. Then direct calculation can show that $\E_X XYX=Y^T$.

In fact, we have
\begin{align*}
(XY)_{i,\beta}=\sum_{\alpha}X_{i,\alpha}Y_{\alpha,\beta}
\end{align*}
and therefore
\begin{align*}
\E_X (XYX)_{i,j}=\E_X (\sum_{\beta}\sum_{\alpha}X_{i,\alpha}Y_{\alpha,\beta}X_{\beta,j})=\sum_{\beta}\sum_{\alpha}\E_X(X_{i,\alpha}Y_{\alpha,\beta}X_{\beta,j})
\end{align*}

We observe that $\E_X(X_{i,\alpha}Y_{\alpha,\beta}X_{\beta,j})=Y_{j,i}$ if $\alpha=j$ and $\beta=i$. When $(\alpha,\beta) \ne (j,i)$, we always have $\E_X(X_{i,\alpha}Y_{\alpha,\beta}X_{\beta,j})=0$. Therefore, we have $\E_X (XYX)_{i,j}=Y_{j,i}$, and then we conclude $\E_X (XYX)=Y^T$. 

Observing that the matrix $(S_1(t)U)^T$ has entries with variance $1$ and zero covariances, the above analysis gives that
\begin{align*}
    &\E_{S_1(t)}((S_1(t)U)^T(S_2(t)U)(S_1(t)U )^TS_2(t)U)\\=&(\E_{S_1(t)}((S_1(t)U)^T(S_2(t)U)(S_1(t)U )^T))(S_2(t)U)\\=&(S_2(t)U)^T(S_2(t)U)
\end{align*}
and
\begin{align*}
    &\E_{S_1(t)}((S_2(t)U)^T(S_1(t)U)(S_2(t)U )^TS_1(t)U)\\=&(S_2(t)U)^T(\E_{S_1(t)}((S_1(t)U)(S_2(t)U )^TS_1(t)U))\\=&(S_2(t)U)^T(S_2(t)U)
\end{align*}

Therefore, we have
\begin{align*}
    &\E_{S_1(t)}(\Gamma(t)^2)\\=&\E_{S_1(t)}((S_1(t)U)^T(S_2(t)U)(S_2(t)U )^TS_1(t)U)\\&+\E_{S_1(t)}((S_2(t)U)^T(S_1(t)U)(S_1(t)U )^TS_2(t)U)\\&+\E_{S_1(t)}((S_1(t)U)^T(S_2(t)U)(S_1(t)U )^TS_2(t)U)\\&+\E_{S_1(t)}((S_2(t)U)^T(S_1(t)U)(S_2(t)U )^TS_1(t)U)
    \\=&\E_{S_1(t)}((S_1(t)U)^T(S_2(t)U)(S_2(t)U )^TS_1(t)U)\\&+\E_{S_1(t)}((S_2(t)U)^T(S_1(t)U)(S_1(t)U )^TS_2(t)U)\\&+2(S_2(t)U)^T(S_2(t)U )
\end{align*}

Since all the matrices in the sum \begin{align*}&\E_{S_1(t)}((S_1(t)U)^T(S_2(t)U)(S_2(t)U )^TS_1(t)U)\\&+\E_{S_1(t)}((S_2(t)U)^T(S_1(t)U)(S_1(t)U )^TS_2(t)U)\\&+2(S_2(t)U)^T(S_2(t)U )\end{align*} are positively semidefinite, we can conclude that
\begin{align*}
    \E_{S_1(t)}((S_1(t)U)^T(S_2(t)U)(S_2(t)U )^TS_1(t)U) \le \E_{S_1(t)}(\Gamma(t)^2)
\end{align*}
and
\begin{align*}
    \E_{S_1(t)}((S_2(t)U)^T(S_1(t)U)(S_1(t)U )^TS_2(t)U) \le \E_{S_1(t)}(\Gamma(t)^2)
\end{align*}

Therefore, by \cite[Theorem 2.10]{carlen2010trace}, we have
\begin{align*}
    &\norm{\E_{Z_{(l,\gamma),1}}(\sum_{(l,\gamma)}(Z_{(l,\gamma),1}U)^T(S_2(t)U)(S_2(t)U )^TZ_{(l,\gamma),1}U)}_{q}^q
    \\=&\norm{\E_{S_1(t)}((S_1(t)U)^T(S_2(t)U)(S_2(t)U )^TS_1(t)U)}_{q}^q
    \\=&\E \tr((\E_{S_1(t)}((S_1(t)U)^T(S_2(t)U)(S_2(t)U )^TS_1(t)U))^q)
    \\ \le & \E \tr((\E_{S_1(t)}(\Gamma(t)^2))^q)
    \\ \le & \E \tr((\Gamma(t)^2)^q)
    \\ \le & \E \tr((\Gamma(t))^{2q})
\end{align*}

Similarly, we also have
\begin{align*}
    &\norm{\E_{Z_{(l,\gamma),1}}(\sum_{(l,\gamma)}(S_2(t)U)^T(Z_{(l,\gamma),1}U)Z_{(l,\gamma),1}U)^T(S_2(t)U )}_{q}^q
    \\ \le & \E \tr((\Gamma(t))^{2q})
\end{align*}

\end{proof}

Therefore, we have bounded all the factors in the last line of (\ref{eq:factor1}), and the calculations give
\begin{align*}
    &(\text{factor 1})^2
 \\ \le &\paren*{\frac{c_{\ref*{theta1theta2}}(pm)^{\frac{1}{q}}\sqrt{\max \{ pd, q \}}}{\cS^\frac{1}{q}}}^{\frac{2qk-4q}{2q-2}}\prod\limits_{\lambda  = 1}^{k/2} {\norm{\Upsilon_\lambda }_{\infty}^2} \\ &\cdot (\frac{1}{\cS} \norm{|\Upsilon_k|^{r}}_{\frac{2q}{2q-2}} \cdot \norm{\E_{Z_{(l,\gamma),1}(t)}(\sum_{(l,\gamma)}(Z_{(l,\gamma),1}(t)U)^T(S_2(t)U)(S_2(t)U )^TZ_{(l,\gamma),1}(t)U)}_{q})^{1/r}
  \\=&\cS^{-1}(c_{\ref*{theta1theta2}}(pm)^{\frac{1}{q}} \sqrt{\max\{pd,q\}})^{\frac{2qk-4q}{2q-2}}\prod\limits_{\lambda  = 1}^{k/2} {\norm{\Upsilon_\lambda }_{\infty}^2} \cdot \norm{\Upsilon_k}_{\frac{2q}{2q-k}}(\E \tr((\Gamma(t))^{2q}))^{(1/q)(1/r)}
\end{align*}

The second case (CASE II) is $\tau_1(1)=2$, or $\tau_1(2)=1$. In this case, we have
\begin{align*}
   & (\Theta_{\tau_{1}(2),1}^T\Theta_{\tau_{1}(1),1})^T (\Theta_{\tau_{1}(2),1}^T\Theta_{\tau_{1}(1),1})\\=&(\Theta_{1,1}^T\Theta_{2,1})^T (\Theta_{1,1}^T\Theta_{2,1})\\=&(S_2(t)U)^T Z_{\eta,1}U (Z_{\eta,1}U)^TS_2(t)U
\end{align*}

Therefore, we have
\begin{align*}
    &\norm{| \Theta_{1,1}^T\Theta_{2,1}|^{1/r} |\Upsilon_k| | \Theta_{1,1}^T\Theta_{2,1}|^{1/r}}_{r'}^{r'} \\ \le  & \E\tr( |\Upsilon_k|^{r}  (S_2(t)U)^T Z_{\eta,1}U (Z_{\eta,1}U)^TS_2(t)U )
    \end{align*}

Following similar calculations as in CASE (I), we obtain
\begin{align*}
    &\E\tr( |\Upsilon_k|^{r}  (S_2(t)U)^T Z_{\eta,1}U (Z_{\eta,1}U)^TS_2(t)U )
    \\=&\E\tr( |\Upsilon_k|^{r}  \E_{\eta}(S_2(t)U)^T Z_{\eta,1}U (Z_{\eta,1}U)^TS_2(t)U )
    \\=&\E\tr( |\Upsilon_k|^{r}  \frac{1}{\cS}(\sum\limits_{(l,\gamma)}(S_2(t)U)^T Z_{(l,\gamma),1}U (Z_{(l,\gamma),1}U)^TS_2(t)U ))
    \\=&\frac{1}{\cS}\E\tr( |\Upsilon_k|^{r}  \E_{Z_{(l,\gamma),1}}(\sum\limits_{(l,\gamma)}(S_2(t)U)^T Z_{(l,\gamma),1}U (Z_{(l,\gamma),1}U)^TS_2(t)U ))
    \\ \le& \frac{1}{\cS} \norm{|\Upsilon_k|^{r}}_{\frac{2q}{2q-2}} \cdot \norm{\E_{Z_{(l,\gamma),1}}(\sum_{(l,\gamma)}(S_2(t)U)^T Z_{(l,\gamma),1}U (Z_{(l,\gamma),1}U)^TS_2(t)U)}_{q}
\end{align*}
because $\eta$ and $Z_{(l,\gamma),1}$ are independent with $\Upsilon_k$ and $S_2(t)$.

By Lemma \ref{zssz}, we still get
\begin{align*}
    &(\text{factor 1})^2
  \\ \le &\cS^{-1}(c_{\ref*{theta1theta2}}(pm)^{\frac{1}{q}} \sqrt{\max\{pd,q\}})^{\frac{2qk-4q}{2q-2}}\prod\limits_{\lambda  = 1}^{k/2} {\norm{\Upsilon_\lambda }_{\infty}^2} \cdot \norm{\Upsilon_k}_{\frac{2q}{2q-k}}(\E \tr((\Gamma(t))^{2q}))^{(1/q)(1/r)}
\end{align*}

We can repeat the same argument for (factor 2) and show that
\begin{align*}
    &(\text{factor 2})^2
  \\ \le &\cS^{-1}(c_{\ref*{theta1theta2}} (pm)^{\frac{1}{q}}\sqrt{\max\{pd,q\}})^{\frac{2qk-4q}{2q-2}}\prod\limits_{\lambda  = k/2+1}^{k-1} {\norm{\Upsilon_\lambda }_{\infty}^2} \cdot \norm{\Upsilon_k}_{\frac{2q}{2q-k}}(\E \tr((\Gamma(t))^{2q}))^{(1/q)(1/r)}
\end{align*}

Combining the estimates for (factor 1) and (factor 2) together, we have
\begin{align*}
    &\sum_{(l,\gamma) \in [n] \times [pm]} \E[ \tr \Theta_{(l,\gamma), 1, \tau_1(1)}^T\Theta_{(l,\gamma), 1, \tau_1(2)}
	\Upsilon_1\Theta_{(l,\gamma), 2, \tau_2(1)}^T\Theta_{(l,\gamma), 2, \tau_2(2)} \\&\cdots
	\Upsilon_2\Theta_{(l,\gamma), k, \tau_k(1)}^T\Theta_{(l,\gamma), k, \tau_k(2)}\Upsilon_k]  \\ \le &
      (c_{\ref*{theta1theta2}}(pm)^{\frac{1}{q}}\sqrt{\max\{pd,q\}})^{\frac{2qk-4q}{2q-2}} (\E \tr((\Gamma(t))^{2q}))^{\frac{1}{q} \cdot \frac{2q-k}{2q-2}} \prod\limits_{\lambda  = 1}^{k-1} {\norm{\Upsilon_\lambda }_{\infty}} \cdot \norm{\Upsilon_k}_{\frac{2q}{2q-k}}
\end{align*}

Now, we have shown that
\begin{align*}
    &\sup \limits_{(\Upsilon_1,...,\Upsilon_k) \in L_{\infty}(S_{\infty}^d)^k}\frac{F(\Upsilon_1,...,\Upsilon_k)}{(\prod\limits_{\lambda  = 1}^{k-1} {\norm{\Upsilon_\lambda }_{\infty}}) \cdot \norm{\Upsilon_k}_{\frac{2q}{2q-k}}} \\\le& (c_{\ref*{theta1theta2}}(pm)^{\frac{1}{q}}\sqrt{\max \{ pd, q \}})^{\frac{2qk-4q}{2q-2}} (\E \tr((\Gamma(t))^{2q}))^{\frac{1}{q} \cdot \frac{2q-k}{2q-2}} 
\end{align*}

The final result follows from Corollary \ref{cor:spinterpo}.

\end{proof}

\section{Leverage Score Sparsified Embeddings} \label{sec:lessproofs}

In this section, we prove our subspace embedding guarantee for the LESS-IC distribution, Theorem \ref{t:less-ic}. The proof is similar to the OSNAP case, and is accomplished via the following results,
\begin{itemize}
    \item Theorem \ref{t:less-ic} establishes the subspace embedding guarantee from a bound on the trace moments of the embedding error analogous to Theorem \ref{t:ose-full} in the OSNAP case.
    \item Lemma \ref{lem:decoupless} shows that it is sufficient to bound the moments of $(S_1U)^TS_2U + (S_2U)^TS_1U$ to control the trace moments of the embedding error, analogous to Lemma \ref{lem:decoup}. This is the decoupling step.
    \item Lemma \ref{lem:lessvar} establishes that the entries of the LESS-IC distribution are uncorrelated and have variance $p$, which means the Gaussian model that we interpolate with should have independent entries with variance $p$, just as in the case of OSNAP.
    \item Lemma \ref{prop:momestdecoupless} bounds the trace moments of the embedding error via interpolation after decoupling, analogous to Lemma \ref{prop:momestdecoupmatrix}.
    \item Lemma \ref{lem:diffineqless} establishes the differential inequality for the derivative of the interpolant, analogous to Lemma \ref{lem:diffineq} in the OSNAP and OSE-IE case.
    \item Lemma \ref{lem:traceineqless} establishes the trace inequality required to obtain the differential inequality in Lemma \ref{lem:diffineqless}, analogous to Lemma \ref{lem:decouptraceineq} in the case of OSNAP.
    \item The trace inequality in Lemma \ref{lem:traceineqless} in turn requires a bound for the row norm moments of $S(t)U$. This is provided by Lemma \ref{lem:lessrownormbound}, analogous to Lemma \ref{lem:rownormbound}.
\end{itemize}

Before we proceed, we state the formal definition of the LESS-IC distribution. The construction is similar to OSNAP, with some changes to reflect the different number and size of subcolumns and the different scaling for non-zero entries across columns.

\begin{definition}[LESS-IC]\label{def:lessindcol}
    Given $(\beta_1,\beta_2)$ leverage scores $z_1,...,z_n$, and $0 < p < 1 $, define
    \begin{align*}
        b_j &:= 
        \max \Big\{ \left\lfloor \frac{1}{\beta_1pz_j} \right\rfloor, 1 \Big\}\quad\text{and}\quad
        s_j := \left\lceil \frac{m}{b_j} \right\rceil.
    \end{align*} 
    An $m \times n$ random matrix $S$ is called a $K$-wise independent unscaled leverage score sparsified embedding with independent columns ($K$-wise independent unscaled LESS-IC), and also $\Pi = (1/\sqrt{pm})S$ is called a $K$-wise independent LESS-IC, corresponding to  $(\beta_1, \beta_2)$-approximate leverage scores  $(z_1,...,z_n)$ with parameter $p$ if it is distributed as
    \begin{align*}
    S &= \sum_{l=1}^n \sum_{\gamma_l=1}^{s_l} \alpha_{(l, \gamma_l)} \xi_{(l,\gamma_l)} e_{\mu_{(l, \gamma_l)}} e_l ^\top 
\end{align*}
where in this expression
\begin{itemize}
    \item the collections $\{ \xi_{(l,\gamma_l)} \}_{l \in [n], \gamma_l \in [s_l]}$ and $\{ \mu_{(l,\gamma_l)} \}_{l \in [n], \gamma_l \in [s_l]}$ are mutually independent;
    \item $\{ \xi_{l,\gamma_l} \}_{l \in [n], \gamma_l \in [s_l]}$ is a collection of $K$-wise independent Rademacher random variables;
    \item $\{ \mu_{(l,\gamma_l)} \}_{l \in [n], \gamma_l \in [s_l]}$ is a collection of $K$-wise independent random variables such that each $\mu_{(l,\gamma_l)}$ is uniformly distributed in $[b_l(\gamma_l-1)+1:\min \{ b_l\gamma_l, m \}]$;
    \item $\alpha_{(l, \gamma_l)} := \sqrt{p(\min \{ b_l\gamma_l, m \} - b_l(\gamma_l-1))}$;
    \item $e_{\mu_{(l, \gamma)}}$ and $e_l$ represent basis vectors in $\R^m$ and $\R^n$ respectively.
\end{itemize}
In addition, if all the random variables in the collections $\{ \xi_{(l,\gamma_{l})} \}_{l \in [n], \gamma \in [s]}$ and $\{ \mu_{(l,\gamma_l)} \}_{l \in [n], \gamma \in [s]}$ are fully independent, then $S$ is called a fully independent unscaled LESS-IC and $\Pi$ is called a fully independent LESS-IC.
\end{definition}

\lessicmainthm*

\begin{remark}
    We can obtain a subspace embedding guarantee for the LESS-IE distribution by following the proof of Theorem \ref{thm:oseiemain} suitably modified for the case of LESS. One can check that Theorem \ref{prop:momestdiag} holds even when $S$ has the LESS-IE distribution with the same values of $\sigma$ and $R$. Thus, a subspace embedding for the LESS-IE distribution holds under the same conditions as Theorem \ref{t:less-ic} with the additional requirement $pm \ge \frac{c \log (\frac{d}{\varepsilon \delta})}{\varepsilon^2}$.  
\end{remark}

\begin{proof}
    The proof of the main statement using Theorem \ref{prop:momestdecoupless} is identical to the proof of Theorem \ref{t:ose-full}. As in the proof of Theorem \ref{t:ose-full}, we shall assume that the collection of random variables $\{ \xi_{l,\gamma_l} \}_{l \in [n], \gamma_l \in [s_l]}$ and $\{ \mu_{l,\gamma_l} \}_{l \in [n], \gamma_l \in [s_l]}$ (See Definition \ref{def:lessindcol}) are fully independent in our calculations, and since the proof proceeds via looking at moments of order $O( \log (\frac{d}{\varepsilon \delta}))$, the same calculations still hold when we have $\log$-wise independence. 

    In this case, we require $q$ to satisfy $$ pm \ge \paren*{\max \left\{ c_{\ref*{prop:momestdecoupless}.2} \sqrt{e} (q)^{5/2}/\varepsilon, c_{\ref*{prop:momestdecoupless}.3} (q)^3 \right\}}^{1+\frac{1}{q-2}} $$. Just as in the case of OSNAP, it is enough to ensure that $c_{\ref*{prop:momestdecoupless}.1} \frac{d+\log(d/\varepsilon\delta)}{(\varepsilon/\sqrt{e})^2} \le m $, and
    \begin{align*}pm \ge & C_1 \max \left\{ \frac{C_2(\log(\frac{d}{\varepsilon \delta}))^{2.5}}{ \varepsilon}, C_3(\log(\frac{d}{\varepsilon \delta}))^3 \right\}
    \end{align*}
    where the form of the constants $C_1, C_2$ and $C_3$ are analogous to those in \eqref{pmassum}, and set $q= \lceil 2\log (\frac{d}{\varepsilon \delta} )\rceil+2$. 
    
    However, the value of $p$ satisfying $pm \ge  C_1 \max \left\{ \frac{C_2(\log(\frac{d}{\varepsilon \delta}))^{2.5}}{ \varepsilon^{1+4/d}}, C_3(\log(\frac{d}{\varepsilon \delta}))^3 \right\}$ has to be smaller than $1$ for our construction to be defined. So we must also have \begin{align*}m \ge & C_1 \max \left\{ \frac{C_2(\log(\frac{d}{\varepsilon \delta}))^{2.5}}{ \varepsilon}, C_3(\log(\frac{d}{\varepsilon \delta}))^3 \right\}\end{align*}

    It is enough to ensure,
    \begin{align*}
        m &\ge \frac{ \max \{ C_1C_2, C_1C_3 \} \paren*{ \log (d/\delta\varepsilon) }^3}{\varepsilon} \\
        &=   \max \{ C_1C_2, C_1C_3 \} \paren*{ \frac{\log(1/\varepsilon)^3}{\varepsilon} + \frac{3\log(1/\varepsilon)^2\log(d/\delta)}{\varepsilon} + \frac{3\log(1/\varepsilon)\log(d/\delta)^2}{\varepsilon} + \frac{\log(d/\delta)^3}{\varepsilon} } \\
    \end{align*}
     Now, we observe that
    \begin{align*}
        \frac{\log(1/\varepsilon)^2}{\varepsilon} + \frac{3\log(1/\varepsilon)^2\log(d/\delta)}{\varepsilon} &\le \frac{1}{\varepsilon^2} + \frac{ 3 \log(d/\delta)}{\varepsilon^2} \\
        \text{and, } \quad \frac{3\log(1/\varepsilon)\log(d/\delta)^2}{\varepsilon}  &\le \frac{3\log(d/\delta)^2}{\varepsilon^2}  \\
    \end{align*}
    So it is enough to have, for some $C_4>0$,
    \[ m \ge C_4 \frac{d + \log(d/\delta)^2 + \log(1/\varepsilon)}{\varepsilon^2} + \frac{\log(d/\delta)^3}{\varepsilon} \]

    We proceed to estimate the number of non-zero entries. With reference to Definition \ref{def:lessindcol}, we have, 
    \[ b_j = \left\lfloor \frac{1}{\beta_1pz_j} \right\rfloor \ge \frac{1}{2\beta_1pz_j} \]
    when $\frac{1}{\beta_1pz_j} \ge 1$. Thus,
    \[ s_j = \left\lceil \frac{m}{b_j} \right\rceil \le \lceil 2\beta_1pmz_j \rceil \le \max \left\{ 1 , 4\beta_1pmz_j \right\}\]
    When $\frac{1}{\beta_1pz_j} \le 1$, 
    \[ s_j = \left\lceil \frac{m}{b_j} \right\rceil = m \le m\beta_1pz_j \]
    
    The total number of non-zero entries in $S$ is,
    \[ \sum_{j=1}^n s_j \le n + \sum_{j=1}^n  4\beta_1pmz_j \le n +  4\beta_1\beta_2 pmd \]
\end{proof}

\begin{lemma}[Decoupling] \label{lem:decoupless}
When $S$ has the fully independent unscaled LESS-IC distribution,
\begin{align*}
    \E [ \tr (U^TS^TSU - pm\cdot I_d)^{2q} ] &= \E \left[ \tr \left( \sum_{i=1}^m \sum_{j,j' =1, j \neq j'}^n s_{ij}s_{ij'} u_ju_{j'}^T \right)^{2q} \right] \\
\end{align*}
Consequently,
\begin{align*}
    \E [ \tr (U^TS^TSU - pm\cdot I_d)^{2q} ] &\le \E_{S,S'} \left[ \tr \left(  2\paren*{(SU)^TS'U + (SU)^TS'U} \right)^{2q} \right]
\end{align*}
where $S'$ is an independent copy of $S$.
\end{lemma}
\begin{proof}
    Letting $s_{ij}$ denote the entries of $S$, 
\begin{align*}
    U^TS^TSU - pm\cdot I_d &= \left( \sum_{i=1}^m U^T\left( \sum_{j=1}^n s_{ij}e_{j} \right) \left( \sum_{j'=1}^n s_{ij'}e_{j'}^T \right)U \right) - pm\cdot I_d \\
    &=  \sum_{i=1}^m \left( \sum_{j=1}^n s_{ij}u_{j} \right) \left( \sum_{j'=1}^n s_{ij'}u_{j'}^T \right)  - pm\cdot I_d \\
\end{align*}
where $u_j^T$ denotes the $j\textsuperscript{th}$ row of $U$. Separating the cases where $j=j'$ and $j \neq j'$, 
\begin{align*}
    U^TS^TSU - pm\cdot I_d &=  \sum_{i=1}^m \sum_{j=1}^n s_{ij}^2 u_ju_j^T  - pm\cdot I_d + \sum_{i=1}^m \sum_{\substack{j,j' =1 \\ j \neq j'}}^n s_{ij}s_{ij'} u_ju_{j'}^T \\
    &=  \sum_{j=1}^n \left( \sum_{i=1}^m s_{ij}^2 \right) u_ju_j^T  - pm\cdot I_d + \sum_{i=1}^m \sum_{\substack{j,j' =1 \\ j \neq j'}}^n s_{ij}s_{ij'} u_ju_{j'}^T
\end{align*}
By construction, $\sum_{i=1}^m s_{ij}^2 = \sum_{\gamma_j =1}^{s_j} \alpha_{(j, \gamma_j)}^2 = pm$, so,
\begin{align*}
    U^TS^TSU - pm\cdot I_d &= \sum_{i=1}^m \sum_{\substack{j,j' =1 \\ j \neq j'}}^n s_{ij}s_{ij'} u_ju_{j'}^T
\end{align*}
    From this point on, the proof is exactly the same as Lemma \ref{lem:decoup}.
\end{proof}

\begin{lemma}[Variance and Uncorrelatedness] \label{lem:lessvar}
Let $p = p_{m,n} \in (0,1]$ and $S=\{s_{ij}\}_{i \in [m], j \in [n]}$ be a $m \times n$ random matrix distributed according to the fully independent unscaled LESS-IC distributions. Then, $\E(s_{ij})=0$ and $\operatorname{Var}(s_{ij})=p$ for all $i \in [m], j \in [n]$, and $\cov(s_{i_1 j_1},s_{i_2 j_2})=0$ for any $\{i_1,i_2\} \subset [m], \{j_1,j_2\} \subset [n]$ and $ (i_1,j_1)\neq (i_2,j_2) $
\end{lemma}
\begin{proof}
We have
\begin{align*}
    \E (s_{ij}) =& \alpha_{(j, \gamma_j)}^2 \E(\xi_{j,\gamma_j}) \Pb(\mu_{(j, \gamma_l)}=i) \\=&\alpha_{(j, \gamma_j)}^2 \cdot 0 \cdot  \Pb(\mu_{(j, \gamma_l)}=i)
    \\=& 0
\end{align*}
because $\E(\xi_{j,\gamma_j})=0$.

\begin{align*}
    \E (s_{ij}^2) =& \alpha_{(j, \gamma_j)}^2 \E(\xi_{j,\gamma_j}^2) \Pb(\mu_{(j, \gamma_l)}=i) \\=&(\sqrt{p(\min \{ b_l\gamma_l, m \} - b_l(\gamma_l-1))})^2 \cdot 1 \cdot  \frac{1}{(\min \{ b_l\gamma_l, m \} - b_l(\gamma_l-1))}
    \\=& p
\end{align*}

For the covariances, we first observe that $\cov(s_{i_1 j_1}, s_{i_2 j_2})=0$ if $(i_1,j_1)$ and $(i_2,j_2)$ belong to two different subcolumns by independence. If $(i_1,j_1)$ and $(i_2,j_2)$ belong to the same subcolumn, we have
\begin{align*}
    \cov(s_{i_1 j_1}, s_{i_2 j_2})=\E(s_{i_1 j_2}s_{i_2 j_1})=\E(0)=0
\end{align*}
because each subcolumn has one hot distribution which means at most one of $s_{i_1 j_1}$ and $s_{i_2 j_2}$ can be nonzero.

\end{proof}

 Using the decoupling result, we bound the trace moments of the embedding error by interpolating between LESS and its Gaussian model exactly as in the proof of Lemma \ref{prop:momestdecoupmatrix} in Section \ref{subsec:osnaptracemom}.

\begin{lemma} [Trace Moments of Embedding Error for LESS]\label{prop:momestdecoupless}
Let $S$ be an $m \times n$ matrix distributed according to the fully independent unscaled LESS-IC distribution with parameter $p$ for some fixed matrix $U$ satisfying $U^TU=I$ with given $(\beta_1, \beta_2)$-approximate leverage scores. Define $X = \frac{1}{\sqrt{pm}}SU$. Given $0< \varepsilon < 1$, there exist constants $c_{\ref*{prop:momestdecoupless}.1}, c_{\ref*{prop:momestdecoupless}.2}, c_{\ref*{prop:momestdecoupless}.3}$ such that for $m \geq c_{\ref*{prop:momestdecoupless}.1} \frac{d+q}{\varepsilon^2}$ and $q \in \N$ satisfying $2 \le q \le m$ and $ pm \ge \paren*{ \max \left\{ \frac{c_{\ref*{prop:momestdecoupless}.2} q^{5/2}}{\varepsilon}, c_{\ref*{prop:momestdecoupless}.3}  q^3 \right\}}^{1+\frac{2}{q-2}} $ ,
\begin{align*} \E[\tr(X^TX - I_d)^{2q}]^\frac{1}{2q} \leq  \varepsilon \end{align*}
\end{lemma}

\begin{proof}
    The structure of the proof is the same as the proof of Theorem \ref{prop:momestdecoupmatrix} and we only highlight changes in the specific values. By Lemma \ref{lem:diffineqless}, we have,
    \begin{align*}
    \frac{d}{dt}\E[\tr \Gamma(t)^{2q}] \le&  \max \limits_{4 \le k \le 2q} (c_{\ref*{lem:diffineqless}}q)^k ((pm)^\frac{1}{q}\sqrt{\max\{pd,q^2 \}})^{\frac{2qk-4q}{2q-2}}(\E[f(S_1(t),S_2(t))])^{1-\frac{k-2}{2q-2}}
\end{align*}

As in the proof of Theorem \ref{prop:momestdecoupmatrix}, we divide the analysis into two cases - (i) $ pd < q^2$ and (ii) $ pd \ge q^2$.

In the first case, following the steps in the proof of Theorem \ref{prop:momestdecoupmatrix}, we get,
\begin{align*}
    &(\E[\tr \Gamma(S_1,S_2)^{2q}])^{\alpha}-(\E[\tr \Gamma(G_1,G_2)^{2q}])^{\alpha} \\\le& (c_5q)^4 ((pm)^\frac{1}{q}\sqrt{\max\{ pd,q^2 \}})^{\frac{2q}{q-1}} \cdot \frac{1}{q-1}+ ((c_5q)^{2q-4} ((pm)^\frac{1}{q}\sqrt{\max\{pd,q^2 \}})^{\frac{2q^2-4q}{q-1}})^{\frac{\alpha}{1-\alpha}}
    \\\le& (c_5q)^4 (q)^{\frac{2q}{q-1}}(pm)^\frac{2}{q-1} \cdot \frac{1}{q-1}+ ((c_5q)^{2q-4} ((pm)^\frac{1}{q}q)^{\frac{2q^2-4q}{q-1}})^{\frac{1}{q-2}}
    \\=& (c_5q)^4 q^2 (q)^{\frac{2}{q-1}}(pm)^\frac{2}{q-1} \cdot \frac{1}{q-1}+ ((c_5q)^{2} q^\frac{2q}{q-1}(pm)^\frac{2}{q-1})
    \\ \le & c_6 q^5(pm)^\frac{2}{q-1}
\end{align*}

Therefore, we have
\begin{align*}
    (\E[\tr \Gamma(S_1,S_2)^{2q}])^{\frac{1}{q-1}}-(\E[\tr \Gamma(G_1,G_2)^{2q}])^{\frac{1}{q-1}} \le c_6 q^5(pm)^\frac{2}{q-1}
\end{align*}

which gives,
\begin{align*}
    (\E[\tr \Gamma(S_1,S_2)^{2q}])^{\frac{1}{2q}}-(\E[\tr \Gamma(G_1,G_2)^{2q}])^{\frac{1}{2q}} \le (c_6q^5)^{\frac{q-1}{2q}} \le c_7 q^{5/2}(pm)^\frac{1}{q}
\end{align*}
so in this case, we need $pm \ge c_8(q)^{5/2}/\varepsilon$. 

The case when $pd > q^2$ gives the same lower bound on $pm$ as the corresponding case in Theorem \ref{prop:momestdecoupmatrix}.

\end{proof}

Here we obtain the differential inequality that arises during interpolation in the proof of Lemma \ref{prop:momestdecoupless}. 
\begin{lemma}[Differential Inequality for LESS]\label{lem:diffineqless}
Let $p>0$. Let $S_1$ and $S_2$ be independent random matrices with the fully independent unscaled LESS-IC distribution with parameter $p$ for some fixed matrix $U$  satisfying $U^TU=I$ with given $(\beta_1, \beta_2)$-approximate leverage scores. Let $G_1$ and $G_2$ be independent random matrices with i.i.d. Guassian entries each with variance $p$, and define the interpolated random matrices,
\begin{align}
\begin{split} 
    S_1(t) = \sqrt{t}S_1 + \sqrt{1-t}G_1 \\
    S_2(t) = \sqrt{t}S_2 + \sqrt{1-t}G_2
\end{split}
\end{align}
Let $f(M_1,M_2)=\tr(((M_1U)^T(M_2U)+(M_2U)^T(M_1U))^{2q})$.
Then there exists $c_{\ref*{lem:diffineqless}} >0$ such that, for any $q\ge 2$,
\begin{align*}
    \frac{d}{dt} \E[f(S_1(t),S_2(t))] \le&  \max \limits_{4 \le k \le 2q} (c_{\ref*{lem:diffineqless}}q)^k ((pm)^\frac{1}{q}\sqrt{\max\{ pd,q^2\}})^{\frac{2qk-4q}{2q-2}}(\E[f(S_1(t),S_2(t))])^{1-\frac{k-2}{2q-2}}
\end{align*}
\end{lemma}

\begin{proof}
    Following the steps in the proof of Lemma \ref{lem:diffineq}, we get,
    \begin{align}
\begin{split} \label{eq:derestdecoupless} 
	&\frac{d}{dt}\E[f(S_1(t),S_2(t))]
 \\=& 	\frac{1}{2}\sum_{k=4}^{2q}
	\frac{t^{\frac{k}{2}-1}}{(k-1)!}
	\sum_{\pi\in\mathrm{P}([k])}
	(-1)^{|\pi|-1}(|\pi|-1)!\,
	\\&\E\Bigg[ \sum_{l \in [n], \gamma_l \in [s_l]} \partial_{Z_{(l,\gamma_l),1|\pi}}\cdots\partial_{Z_{(l,\gamma_l ),k|\pi}}f_{1,S_2(t)}(S_1(t))
	\Bigg]\\&+\frac{1}{2}\sum_{k=4}^{2q}
	\frac{t^{\frac{k}{2}-1}}{(k-1)!}
	\sum_{\pi\in\mathrm{P}([k])}
	(-1)^{|\pi|-1}(|\pi|-1)!\,
	\\&\E\Bigg[ \sum_{(l,\gamma_l) \in [n] \times [pm]} \partial_{Z_{(l,\gamma_l),1|\pi}}\cdots\partial_{Z_{(l,\gamma_l),k|\pi}}f_{2,S_1(t)}(S_2(t))
	\Bigg]\\=:& T_1 + T_2
 \end{split}
\end{align}
where $f_{1,S_2(t)}(S_1(t))$ and $f_{2,S_1(t)}(S_2(t))$ are as in the proof of Lemma \ref{lem:diffineq} and $Z_{(l,\gamma_l),j|\pi}$ are constructed from $Z_{(l,\gamma_l)}$ appearing in $S$ written as the sum $\sum_{l=1}^n \sum_{\gamma_l =1}^{s_l} Z_{l,\gamma_l}$ when $S$ has the unscaled LESS-IC distribution. Once again we remind the reader about our assumption that the $Z_{(l,\gamma_l)}$ are independent, with the additional remark that since all quantities involved are moments of order at most $4q$, the same calculations hold even if the $Z_{(l,\gamma)}$ are $4q$-wise independent.

In this case, we use Lemma \ref{lem:traceineqless} to estimate \eqref{eq:derestdecoupless} as in the proof of Lemma \ref{lem:diffineq} to get,
\begin{align*}
    \frac{d}{dt}\E[\tr \Gamma(t)^{2q}] \le&  \max \limits_{4 \le k \le 2q} (c_{\ref*{lem:diffineqless}}q)^k ((pm)^\frac{1}{q}\sqrt{\max\{ pd,q^2 \}})^{\frac{2qk-4q}{2q-2}}(\E[f(S_1(t),S_2(t))])^{1-\frac{k-2}{2q-2}}
\end{align*}
\end{proof}

As in the oblivious case, a trace inequality is the key step in the proof of Lemma \ref{lem:diffineqless}.
\begin{lemma} [Trace Inequalities for LESS]\label{lem:traceineqless}
Let $S_1(t)$ and $S_2(t)$ be as in Lemma \ref{lem:diffineqless}. Let $\Gamma(t)=(S_1(t)U)^T(S_2(t)U)+(S_2(t)U)^T(S_1(t)U)$. Let $\mathcal{Z}=\{\xi_{(l,\gamma)}: l \in [n], \gamma_l \in [s_l] \} \cup \{\mu_{(l,\gamma)}:l \in [n], \gamma_l \in [s_l] \}$ be a family of mutually independent random variables be the family of mutually independent random variables generating an instance of $S_1$ with the unscaled LESS-IC distribution corresponding to some $(\beta_1, \beta_2)$-approximate leverage scores for $U$. Let $q \ge 2$ and $3 \le k \le 2q$. Let $\{\mathcal{Z}_{\lambda}\}_{\lambda \in [k]}$ be a family of (possibly dependent) random elements, where for each $\lambda \in [k]$, the random element
\begin{align*}
    \mathcal{Z}_{\lambda}=\{\xi_{(l,\gamma),\lambda}:l \in [n], \gamma_l \in [s_l]\} \cup \{\mu_{(l,\gamma),\lambda}:l \in [n], \gamma_l \in [s_l]\}
\end{align*}
has the same distribution as 
\begin{align*}
    \mathcal{Z}=\{\xi_{(l,\gamma)}:l \in [n], \gamma_l \in [s_l]\} \cup \{\mu_{(l,\gamma)}:l \in [n], \gamma_l \in [s_l]\}
\end{align*}
Let $Z_{(l,\gamma)}=\xi_{(l,\gamma)} e_{\mu_{(l, \gamma)}} e_l ^T$ and $Z_{(l,\gamma),\lambda}=\xi_{(l,\gamma),\lambda} e_{\mu_{(l, \gamma),\lambda}} e_l ^T$.
Let $\{\Upsilon_1,...,\Upsilon_k\}$ be a family of $L_{\infty}(S_{\infty}^d)$ random matrices.
Assume further that the collection $\{\mathcal{Z}_{\lambda}\}_{\lambda \in [k]}$ is independent of $S_1, S_2, G_1, G_2$, and $\{\Upsilon_1,...,\Upsilon_k\}$. (In other words, $\{\Upsilon_1,...,\Upsilon_k\}$ can possibly be dependent with $S_1, S_2, G_1, G_2$.)
For each $l \in [n], \gamma_l \in [s_l]$ and $\lambda \in k$, we define random vectors $\Theta_{(l,\gamma), \lambda, 1}, \Theta_{(l,\gamma), \lambda, 2} \in \R^d$ such that
\begin{align*}
    \Theta_{(l,\gamma), \lambda, 1} = \xi_{(l,\gamma),\lambda} \alpha_{(l,\gamma),\lambda} u_l^T \text{ and } \Theta_{(l,\gamma), \lambda, 2} = e_{\mu_{(l,\gamma),\lambda}}^TS_2(t)U
\end{align*}
where $e_{\mu_{(l,\gamma),\lambda}}$ represents the $\mu_{(l,\gamma),\lambda}$th coordinate vector. Then, given $0 \le \rho_1,...,\rho_k \le +\infty$ such that $\sum \limits_{\lambda=1}^k \frac{1}{\rho_{\lambda}}=1-\frac{k}{2q}$, $\tau_1, \ldots, \tau_k \in \sym(\{1,2 \})$, there exists $c_{\ref*{lem:traceineqless}} >0$ such that
\begin{align*}
    &\sum_{l \in [n], \gamma_l \in [s_l]} \E[ \tr \Theta_{(l,\gamma), 1, \tau_1(1)}^T\Theta_{(l,\gamma), 1, \tau_1(2)}
	\Upsilon_1\Theta_{(l,\gamma), 2, \tau_2(1)}^T\Theta_{(l,\gamma), 2, \tau_2(2)} \\&\cdots
	\Upsilon_2\Theta_{(l,\gamma), k, \tau_k(1)}^T\Theta_{(l,\gamma), k, \tau_k(2)}\Upsilon_k]  \\ \le &
     (c_{\ref*{lem:traceineqless}}(pm)^\frac{1}{q}\sqrt{\max\{\beta pd,q^2\}})^{\frac{2qk-4q}{2q-2}} (\E \tr((\Gamma(t))^{2q}))^{\frac{1}{q} \cdot \frac{2q-k}{2q-2}} \prod\limits_{\lambda  = 1}^{k} {\norm{\Upsilon_\lambda }_{\rho_{\lambda}}} 
\end{align*}

\end{lemma}
\begin{proof}
    The structure of the proof is exactly the same as the proof of Lemma \ref{lem:decouptraceineq}, and only the specific expressions differ. We define the functional $F(\Upsilon_1,...,\Upsilon_k)$ exactly as in the proof of Lemma \ref{lem:decouptraceineq}, and proceed to prove the claim for when  $\rho_1= \cdots =\rho_{k-1}=\infty$ and $\rho_k=\frac{q}{q-k}$.

    Let $\eta = (\eta(1), \eta(2))$ be a uniformly distibuted random variable in $\{ (l, \gamma_l) | l \in [n], \gamma_l \in [s_l] \}$ and for all $\lambda \in [k]$, define random variables $\Theta_{1, \lambda} = \Theta_{\eta, \lambda, 1}$ and $\Theta_{2, \lambda} = \Theta_{\eta, \lambda, 2}$. Then, for $\mathcal{S} = \sum_{l=1}^n s_l$,
    
    \begin{align*}
        &\sum_{l \in [n], \gamma_l \in [s_l]} \E[ \tr \Theta_{(l,\gamma_l), 1, \tau_1(1)}^T\Theta_{(l,\gamma_l), 1, \tau_1(2)}
	\Upsilon_1\Theta_{(l,\gamma_l), 2, \tau_2(1)}^T\Theta_{(l,\gamma_l), 2, \tau_2(2)}\\&\cdots
	\Upsilon_2\Theta_{(l,\gamma_l), k, \tau_k(1)}^T\Theta_{(l,\gamma_l), k, \tau_k(2)}\Upsilon_k] \\ =
    &\mathcal{S} \cdot \E[
	\tr \Theta_{\tau_{1}(1),1}^T\Theta_{\tau_{1}(2),1}
	\Upsilon_1\Theta_{\tau_{2}(1),2}^T\Theta_{\tau_{2}(2),2}\cdots
	\Upsilon_2\Theta_{\tau_{k}(1),k}^T\Theta_{\tau_{k}(2),k}\Upsilon_k]
    \end{align*}

Defining factor 1 exactly as in Lemma \ref{lem:decouptraceineq}, we have,
\begin{align*}
(\text{factor 1})^2 \le & \norm{\Theta_{\tau_{1}(2),1}^T\Theta_{\tau_{1}(1),1}}_{2q}^{2-2/r} \prod\limits_{\lambda  = 2}^{k/2}\norm{\Theta_{\tau_{\lambda}(1),\lambda}^T\Theta_{\tau_{\lambda}(2),\lambda}}_{2q}^2\prod\limits_{\lambda  = 1}^{k/2} {\norm{\Upsilon_\lambda }_{\infty}^2} \\ &\cdot \norm{|\Theta_{\tau_{1}(2),1}^T\Theta_{\tau_{1}(1),1}|^{1/r} |\Upsilon_k| |\Theta_{\tau_{1}(2),1}^T\Theta_{\tau_{1}(1),1}|^{1/r}}_r
\end{align*}

To proceed, we prove an analogue of Lemma \ref{theta1theta2} to bound $\norm{\Theta_{\tau_\lambda(1),\lambda}^T\Theta_{\tau_\lambda(2),\lambda}}_{q} = \norm{\Theta_{1,\lambda}^T\Theta_{2,\lambda}}_{q}$.

\begin{lemma}[Product of Random Rows for LESS] \label{theta1theta2less}
Let $\Theta_{(l,\gamma), \lambda, 1}, \Theta_{(l,\gamma), \lambda, 2} \in \R^d$ be as in Lemma \ref{lem:traceineqless}. Let $\eta = (\eta(1), \eta(2))$ be a uniformly distributed random variable in $\{ (l, \gamma_l) | l \in [n], \gamma_l \in [s_l] \}$ such that $\eta$ is independent with $\{\mathcal{Z}_{\lambda}\}_{\lambda \in [k]}$, $S_1, S_2, G_1, G_2$. For all $\lambda \in [k]$, define random vectors $\Theta_{1, \lambda} = \Theta_{\eta, \lambda, 1}$ and $\Theta_{2, \lambda} = \Theta_{\eta, \lambda, 2}$. Let $q \ge 2$. Then there exists a constant $c_{\ref*{theta1theta2less}} >0$ such that,
\begin{align*}
\norm{\Theta_{1,\lambda}^T\Theta_{2,\lambda}}_{q} &\le  \frac{ c_{\ref*{theta1theta2less}}(pm)^\frac{1}{q}\sqrt{\max \{ pd, q^2 \}} }{\mathcal{S}^\frac{1}{q}}
\end{align*}
\end{lemma}

\begin{proof}

Note that for each fixed $1 \le \lambda \le k/2$, 
\begin{align*}
    \norm{\Theta_{1,\lambda}^T\Theta_{2,\lambda}}_{q} &= \norm{\xi_{\eta, \lambda}\alpha_{\eta} u_{\eta(1)} e_{\mu_{\eta,\lambda}}^TS_2(t)U}_{q} \\
    &\le \left( \E [\tr \abs{\xi_{\eta, \lambda}\alpha_{\eta} u_{\eta(1)} e_{\mu_{\eta,\lambda}}^TS_2(t)U}^{q} ] \right)^\frac{1}{q} \\
    &= \left( \E [\tr \abs{\alpha_{\eta} u_{\eta(1)} e_{\mu_{\eta,\lambda}}^TS_2(t)U}^{q} ] \right)^\frac{1}{q} \\
    &\le \left( \E \left[ \frac{1}{d} \norm{\alpha_{\eta} u_{\eta(1)} }^{q}\norm{e_{\mu_{\eta,\lambda}}^TS_2(t)U}^{q} \right] \right)^\frac{1}{q} \\
\end{align*} 
Conditioning on $\eta$, we take expectation of the second factor over $\mu_{\eta, \lambda}$ and $S_2$,
\begin{align*}
    \norm{\Theta_{1,\lambda}^T\Theta_{2,\lambda}}_{q} &\le \left( \E_{\eta} \left[ \frac{1}{d} \norm{\alpha_{\eta} u_{\eta(1)} }^{q} \E_{\mu_{\eta, \lambda}, S_2(t)}[ \norm{e_{\mu_{\eta,\lambda}}^TS_2(t)U}^{q}] \right] \right)^\frac{1}{q} \\
\end{align*}

Note that, conditioned on $\eta$, $\mu_{\eta, \lambda}$ is uniformly distributed over a subset of $[m]$. Therefore, by Lemma \ref{lem:lessrownormbound}, we have  
\begin{align*}
    \norm{\Theta_{1,\lambda}^T\Theta_{2,\lambda}}_{q} &\le \left( \E_{\eta} \left[ \frac{1}{d} \norm{\alpha_{\eta} u_{\eta(1)} }^{q} \E_{\mu_{\eta, \lambda}, S_2(t)}[ \norm{e_{\mu_{\eta,\lambda}}^TS_2(t)U}^{q}] \right] \right)^\frac{1}{q} \\
    &\le c_{\ref*{lem:lessrownormbound}}\sqrt{\max \{ pd, q^2 \}} \left( \E_{\eta} \left[ \frac{1}{d} \norm{\alpha_{\eta} u_{\eta(1)} }^{q} \right] \right)^\frac{1}{q} \\
    &\le c_{\ref*{lem:lessrownormbound}}\sqrt{\max \{ pd, q^2 \}} \left( \frac{1}{\mathcal{S}d} \sum_{l=1}^n \sum_{\gamma_l=1}^{s_l} \abs{\alpha_{(l,\gamma_l)}}^{q} \norm{u_l}^{q} \right)^\frac{1}{q} \\
\end{align*}

Note that when $b_l \ge m$, we have $s_l=1$ and $\abs{\alpha_{(l,\gamma_l)}} \le \sqrt{pm}$. Moreover, $b_l \ge m$ only when $1/\beta_1pz_l \ge m$, which means we have $pm\norm{u_l}^2 \le \norm{u_l}^2/\beta_1z_l \le 1$. This means, $\abs{\alpha_{(l,\gamma_l)}}^{q} \norm{u_l}^{q} \le (pm\norm{u_l}^2)^{q/2} \le pm\norm{u_l}^2$.

When $b_l < m$, we have $\abs{\alpha_{(l,\gamma_l)}}\norm{u_l} \le \sqrt{pb_j}\norm{u_l} \le \norm{u_l}/\sqrt{\beta_1z_j} \le 1$ when $\lfloor 1/\beta_1pz_j \rfloor \ge 1$ and $b_j = \lfloor 1/\beta_1pz_j \rfloor$ because $pb_j \le 1/\beta_1z_j \le 1/\norm{u_j}^2$.
    When $b_j=1$, we still have $\abs{\alpha_{(l,\gamma_l)}}\norm{u_l} \le \sqrt{pb_j}\norm{u_l} \le \sqrt{p}\norm{u_l} \le 1$ since $p \le 1$.

So, we have $\sum_{\gamma_l=1}^{s_l} \abs{\alpha_{(l,\gamma_l)}}^{q} \norm{u_l}^{q} \le \sum_{\gamma_l=1}^{s_l} \abs{\alpha_{(l,\gamma_l)}}^{2} \norm{u_l}^{2} \le pb_l s_l \norm{u_l}^{2}$ (Since $\alpha_{(l,\gamma_l)} \le \sqrt{p b_l}$ by definition). But $s_l = \lceil m/b_l \rceil \le 2m/b_l$. So, $pb_l s_l \norm{u_l}^{2} \le 2pm \norm{u_l}^{2} $.

We can now continue to bound,
\begin{align*}
\norm{\Theta_{\tau_j(1),j}^T\Theta_{\tau_j(2),j}}_{q} &\le c_{\ref*{lem:lessrownormbound}}\sqrt{\max \{ pd, q^2 \}} \left( \frac{1}{\mathcal{S}d} \sum_{l=1}^n \sum_{\gamma_l=1}^{s_l} \abs{\alpha_{(l,\gamma_l)}}^{q} \norm{u_l}^{q} \right)^\frac{1}{q} \\
    &\le c_{\ref*{lem:lessrownormbound}}\sqrt{\max \{ pd, q^2 \}} \left( \frac{1}{\mathcal{S}d} \sum_{l=1}^n  2pm\norm{u_l}^2   \right)^\frac{1}{q} \\
    &\le c_{\ref*{lem:lessrownormbound}}\sqrt{\max \{ pd, q^2 \}} \left( \frac{2pmd}{\mathcal{S}d}  \right)^\frac{1}{q} \\
    &\le c_{\ref*{lem:lessrownormbound}}\sqrt{\max \{ pd, q^2 \}} \left( \frac{1}{\mathcal{S}}  \right)^\frac{1}{q} \left( 2pm  \right)^\frac{1}{q}
\end{align*}
\end{proof}

Using the same calculations as Lemma \ref{lem:decouptraceineq} for the term \begin{align*}
    \norm{|\Theta_{\tau_{1}(2),1}^T\Theta_{\tau_{1}(1),1}|^{1/r} |\Upsilon_k| |\Theta_{\tau_{1}(2),1}^T\Theta_{\tau_{1}(1),1}|^{1/r}}_r
\end{align*} we get
\begin{align*}
    (\text{factor 1})^2
 &\le  \paren*{\frac{ c_{\ref*{theta1theta2less}}(pm)^\frac{1}{q}\sqrt{\max \{pd, q^2 \}} }{\mathcal{S}^\frac{1}{2q}}}^{k-2/r}\prod\limits_{\lambda  = 1}^{k/2} {\norm{\Upsilon_\lambda }_{\infty}^2} \cdot  \paren*{\frac{1}{\cS} \norm{|\Upsilon_k|^{r}}_{\frac{2q}{2q-2}} \cdot \paren*{\E \tr((\Gamma(t))^{2q})}^\frac{1}{q}}^\frac{1}{r} \\
 &\le \frac{1}{\cS} \paren*{ c_{\ref*{theta1theta2less}}(pm)^\frac{1}{q}\sqrt{\max \{pd, q^2 \}}  }^{k-2/r}\prod\limits_{\lambda  = 1}^{k/2} {\norm{\Upsilon_\lambda }_{\infty}^2} \cdot   \norm{\Upsilon_k}_{\frac{2q}{2q-k}} \cdot \paren*{\E \tr((\Gamma(t))^{2q})}^\frac{1}{qr}  \\
\end{align*}
By repeating the same argument for factor 2, we get,
\begin{align*}
    &\sum_{l \in [n], \gamma_l \in [s_l]} \E[ \tr \Theta_{(l,\gamma_l), 1, \tau_1(1)}^T\Theta_{(l,\gamma_l), 1, \tau_1(2)}
	\Upsilon_1\Theta_{(l,\gamma_l), 2, \tau_2(1)}^T\Theta_{(l,\gamma_l), 2, \tau_2(2)} \\&\cdots
	\Upsilon_2\Theta_{(l,\gamma_l), k, \tau_k(1)}^T\Theta_{(l,\gamma_l), k, \tau_k(2)}\Upsilon_k]  \\ \le &
       (c_{\ref*{theta1theta2less}}(pm)^\frac{1}{q}\sqrt{\max\{pd,q^2\}})^{\frac{2qk-4q}{2q-2}} (\E \tr((\Gamma(t))^{2q}))^{\frac{1}{q} \cdot \frac{2q-k}{2q-2}} \prod\limits_{\lambda  = 1}^{k-1} {\norm{\Upsilon_\lambda }_{\infty}} \cdot \norm{\Upsilon_k}_{\frac{2q}{2q-k}}
\end{align*}
\end{proof}

\begin{lemma}[Row Norm for LESS-IC]\label{lem:lessrownormbound}
    Let $S(t):= \sqrt{t}S + \sqrt{1-t}G$, where $S$ has the fully independent unscaled LESS-IC distribution as in Lemma \ref{lem:diffineqless} and $G$ is an $m \times n$ matrix with i.i.d. Gausian entries with variance $p$. Let $\mu$ be a random variable uniformly distributed in $\phi \neq I \subset [m]$ and independent of $S$ and $G$. Then, there exists $c_{\ref*{lem:lessrownormbound}} > 0$, such that 
    \begin{align*}
        \E_{\mu, S(t)} [ \norm{e_{\mu}^TS(t)U}^{q} ]^\frac{1}{q} &\le c_{\ref*{lem:lessrownormbound}}\sqrt{\max \{ pd, q^2 \}} 
    \end{align*}
\end{lemma}
\begin{proof}
    By Hölder's inequality, it suffices to prove these bounds for moments of the order of the smallest even integer bigger than $q$, so without loss of generality, we may assume that $q$ is an even integer. Moreover,
    \begin{align*}
        \E_{\mu, S(t)} [ \norm{e_{\mu}^TS(t)U}^{q}] =& \frac{1}{\abs{I}}\sum_{i \in I} \E_{S(t)} [ \norm{e_{i}^TS(t)U}^{q}] \\
    \end{align*}
    In what follows, we fix $i \in I$, and obtain a bound for this fixed $i$. We observe that, 
\begin{align*}
    \E_{S(t)} [ \norm{e_{i}^TS(t)U}^{q}]^\frac{1}{q} &= \E_{S(t)} [ \norm{e_{i}^T(\sqrt{t}S + \sqrt{1-t}G)U}^{q}]^\frac{1}{q} \\
    &\le \E_{S(t)} [ \paren*{ \norm{e_{i}^TSU} + \norm{e_{i}^TGU} }^{q}]^\frac{1}{q} \\
    &\le \E_{S} [ \norm{e_{i}^TSU}^{q}]^\frac{1}{q} + \E_{G} [ \norm{e_{i}^TGU}^{q}]^\frac{1}{q}
\end{align*}
Since $\E_{G} [ \norm{e_{i}^TGU}^{q}]^\frac{1}{q}$ is simply the $q\textsuperscript{th}$ moment of a $d$-dimensional Gaussian random vector whose independent components have variance $p$,  $\E_{G} [ \norm{e_{i}^TGU}^{q}]^\frac{1}{q} \le c_1\sqrt{\max \{pd, pq \}}$, and it is enough to prove the bound in the statement of the lemma for $\E_{S} [ \norm{e_{i}^TSU}^{q}]^\frac{1}{q}$.

Note that in the case when $S$ has the unscaled LESS-IC distribution, $S_{1,1}$ is non-zero when the one hot distribution on the column submatrix $S_{[1:\min \{b_1, m \}]\times{1}}$ has it's non zero entry on the first row. The probability that this happens is $1/\min \{b_1, m \} = p/\alpha_{1,1}^2$. To generalise this to the entry $S_{i,j}$ for a given $(i,j)$, observe that for each $(i,j)$, we can associate a unique tuple $(j,r)$ with $r=r(i,j)$ such that the random variable $\mu_{j,r}$ determines whether $S_{ij}$ is non zero. $\mu_{j,r}$ is uniformly distributed in an interval of size $\alpha^2_{j,r}/p$, so the probability that $S_{ij}$ is non-zero is $p/\alpha^2_{j,r}$. Moreover, we assume the columns of $S$ are independent and remark that the calculations hold even for $\log$-wise independence since $\norm{e_{i}^TS(t)U}^{q}$ is a polynomial of order $q = O(\log (d/\varepsilon\delta))$ in the entries of the first row of $S$ only sees products of $O(\log (d/\varepsilon\delta))$ many random variables.

From the above discussion, we can write,
\begin{align*}
    e_{i}^TSU &= \sum_{j=1}^n \delta_{ij} \alpha_{j,r(i,j)} \xi_{i,j} u_j^T \\
    &=: \sum_{j=1}^n Z_{ij}
\end{align*}
where,
\begin{itemize}
    \item $\{ \delta_{ij} \}_{j \in [n]}$ are independent Bernoulli random variables which are non zero with probability $p/\alpha^2_{j,r(i,j)}$ respectively.
    \item $\{ \xi_{i,j} \}_{j \in [n]}$ is a collection of independent Rademacher random variables.
    \item $\{ u_j^T \}_{j \in [n]}$ are rows of the matrix $U$.
\end{itemize}

We shall use a recursive relation to estimate $\E \left[ \left\|\sum_{j=1}^n  Z_{ij} \right\|^{2q} \right]^\frac{1}{q}$, and proceed to obtain this recursive inequality,
    \begin{align} 
    \begin{split} 
        &\E \left[ \left\|\sum_{j=1}^n  Z_{ij} \right\|^{2q} \right]^\frac{1}{q} \\
        &= \E \left[ \left ( \left\|\sum_{j=1}^n  Z_{ij} \right\|^{2} \right)^q \right]^\frac{1}{q} \\
        &= \E \left[ \left ( \sum_{j=1}^n \delta_{ij}^2 \alpha^2_{j,r(i,j)} \| u_j \| ^2 + \sum_{ \substack{j_1,j_2 =1 \\ j_1 \neq j_2} }^n \alpha_{j_1,r(i,j_1)}\alpha_{j_2,r(i,j_2)} \delta_{ij_1} \delta_{ij_2} \xi_{i,j_1} \xi_{i,j_2} u_{j_1}^\top u_{j_2}  \right)^q \right]^\frac{1}{q} \\
        &\le \E \left[ \left ( \sum_{j=1}^n \delta_{ij}^2 \alpha^2_{j,r(i,j)} \| u_j \| ^2 \right)^q \right]^\frac{1}{q} +
        \E \left[ \left ( \sum_{ \substack{j_1,j_2 =1 \\ j_1 \neq j_2} }^n \alpha_{j_1,r(i,j_1)}\alpha_{j_2,r(i,j_2)} \delta_{ij_1} \delta_{ij_2} \xi_{i,j_1} \xi_{i,j_2} u_{j_1}^\top u_{j_2}  \right)^q \right]^\frac{1}{q} \label{eq:rnormrecur}
    \end{split}
    \end{align}
    Using decoupling (\cite[Theorem 6.1.1]{vershynin2018high}), 
    \begin{align*}
        &\E \left[ \left ( \sum_{ \substack{j_1,j_2 =1 \\ j_1 \neq j_2} }^n \alpha_{j_1,r(i,j_1)}\alpha_{j_2,r(i,j_2)} \delta_{ij_1} \delta_{ij_2} \xi_{i,j_1} \xi_{i,j_2} u_{j_1}^\top u_{j_2}  \right)^q \right] \\
        &\le \E \left[ \left ( 4 \sum_{ j_1,j_2 =1  }^n \alpha_{j_1,r(i,j_1)}\alpha_{j_2,r(i,j_2)} \delta_{ij_1} \widetilde{\delta}_{ij_2} \xi_{i,j_1} \widetilde{\xi}_{i,j_2} u_{j_1}^\top u_{j_2}  \right)^q \right] \\
        &= \E \left[ \left ( 4 \left\langle \alpha_{j,r(i,j)} \sum_{j=1}^n \delta_{ij} \xi_{ij} u_{j}, \sum_{j=1}^n \alpha_{j,r(i,j)} \widetilde{\delta}_{ij} \widetilde{\xi}_{ij} u_{j} \right\rangle  \right)^q \right] \\
    \end{align*}
    where $\{ \widetilde{\delta}_{ij}, \widetilde{\xi}_{ij} \}_{i \in [m], j \in [n]}$ are independent copies of $\{ {\delta}_{ij}, {\xi}_{ij} \}_{i \in [m], j \in [n]}$.
    
    Conditioning on $\{ \widetilde{\delta}_{ij}, \widetilde{\xi}_{ij} \}_{i \in [m], j \in [n]}$, we have
    \begin{align*}
        &\E \left[ \left( 4 \left\langle \sum_{j=1}^n \alpha_{j,r(i,j)} \delta_{ij} \xi_{ij} u_{j}, \sum_{j=1}^n \alpha_{j,r(i,j)} \widetilde{\delta}_{ij} \widetilde{\xi}_{ij} u_{j} \right\rangle  \right)^q \right] 
        \\=& \E \left[ \E \left[ \left ( 4 \left\langle \sum_{j=1}^n \alpha_{j,r(i,j)} \delta_{ij} \xi_{ij} u_{j}, \sum_{j=1}^n \alpha_{j,r(i,j)} \widetilde{\delta}_{ij} \widetilde{\xi}_{ij} u_{j} \right\rangle  \right)^q \Bigg| \{ \widetilde{\delta}_{ij}, \widetilde{\xi}_{ij} \}_{i \in [m], j \in [n]} \right] \right] \\ = &
        \E \left[ \left\| \sum_{j=1}^n \alpha_{j,r(i,j)} \widetilde{\delta}_{ij} \widetilde{\xi}_{ij} u_{j} \right\|^q \E \left[ \left( 4 \left\langle \sum_{j=1}^n \alpha_{j,r(i,j)} \delta_{ij} \xi_{ij} u_{j}, v \right\rangle  \right)^q \right] \right]
    \end{align*}
    for some fixed unit vector $v \in \R^d$. Plugging this estimate back into \eqref{eq:rnormrecur}, 
    \begin{align*}
        \E \left[ \left\|\sum_{j=1}^n  Z_{ij} \right\|^{2q} \right]^\frac{1}{q} \le &\E \left[ \left ( \sum_{j=1}^n \delta_{ij}^2 \alpha^2_{j,r(i,j)} \| u_j \| ^2 \right)^q \right]^\frac{1}{q} \\&+ \E \left[ \left\|\sum_{j=1}^n  Z_{ij} \right\|^{q} \right]^\frac{1}{q} \E \left[ \left( 4 \left\langle \sum_{j=1}^n \alpha_{j,r(i,j)} \delta_{ij} \xi_{ij} u_{j}, v \right\rangle  \right)^q \right]^\frac{1}{q} 
    \end{align*}
    Both $\E \left[ \left ( \sum_{j=1}^n \delta_{ij}^2 \alpha_{j,r(i,j)}^2 \| u_j \| ^2 \right)^q \right]^\frac{1}{q}$ and $\E \left[ \left( 4 \left\langle \sum_{j=1}^n \alpha_{j,r(i,j)}\delta_{ij} \xi_{ij} u_{j}, v \right\rangle  \right)^q \right]^\frac{1}{q}$ shall be computed using tail probabilities obtained via Chernoff bounds. To this end, we first look at $\E [ \exp (4\lambda \langle \alpha_{j,r(i,j)} \delta_{ij} \xi_{ij} u_{j}, v \rangle ) ]$ for a fixed $j$,
    \begin{align*}
        \E \left[ \exp \left(4\lambda \left\langle \alpha_{j,r(i,j)} \delta_{ij} \xi_{ij} u_{j}, v \right\rangle \right) \right] &= 1 + \frac{p}{\alpha^2_{j,r(i,j)}} \left( \cosh\left(  4\lambda\alpha_{j,r(i,j)} \left\langle  u_{j}, v \right\rangle \right) - 1\right) \\
        &= 1 + \frac{\cosh\left(  4\lambda\alpha_{j,r(i,j)} \left\langle  u_{j}, v \right\rangle \right) - 1}{16\lambda^2\alpha^2_{j,r(i,j)}\langle  u_{j}, v \rangle^2} \cdot 16\lambda^2 \langle  u_{j}, v \rangle^2 p \\ 
    \end{align*}
    Since $\cosh$ is an even function, we may assume that $\langle  u_{j}, v \rangle \ge 0$. Then, $4\lambda \alpha_{j,r(i,j)} \left\langle  u_{j}, v \right\rangle  \le 4\lambda \alpha_{j,r(i,j)} \norm{u_j} \le 4\lambda \sqrt{pb_j} \norm{u_j} \le 4\lambda$ when $\lfloor 1/\beta_1pz_j \rfloor \ge 1$ and $b_j = \lfloor 1/\beta_1pz_j \rfloor$ because $pb_j \le 1/\beta_1z_j \le 1/\norm{u_j}^2$.
    When $b_j=1$, we still have $4\lambda \alpha_{j,r(i,j)} \left\langle  u_{j}, v \right\rangle   \le 4\lambda \sqrt{p} \cdot 1 \le 4\lambda$.
    
    Using the fact that $\frac{\cosh(x)-1}{x^2}$ is increasing in $x$ for $x \ge 0$, we get,
    \begin{align*}
        \E \left[ \exp \left(4\lambda \left\langle \alpha_{j,r(i,j)} \delta_{ij} \xi_{ij} u_{j}, v \right\rangle \right) \right] &\le 1 + \frac{\cosh(4\lambda)-1}{16\lambda^2} \cdot 16\lambda^2 \langle  u_{j}, v \rangle^2 p \\
        &\le 1 + \langle  u_{j}, v \rangle^2 p e^{4\lambda} \\
        &\le \exp( \langle  u_{j}, v \rangle^2 p e^{4\lambda} )
    \end{align*}
    So,
    \begin{align*}
        \E \left[ \exp \left(4\lambda \left\langle \sum_{j=1}^n\alpha_{j,r(i,j)} \delta_{ij} \xi_{ij} u_{j}, v \right\rangle \right) \right] 
        &\le \exp \left( \sum_{j=1}^n \langle  u_{j}, v \rangle^2 p e^{4\lambda} \right) \\
        &\le \exp(p e^{4\lambda})
    \end{align*}
    This gives a Chernoff tail bound,
    \begin{align*}
        \Pb \left[ \left\lvert \left\langle \sum_{j=1}^n\alpha_{j,r(i,j)} \delta_{ij} \xi_{ij} u_{j}, v \right\rangle \right\rvert \ge t \right] \le 2\exp \left( \frac{t}{4} \left( 1 - \log \left( \frac{t}{4p} \right) \right) \right)
    \end{align*}
    and a standard moment computation gives, $\E \left[ \left( 4 \left\langle \sum_{j=1}^n \alpha_{j,r(i,j)}\delta_{ij} \xi_{ij} u_{j}, v \right\rangle  \right)^q \right]^\frac{1}{q} \le C_1 q $.

    To deal with $\E \left[ \left ( \sum_{j=1}^n \delta_{ij}^2 \alpha_{j,r(i,j)}^2 \| u_j \| ^2 \right)^q \right]^\frac{1}{q}$, we use the following Rosenthal’s inequality from \cite[p.442]{boucheron2013concentration}.

\begin{lemma}[Theorem 14.10 from \cite{boucheron2013concentration}]\label{lem:rosenthal}
Let $Z=\sum_{j=1}^nX_i$ where $X_1,...,X_n$ are independent and nonnegative random variables. Then the exists a constant $_{\ref{lem:rosenthal}}$ such that, for all integers $q \ge 1$, we have
\begin{align*}
    (\E (Z^q))^{1/q} \le 2 \E( Z)+ c_{\ref*{lem:rosenthal}}q (\E ((\max \limits_{j=1,...,n} X_j)^q))^{1/q}
\end{align*}
\end{lemma}

We use this lemma with $X_j=\delta_{ij}^2 \alpha_{j,r(i,j)}^2 \| u_j \| ^2$. We first calculate that
\begin{align*}
    \E (\sum_{j=1}^n \delta_{ij}^2 \alpha_{j,r(i,j)}^2 \| u_j \| ^2) = \sum_{j=1}^n (p/\alpha_{j,r(i,j)}^2)\alpha_{j,r(i,j)}^2\| u_j \| ^2 =pd
\end{align*}

Next, we calculate $(\E ((\max \limits_{j=1,...,n} X_j)^q))^{1/q}$. Since $\alpha_{j,r(i,j)} \le \sqrt{pb_j}$ and $\delta_{ij} \le 1$, we have $X_j=\delta_{ij}^2 \alpha_{j,r(i,j)}^2 \| u_j \| ^2 \le pb_j \| u_j \| ^2$. Now we consider two cases depending on whether $b_j=\left\lfloor \frac{1}{\beta_1pz_j} \right\rfloor$ or $b_j=1$. First, if $b_j=\left\lfloor \frac{1}{\beta_1pz_j} \right\rfloor$, then we have $X_j \le pb_j \| u_j \| ^2 \le p \frac{1}{\beta_1pz_j} \beta_1z_j \le 1$. Second, if $b_j=1$, then we have $X_j \le pb_j \| u_j \| ^2 \le p  \le 1$. In conclusion, we have $\max \limits_{j=1,...,n} \delta_{ij}^2 \alpha_{j,r(i,j)}^2 \| u_j \| ^2 \le 1$ almost surely, and therefore we have $(\E ((\max \limits_{j=1,...,n} \delta_{ij}^2 \alpha_{j,r(i,j)}^2 \| u_j \| ^2)^q))^{1/q} \le 1$.

Therefore, we have
\begin{align*}
    \E \left[ \left ( \sum_{j=1}^n \delta_{ij}^2 \alpha_{j,r(i,j)}^2 \| u_j \| ^2 \right)^q \right]^\frac{1}{q} \le 2pd+c_2q
\end{align*}
for some constant $c_2$.

    We thus have,
    \begin{align*}
        \E \left[ \left\|\sum_{j=1}^n  Z_{ij} \right\|^{2q} \right]^\frac{1}{2q} &\le \sqrt{2 pd} + \sqrt{c_2q} + \E \left[ \left\|\sum_{j=1}^n  Z_{ij} \right\|^{q} \right]^\frac{1}{2q} \sqrt{C_1q} 
    \end{align*}
    Let $C_2=\max\{2,c_2,(3C_1)^2\}$.   Observe that $ \E \left[ \left\|\sum_{j=1}^n  Z_{ij} \right\|^{2} \right]^\frac{1}{2} = \sqrt{pd}$. Assuming that \begin{align*}\E \left[ \left\|\sum_{j=1}^n  Z_{ij} \right\|^{q} \right]^\frac{1}{q} \le 3\sqrt{C_2} \sqrt{\max \{pd, q^2 \}} \end{align*} for some $q$, the above relation gives,
    \begin{align*}
        \E \left[ \left\|\sum_{j=1}^n  Z_{ij} \right\|^{2q} \right]^\frac{1}{2q} &\le \sqrt{C_2pd} + \sqrt{C_2q} + \left( 3\sqrt{C_2}\sqrt{\max \{pd, q^2 \}} \right)^\frac{1}{2} \sqrt{C_1q} \\
        &\le \left( 2\sqrt{C_2} + \sqrt{3C_1}C_2^{1/4} \right) \sqrt{\max \{pd, q^2 \}} \\
        &\le 3\sqrt{C_2} \sqrt{\max \{pd, q^2 \}}
    \end{align*}

\end{proof}
\section{Oblivious Subspace Embedding with Independent Entries} \label{sec:oseieproof}

In this section, we consider the subspace embedding property for the following classical model with independent entries in the matrix $S$.

\begin{definition}[OSE-IE]\label{def:oseie}
An $m \times n$ random matrix $S$ is called an unscaled oblivious subspace embedding with independent entries (unscaled OSE-IE) with parameter $p$ if $S$ has i.i.d. entries $s_{i,j}=\delta_{(i,j)} \xi_{(i,j)}$ where $\delta_{(i,j)}$ are i.i.d. Bernoulli random variables taking value 1 with probability $p \in (0,1]$ and $\xi_{(i,j)}$ are i.i.d. random variables independent with $\delta_{(i,j)}$ and satisfy $\Pb(\xi_{(i,j)}=1)=\Pb(\xi_{(i,j)}=-1)=1/2$. And in this case, $\Pi = (1/\sqrt{pm})S$ is call an OSE-IE with parameter $p$.
\end{definition}

For this model, we have the following subspace embedding guarantee,

\begin{restatable}[High Probability Bounds for the Embedding Error for OSE-IE]{theorem}{oseiemainthm} \label{thm:oseiemain}
   Let $S$ be an $m \times n$ matrix distributed according to the unscaled OSE-IE distribution with parameter $p$. Let $U$ be an arbitrary $n \times d$ deterministic matrix such that $U^TU=I$. Then, there exist constants $c_{\ref*{thm:oseiemain}.1}>0$ and $c_{\ref*{thm:oseiemain}.2}>0$ such that for any $0 < \varepsilon , \delta < 1$ and $d>10$,
we have
\begin{equation}
    \begin{aligned}
\Pb \left( 1 - \varepsilon  \leq s_{\min}((1/\sqrt{pm})SU)   \leq s_{\max}((1/\sqrt{pm})SU) \leq 1 + \varepsilon \right) \geq 1-\delta
\end{aligned}
\end{equation}
when $m > c_{\ref*{thm:oseiemain}.1}  \frac{d + \log(1/\varepsilon\delta)}{\varepsilon^2} $  and $$pm \geq \min \left\{ c_{\ref*{thm:oseiemain}.2}\paren*{\frac{(\log (d/\varepsilon\delta))^2}{\varepsilon}+\frac{(\log (d/\varepsilon\delta))}{\varepsilon^2} + (\log (d/\varepsilon\delta))^3 } ,m \right\}$$
\end{restatable}

The overall structure of the proof of Theorem \ref{thm:oseiemain} is the same as the proof in the OSNAP case and we only highlight the differences from the proof of the OSNAP case discussed in Section \ref{sec:osnap-proof}. A key difference between the above result and the corresponding result for OSNAP is the $1/\varepsilon^2$ dependence in the lower bound for sparsity. This arises due to our approach of decomposing $U^TS^TSU - pm\cdot I_d$ as
\[ U^TS^TSU - pm\cdot I_d  =  \sum_{j=1}^n \left( \sum_{i=1}^m s_{ij}^2 \right) u_ju_j^T  - pm\cdot I_d + \sum_{i=1}^m \sum_{\substack{j,j' =1 \\ j \neq j'}}^n s_{ij}s_{ij'} u_ju_{j'}^T  \]
where we label the former term as the diagonal term and the latter as the off diagonal term. In the OSNAP model, we have $\sum_{i=1}^m s_{ij}^2 = pm$ by construction and therefore the diagonal term vanishes, but when $S$ has the OSE-IE distribution, we need not necessarily have $\sum_{i=1}^m s_{ij}^2 = pm$, which means we need to analyze diagonal term in addition to the off-diagonal term analyzed in Lemma \ref{lem:decoup}. Observing that, 
\[ \sum_{j=1}^n \left( \sum_{i=1}^m s_{ij}^2 \right) u_ju_j^T  - pm\cdot I_d = \sum_{j=1}^n \left( \sum_{i=1}^m \paren*{s_{ij}^2 - p} \right) u_ju_j^T \]
and using Minkowski's inequality,
\begin{align*}
    \E [ \tr (U^TS^TSU - pm\cdot I_d)^{2q} ]^\frac{1}{2q} &\le \E \left[ \tr \left( \sum_{i=1}^m \sum_{j =1}^n (s_{ij}^2 - p)u_ju_{j}^T \right)^{2q} \right]^\frac{1}{2q} \\
    & \qquad + \E \left[ \tr \left( \sum_{i=1}^m \sum_{j,j' =1, j \neq j'}^n s_{ij}s_{ij'} u_ju_{j'}^T \right)^{2q} \right]^\frac{1}{2q} \\
\end{align*}
Proposition \ref{prop:momestdiag} controls the diagonal term (the former term), which gives rise to the $1/\varepsilon^2$ dependence, and is analyzed in Section \ref{subsec:oseiediag}. The analysis of the off-diagonal term (the latter term) is similar to the OSNAP case and is discussed in Section \ref{subsec:oseiemain}.

\subsection{Controlling the diagonal term when $S$ is the OSE-IE distribution} \label{subsec:oseiediag}

\begin{proposition}[Diagonal Term] \label{prop:momestdiag}
    Let $S$ be an $m \times n$ matrix distributed according to the unscaled OSE-IE distribution with parameter $p$. Let $U$ be an arbitrary $n \times d$ deterministic matrix such that $U^TU=I$. Given $0< \varepsilon < 1$ and $q \ge \log(d) $, there exist constants $c_{\ref*{prop:momestdiag}}$ such that for $m \geq d \ge 20$ and $pm \geq c_{\ref*{prop:momestdiag}} \frac{q }{\varepsilon^2} $,

\[ \E \left[ \tr \left( \frac{1}{pm }\sum_{i=1}^m \sum_{j =1}^n (s_{ij}^2 - p)u_ju_{j}^T \right)^{2q} \right]^\frac{1}{2q} \leq  \varepsilon \]

\end{proposition}
Proposition \ref{prop:momestdiag} shows that, to make the diagonal term small, it suffices to require the nonzero entries per column to be greater than or equal to $c\frac{\log(d)}{\varepsilon^2}$. Note that
    \[ \E \left[ \tr \left( \frac{1}{pm }\sum_{i=1}^m \sum_{j =1}^n (s_{ij}^2 - p)u_ju_{j}^T \right)^{2q} \right]^\frac{1}{2q} \ge  \E \left[ \tr \left( \frac{1}{pm }\sum_{i=1}^m \sum_{j =1}^n (s_{ij}^2 - p)u_ju_{j}^T \right)^{2} \right]^\frac{1}{2} \]
    and the latter can be of the order $1/\sqrt{pm}$ when $U$ corresponds to a $d$-dimensional coordinate subspace of $\R^n$. Thus, the $1/\varepsilon^2$ dependence in the lower bound on $pm$ is necessary when using this approach to prove the subspace embedding guarantee.

\begin{proof}
    First, note that, 
    \[ \E \left[ \tr \left( \sum_{i=1}^m \sum_{j =1}^n (s_{ij}^2 - p)u_ju_{j}^T \right)^{2q} \right]^\frac{1}{2q} \le \E \left[  \norm*{\sum_{i=1}^m \sum_{j =1}^n (s_{ij}^2 - p)u_ju_{j}^T}^{2q} \right]^\frac{1}{2q} \]
    and we shall use Matrix Bernstein's inequality to estimate this moment. We shall only use a specific case of this theorem that holds more generally. 
    \begin{lemma}[Matrix Bernstein, Theorem 6.1.1 \cite{troppmatrixconc}]
        For a finite sequence of $d \times d$ independent mean zero symmetric bounded random matrices $\{ S_k \}$ with $\norm{S_k} \le R$, let,
        $$\sigma = \norm*{\sum_k \E[S^2_k]}$$
        Then,
        \[ \P \paren*{ \norm*{\sum_k S_k} \ge t } \le 2d \exp \paren*{\frac{-t^2/2}{\sigma + \frac{Rt}{3}}} \]
    \end{lemma}
    Recall that we are assuming that the entries of $S$ are fully independent, but since we are only dealing with moments of order up to $4q$, the same bounds will hold even if the entries of $S$ are $4q$-wise independent.
    
    In our case, for the sum $\sum_{i=1}^m \sum_{j =1}^n (s_{ij}^2 - p)u_ju_{j}^T$, we have $R = 1$ and,
    \begin{align*}
        \sigma &= \norm*{\sum_{i=1}^m \sum_{j =1}^n \E[(s_{ij}^2 - p)^2] \norm{u_j}^2u_ju_{j}^T } \\
        &= \norm*{\sum_{j =1}^n pm(1-p) \norm{u_j}^2u_ju_{j}^T } \\  
        &\le pm
    \end{align*}
    where we use the fact that $\sum_{j =1}^n (1-p) \norm{u_j}^2u_ju_{j}^T \prec \sum_{j =1}^n  u_ju_{j}^T = I_d$. Then,
    \begin{align*}
        \P \paren*{ \norm*{\sum_{i=1}^m \sum_{j =1}^n (s_{ij}^2 - p)u_ju_{j}^T} \ge t' } &\le 2d \exp \paren*{\frac{-(t')^2/2}{pm + \frac{t}{3}}} \\
        \P \paren*{ \norm*{\paren*{\frac{1}{pm}\sum_{i=1}^m \sum_{j =1}^n s_{ij}^2u_ju_{j}^T} -  I_d } \ge \frac{t'}{pm} } &\le 2d \exp \paren*{\frac{-(t')^2/2}{pm + \frac{t}{3}}} \\
        \P \paren*{ \norm*{ \paren*{\frac{1}{pm}\sum_{i=1}^m \sum_{j =1}^n s_{ij}^2u_ju_{j}^T} -  I_d } \ge t } &\le 2d \exp \paren*{\frac{-(pmt)^2/2}{pm + \frac{pmt}{3}}} \\
        &= 2d \exp \paren*{\frac{-pmt^2/2}{1 + \frac{t}{3}}} \\
    \end{align*}
    Let $X$ denote $\norm*{\paren*{\frac{1}{pm}\sum_{i=1}^m \sum_{j =1}^n s_{ij}^2u_ju_{j}^T} -  I_d}$ for convenience. Then we have the mixed tail,
    \begin{align*}
        \P(X \ge t) \le \begin{cases}
            2d \exp \paren*{-pmt^2/4} \text{ for } 0 \le t \le 3\\
            2d \exp \paren*{-3pmt/4} \text{ for } 3 < t\\
        \end{cases}
    \end{align*}
    We are interested in $\E[X^{2q}]^\frac{1}{2q}$, and, for $\varepsilon < 1$,
    \begin{align*}
        \E[X^{2q}] &= 2q \int_0^\infty t^{2q-1} \P(X \ge t) dt \\
        &= 2q \int_0^\varepsilon t^{2q-1} \P(X \ge t) dt + 2q \int_\varepsilon^3 t^{2q-1} \P(X \ge t) dt + 2q \int_3^\infty t^{2q-1} \P(X \ge t) dt \\
        &\le 2q \cdot \varepsilon^{2q} + 2q \int_\varepsilon^3 t^{2q-1} \P(X \ge t) dt + 2q \int_3^\infty t^{2q-1} \P(X \ge t) dt \\
    \end{align*}
    When $pm \ge 4q/\varepsilon^2$, 
    \begin{align*}
        2q \int_\varepsilon^3 t^{2q-1} \P(X \ge t) dt &\le 2q \int_\varepsilon^3 t^{2q-1} 2d \exp \paren*{-pmt^2/4} dt \\
        &\le 2q \int_\varepsilon^3 t^{2q-1} 2d \exp \paren*{-2q t^2/2\varepsilon^2} dt \\
    \end{align*}
    Using change of variables $\hat{t} = \sqrt{2q}t/\varepsilon$,
    \begin{align*}
        2q \int_{\varepsilon}^3 t^{2q-1} \P(X \ge t) dt 
        &\le 2q \cdot 2d \int_0^\infty \paren*{\frac{\varepsilon}{\sqrt{2q}}}^{2q} \hat{t}^{2q-1}  \exp \paren*{-\hat{t}^2/2} d\hat{t} \\
        &\le 4qd \paren*{\frac{\varepsilon}{\sqrt{2q}}}^{2q} \int_0^\infty  \hat{t}^{2q-1}  \exp \paren*{-\hat{t}^2/2} d\hat{t} \\
        &\le \sqrt{2\pi} 4qd \paren*{\frac{\varepsilon \sqrt{2q-1} }{\sqrt{2q}}}^{2q}   \\
        &\le \sqrt{2\pi} 4qd \paren*{\varepsilon}^{2q}   \\
    \end{align*}
    Continuing,
    \begin{align*}
        2q \int_3^\infty t^{2q-1} \P(X \ge t) dt &\le 4qd \int_3^\infty t^{2q-1} \exp \paren*{-3pmt/4} dt \\
        &\le 4qd \int_\varepsilon^3 t^{2q-1} \exp \paren*{-3q t/\varepsilon^2} dt \\
    \end{align*}
    with the change of variables $\tilde{t} = 3qt/\varepsilon^2$,
    \begin{align*}
        2q \int_3^\infty t^{2q-1} \P(X \ge t) dt  &\le 4qd \paren*{\frac{\varepsilon^2}{3q}}^{2q} \int_0^\infty (\tilde{t})^{2q-1} \exp \paren*{-\tilde{t}} d\tilde{t} \\
        &\le 4qd \paren*{\frac{\varepsilon^2(2q-1)}{3q}}^{2q} \\
        &\le 4qd \paren*{\varepsilon^2}^{2q}
    \end{align*}
    So, $$\E[X^{2q}]^\frac{1}{2q} \le (18qd)^\frac{1}{2q} \varepsilon \le \exp \paren*{\frac{2\log d + \log 18}{2q}}  \varepsilon \le e^{1.5} \varepsilon$$

    By starting with $\varepsilon/e^{1.5}$ instead, we can get the claim made in the statement.
\end{proof}

\subsection{Proving the subspace embedding guarantee for the OSE-IE distribution} \label{subsec:oseiemain}

In this section, we prove \ref{thm:oseiemain}.

\oseiemainthm*

\begin{proof}
    By the proof of Theorem \ref{t:ose-full}, we see that it is enough to show that $\E[\tr(X^TX - I_d)^{2q}]^\frac{1}{2q} \leq  \varepsilon$, or, equivalently, $ \E [ \tr (U^TS^TSU - pm\cdot I_d)^{2q} ]^\frac{1}{2q} \le pm \varepsilon$ whenever $p$ and $q$ satisfy $ C\paren*{\frac{q^2}{\varepsilon} + q^3} \le pm$ and $m \ge C_1\frac{d+q}{\varepsilon^2}$.
    At the beginning of Section \ref{sec:oseieproof}, we saw that when $S$ has the OSE-IE distribution, 
    \begin{align*}
    \E [ \tr (U^TS^TSU - pm\cdot I_d)^{2q} ]^\frac{1}{2q} &\le \E \left[ \tr \left( \sum_{i=1}^m \sum_{j =1}^n (s_{ij}^2 - p)u_j u_{j}^T \right)^{2q} \right]^\frac{1}{2q} \\
    & \qquad + \E \left[ \tr \left( \sum_{i=1}^m \sum_{j,j' =1, j \neq j'}^n s_{ij}s_{ij'} u_j u_{j'}^T \right)^{2q} \right]^\frac{1}{2q} \\
    \end{align*}
    By Proposition \ref{prop:momestdiag}, the former term is bounded by $pm\varepsilon/2$ when $pm \geq 4c_{\ref*{prop:momestdiag}} \frac{q}{\varepsilon^2} $ and $m \ge 20$. For the latter term, we see that the proof of Lemma \ref{lem:decoup}  still applies and we have, 
    \[ \E \left[ \tr \left( \sum_{i=1}^m \sum_{j,j' =1, j \neq j'}^n s_{ij}s_{ij'} u_j u_{j'}^T \right)^{2q} \right]^\frac{1}{2q} \le 2 \E[\tr(\Gamma(S_1,S_2))^{2q}]^\frac{1}{2q} \]
    Following the proof of Theorem \ref{prop:momestdecoupmatrix} and using Lemma \ref{lem:decouptraceineqose}, we see that $\E[\tr(\Gamma(S_1,S_2))^{2q}]^\frac{1}{2q} \le \varepsilon/4$ when $ C\paren*{\frac{q^2}{\varepsilon} + q^3}^{1+\frac{2}{q-2}} \le pm$ and $m \ge C_1\frac{d+q}{\varepsilon^2}$. The claim follows as in the proof of Theorem \ref{t:ose-full}. 
\end{proof}

\subsection{Trace Inequality in the OSE-IE Case}
The trace inequality required to obtain the differential inequality for the interpolant between the moments of $(S_1U)^TS_2U$ and $(G_1U)^TG_2U$ has a slightly different proof in the OSE-IE case than in the OSNAP case, even though both bounds are the same.

\begin{lemma} \label{lem:decouptraceineqose}
    Let $S_1(t)$ and $S_2(t)$ be as in Lemma \ref{lem:diffineq} with both having the OSE-IE distribution. Let $\Gamma(t)=(S_1(t)U)^T(S_2(t)U)+(S_2(t)U)^T(S_1(t)U)$. Let $\mathcal{Z}=\{\xi_{(l,\gamma)}:(l,\gamma) \in [n] \times [m]\} \cup \{\delta_{(l,\gamma)}:(l,\gamma) \in [n] \times [m]\}$ be the family of mutually independent random variables generating an instance of $S_1$ with the OSE-IE distribution. Let $q \ge 2$ and $3 \le k \le 2q$. Let $\{\mathcal{Z}_{\lambda}\}_{\lambda \in [k]}$ be a family of (possibly dependent) random elements, where for each $\lambda \in [k]$, the random element
\begin{align*}
    \mathcal{Z}_{\lambda}=\{\xi_{(l,\gamma),\lambda}:(l,\gamma) \in [n] \times [m]\} \cup \{\delta_{(l,\gamma),\lambda}:(l,\gamma) \in [n] \times [m]\}
\end{align*}
has the same distribution as 
\begin{align*}
    \mathcal{Z}=\{\xi_{(l,\gamma)}:(l,\gamma) \in [n] \times [m]\} \cup \{\delta_{(l,\gamma)}:(l,\gamma) \in [n] \times [m]\}
\end{align*}
Let $Z_{(l,\gamma)}=\xi_{(l,\gamma)} \delta_{(l,\gamma)} e_{\gamma} e_l^T$ and $Z_{(l,\gamma),\lambda}= \xi_{(l,\gamma), \lambda} \delta_{(l,\gamma), \lambda} e_{\gamma} e_l^T$.
Let $\{\Upsilon_1,...,\Upsilon_k\}$ be a family of $L_{\infty}(S_{\infty}^d)$ random matrices.
Assume further that the collection $\{\mathcal{Z}_{\lambda}\}_{\lambda \in [k]}$ is independent of $S_1, S_2, G_1, G_2$, and $\{\Upsilon_1,...,\Upsilon_k\}$. (In other words, $\{\Upsilon_1,...,\Upsilon_k\}$ can possibly be dependent with $S_1, S_2, G_1, G_2$.)
For each $(l,\gamma) \in [n] \times [m]$ and $\lambda \in k$, we define random vectors $\Theta_{(l,\gamma), \lambda, 1}, \Theta_{(l,\gamma), \lambda, 2} \in \R^d$ such that
\begin{align*}
    \Theta_{(l,\gamma), \lambda, 1} = \xi_{(l,\gamma),\lambda} \delta_{(l,\gamma), \lambda} \mathbf{u}_l^T \text{ and } \Theta_{(l,\gamma), \lambda, 2} = e_{\gamma}^TS_2(t)U
\end{align*}
where $e_{\gamma}$ represents the $\gamma$-th coordinate vector. Then, given $0 \le \beta_1,...,\beta_k \le +\infty$ such that $\sum \limits_{\lambda=1}^k \frac{1}{\beta_{\lambda}}=1-\frac{k}{2q}$, $\tau_1, \ldots, \tau_k \in \sym(\{1,2 \})$, there exists $c_{\ref*{lem:decouptraceineqose}}>0$ such that
\begin{align*}
    &\sum_{(l,\gamma) \in [n] \times [m]} \E[ \tr \Theta_{(l,\gamma), 1, \tau_1(1)}^T\Theta_{(l,\gamma), 1, \tau_1(2)}
	\Upsilon_1\Theta_{(l,\gamma), 2, \tau_2(1)}^T\Theta_{(l,\gamma), 2, \tau_2(2)} \\&\cdots
	\Upsilon_2\Theta_{(l,\gamma), k, \tau_k(1)}^T\Theta_{(l,\gamma), k, \tau_k(2)}\Upsilon_k]  \\ \le &
       (c_{\ref*{lem:decouptraceineqose}}(pm)^\frac{1}{q}\sqrt{\max\{pd,q\}})^{\frac{2qk-4q}{2q-2}} (\E \tr((\Gamma(t))^{2q}))^{\frac{1}{q} \cdot \frac{2q-k}{2q-2}} \prod \limits_{\lambda=1}^k \norm{\Upsilon_{\lambda}}_{\beta_{\lambda}}
\end{align*}
\end{lemma}

\begin{proof}
    The structure of the proof is exactly the same as the proof of Lemma \ref{lem:decouptraceineq}, and only the specific expressions differ. We define the functional $F(\Upsilon_1,...,\Upsilon_k)$ exactly as in the proof of Lemma \ref{lem:decouptraceineq}, and proceed to prove the claim for when  $\beta_1= \cdots =\beta_{k-1}=\infty$ and $\beta_k=\frac{q}{q-k}$. 
    
    Let $\eta = (\eta(1), \eta(2))$ be a uniformly distibuted random variable in $[n]\times [m]$ and for all $\lambda \in [k]$, define random variables $\Theta_{1, \lambda} = \Theta_{\eta, \lambda, 1}$ and $\Theta_{2, \lambda} = \Theta_{\eta, \lambda, 2}$. Then, for $\mathcal{S} = mn$,
    
    \begin{align*}
        &\sum_{l \in [n], \gamma \in [m]} \E[ \tr \Theta_{(l,\gamma_l), 1, \tau_1(1)}^T\Theta_{(l,\gamma_l), 1, \tau_1(2)}
	\Upsilon_1\Theta_{(l,\gamma_l), 2, \tau_2(1)}^T\Theta_{(l,\gamma_l), 2, \tau_2(2)}\\&\cdots
	\Upsilon_2\Theta_{(l,\gamma_l), k, \tau_k(1)}^T\Theta_{(l,\gamma_l), k, \tau_k(2)}\Upsilon_k] \\ =
    &\mathcal{S} \cdot \E[
	\tr \Theta_{\tau_{1}(1),1}^T\Theta_{\tau_{1}(2),1}
	\Upsilon_1\Theta_{\tau_{2}(1),2}^T\Theta_{\tau_{2}(2),2}\cdots
	\Upsilon_2\Theta_{\tau_{k}(1),k}^T\Theta_{\tau_{k}(2),k}\Upsilon_k]
    \end{align*}

Defining factor 1 exactly as in Lemma \ref{lem:decouptraceineq}, we have,
\begin{align*}
(\text{factor 1})^2 \le & \norm{\Theta_{\tau_{1}(2),1}^T\Theta_{\tau_{1}(1),1}}_{2q}^{2-2/r} \prod\limits_{\lambda  = 2}^{k/2}\norm{\Theta_{\tau_{\lambda}(1),\lambda}^T\Theta_{\tau_{\lambda}(2),\lambda}}_{2q}^2\prod\limits_{\lambda  = 1}^{k/2} {\norm{\Upsilon_\lambda }_{\infty}^2} \\ &\cdot \norm{|\Theta_{\tau_{1}(2),1}^T\Theta_{\tau_{1}(1),1}|^{1/r} |\Upsilon_k| |\Theta_{\tau_{1}(2),1}^T\Theta_{\tau_{1}(1),1}|^{1/r}}_r
\end{align*}

To proceed, we prove an analogue of Lemma \ref{theta1theta2} to bound $\norm{\Theta_{\tau_j(1),j}^T\Theta_{\tau_j(2),j}}_{q}$,

\begin{lemma} \label{theta1theta2ose}
Let $\Theta_{\tau_j(1),j}, \Theta_{\tau_j(2),j} \in \R^d$ be as in Lemma \ref{lem:decouptraceineqose}. Let $\eta = (\eta(1), \eta(2))$ be a uniformly distributed random variable in $[n]\times[m]$ such that $\eta$ is independent with $\{\mathcal{Z}_{\lambda}\}_{\lambda \in [k]}$, $S_1, S_2, G_1, G_2$. For all $\lambda \in [k]$, define random vectors $\Theta_{1, \lambda} = \Theta_{\eta, \lambda, 1}$ and $\Theta_{2, \lambda} = \Theta_{\eta, \lambda, 2}$. Let $q \ge \log(pm)$. Then there exists a constant $c_{\ref*{theta1theta2ose}} >0$ such that,
\begin{align*}
\norm{\Theta_{\tau_j(1),j}^T\Theta_{\tau_j(2),j}}_{q} &\le  \frac{ c_{\ref*{theta1theta2ose}}\sqrt{\max \{ pd, q \}} }{\mathcal{S}^\frac{1}{q}}
\end{align*}
\end{lemma}

\begin{proof}
Note that for each fixed $1 \le j \le k/2$, 
\begin{align*}
    \norm{\Theta_{\tau_j(1),j}^T\Theta_{\tau_j(2),j}}_{q} &= \norm{\xi_{\eta} \delta_{\eta} \mathbf{u}_{\eta(1)} e_{\eta(2)}^TS_2(t)U}_{q} \\
    &\le \left( \E [\tr \abs{\xi_{\eta} \delta_{\eta} \mathbf{u}_{\eta(1)} e_{\eta(2)}^TS_2(t)U}^{q} ] \right)^\frac{1}{q} \\
    &= \left( \E [\tr \abs{ \delta_{\eta} \mathbf{u}_{\eta(1)} e_{\eta(2)}^TS_2(t)U}^{q} ] \right)^\frac{1}{q} \\
    &\le \left( \E \left[ \frac{1}{d} \norm{ \delta_{\eta} \mathbf{u}_{\eta(1)} }^{q}\norm{e_{\eta(2)}^TS_2(t)U}^{q} \right] \right)^\frac{1}{q} \\
\end{align*} 
Conditioning on $\eta$, we take expectation of the first factor over $\delta_{\eta}$ and the second factor over $S_2$,  
\begin{align*}
    \norm{\Theta_{\tau_j(1),j}^T\Theta_{\tau_j(2),j}}_{q} &\le \left( \E_{\eta} \left[ \frac{1}{d} \E_{\delta_{\eta}}\sqbr*{\norm{\delta_{\eta} \mathbf{u}_{\eta(1)} }^{q}} \E_{S_2(t)}[ \norm{e_{\eta(2)}^TS_2(t)U}^{q}] \right] \right)^\frac{1}{q} \\
\end{align*}

Note that, after conditioning on $\eta$, $\eta(2)$ is fixed. In other words, $\eta(2)$ is uniformly distributed over a one element subset of $[m]$. Therefore, by Lemma \ref{lem:rownormbound}, we have  
\begin{align*}
    \norm{\Theta_{\tau_j(1),j}^T\Theta_{\tau_j(2),j}}_{q} &\le \left( \E_{\eta} \left[ \frac{1}{d} \E_{\delta_{\eta}}\sqbr*{\norm{\delta_{\eta} \mathbf{u}_{\eta(1)} }^{q}} \E_{S_2(t)}[ \norm{e_{\eta(2)}^TS_2(t)U}^{q}] \right] \right)^\frac{1}{q} \\
    &\le c_{\ref*{lem:rownormbound}}\sqrt{\max \{ pd, q \}} \left( \E_{\eta} \left[ \frac{p}{d} \norm{ \mathbf{u}_{\eta(1)} }^{q} \right] \right)^\frac{1}{q} \\
    &\le c_{\ref*{lem:rownormbound}}\sqrt{\max \{ pd, q \}} \left( \frac{p}{\mathcal{S}d} \sum_{l=1}^n \sum_{\gamma=1}^{m}  \norm{\mathbf{u}_l}^{q} \right)^\frac{1}{q} \\
    &= c_{\ref*{lem:rownormbound}}\sqrt{\max \{ pd, q \}} \left( \frac{pm}{\cS  d} \sum_{l=1}^{n} \norm{u_l}^{q} \right)^\frac{1}{q} \\
    &\le c_{\ref*{lem:rownormbound}}\sqrt{\max \{ pd, q \}} \left( \frac{pm}{\cS  d} \sum_{l=1}^{n} \norm{u_l}^{2} \right)^\frac{1}{q} \\
    &\le \frac{c_{\ref*{theta1theta2ose}}(pm)^\frac{1}{q}\sqrt{\max \{ pd, q \}}}{\cS^\frac{1}{q}}
\end{align*}

where in the last line we use the fact that $\sum_{l=1}^{n} \norm{u_l}^{2} = d$.

\end{proof}

Using the same calculations as Lemma \ref{lem:decouptraceineq} for the term \begin{align*}\norm{|\Theta_{\tau_{1}(2),1}^T\Theta_{\tau_{1}(1),1}|^{1/r} |\Upsilon_k| |\Theta_{\tau_{1}(2),1}^T\Theta_{\tau_{1}(1),1}|^{1/r}}_r\end{align*} we get 
\begin{align*}
    (\text{factor 1})^2
 &\le  \paren*{\frac{ c_{\ref*{theta1theta2ose}}\sqrt{\max \{ pd, q \}} }{\mathcal{S}^\frac{1}{2q}}}^{k-2/r}\prod\limits_{\lambda  = 1}^{k/2} {\norm{\Upsilon_\lambda }_{\infty}^2} \cdot  \paren*{\frac{1}{\cS} \norm{|\Upsilon_k|^{r}}_{\frac{2q}{2q-2}} \cdot \paren*{\E \tr((\Gamma(t))^{2q})}^\frac{1}{q}}^\frac{1}{r} \\
 &\le \frac{1}{\cS} \paren*{ c_{\ref*{theta1theta2ose}}\sqrt{\max \{ pd, q \}}  }^{k-2/r}\prod\limits_{\lambda  = 1}^{k/2} {\norm{\Upsilon_\lambda }_{\infty}^2} \cdot   \norm{\Upsilon_k}_{\frac{2q}{2q-k}} \cdot \paren*{\E \tr((\Gamma(t))^{2q})}^\frac{1}{qr}  \\
\end{align*}
By repeating the same argument for factor 2, we get,
\begin{align*}
    &\sum_{l \in [n], \gamma_l \in [s_l]} \E[ \tr \Theta_{(l,\gamma_l), 1, \tau_1(1)}^T\Theta_{(l,\gamma_l), 1, \tau_1(2)}
	\Upsilon_1\Theta_{(l,\gamma_l), 2, \tau_2(1)}^T\Theta_{(l,\gamma_l), 2, \tau_2(2)} \\&\cdots
	\Upsilon_2\Theta_{(l,\gamma_l), k, \tau_k(1)}^T\Theta_{(l,\gamma_l), k, \tau_k(2)}\Upsilon_k]  \\ \le &
       (c_{\ref*{theta1theta2ose}}\sqrt{\max \{ pd, q \}})^{\frac{2qk-4q}{2q-2}} (\E \tr((\Gamma(t))^{2q}))^{\frac{1}{q} \cdot \frac{2q-k}{2q-2}} ({\prod\limits_{\lambda  = 1}^{k-1} {\norm{\Upsilon_\lambda }_{\infty}} )\cdot \norm{\Upsilon_k}_{\frac{2q}{2q-k}}}
\end{align*}

\end{proof}

\section{Conclusions}

We give an oblivious subspace embedding with optimal embedding dimension that achieves near-optimal sparsity, thus nearly matching a conjecture of Nelson and Nguyen in terms of the best sparsity attainable by an optimal oblivious subspace embedding. We also propose a fast algorithm for constructing low-distortion subspace embeddings, based on a new family of Leverage Score Sparsified embeddings with Independent Columns (LESS-IC). This new algorithm leads to speedups in downstream applications such as optimization problems based on constrained or regularized least squares. As a by-product of our analysis, we develop a new set of tools for matrix universality, combining a decoupling argument with a two-dimensional interpolation method, which are likely of independent interest.

\printbibliography

@article{tropp2025comparison,
  title={Comparison theorems for the minimum eigenvalue of a random positive-semidefinite matrix},
  author={Tropp, Joel A},
  journal={arXiv preprint arXiv:2501.16578},
  year={2025}
}

@inproceedings{cohen2015dimensionality,
  title={Dimensionality reduction for k-means clustering and low rank approximation},
  author={Cohen, Michael B and Elder, Sam and Musco, Cameron and Musco, Christopher and Persu, Madalina},
  booktitle={Proceedings of the forty-seventh annual ACM symposium on Theory of computing},
  pages={163--172},
  year={2015}
}

@book{vershyninhdp,
  title={High-Dimensional Probability: An Introduction with Applications in Data Science},
  author={Vershynin, R.},
  isbn={9781108244541},
  series={Cambridge Series in Statistical and Probabilistic Mathematics},
  url={https://books.google.com/books?id=TahxDwAAQBAJ},
  year={2018},
  publisher={Cambridge University Press}
}

@inproceedings{derezinski2024recent,
  title={Recent and Upcoming Developments in Randomized Numerical Linear Algebra for Machine Learning},
  author={Derezi{\'n}ski, Micha{\l} and Mahoney, Michael W},
  booktitle={Proceedings of the 30th ACM SIGKDD Conference on Knowledge Discovery and Data Mining},
  pages={6470--6479},
  year={2024}
}

@inproceedings{sarlos2006improved,
  title={Improved approximation algorithms for large matrices via random projections},
  author={Sarlos, Tamas},
  booktitle={2006 47th annual IEEE symposium on foundations of computer science (FOCS'06)},
  pages={143--152},
  year={2006},
  organization={IEEE}
}

@inproceedings{clarkson2013low,
  title={Low rank approximation and regression in input sparsity time},
  author={Clarkson, Kenneth L and Woodruff, David P},
  booktitle={Proceedings of the forty-fifth annual ACM symposium on Theory of Computing},
  pages={81--90},
  year={2013}
}

@inproceedings{drineas2006sampling,
  title={Sampling algorithms for $\ell_2$ regression and applications},
  author={Drineas, Petros and Mahoney, Michael W and Muthukrishnan, S},
  booktitle={Proceedings of the seventeenth annual ACM-SIAM symposium on Discrete algorithm},
  pages={1127--1136},
  year={2006},
}

@article{newton-less,
  title={Newton-LESS: Sparsification without Trade-offs for the Sketched Newton Update},
  author={Derezi\'nski, Micha{\l} and Lacotte, Jonathan and Pilanci, Mert and Mahoney, Michael W},
  journal={Advances in Neural Information Processing Systems},
  volume={34},
  pages={2835--2847},
  year={2021}
}

@inproceedings{less-embeddings,
  title={Sparse sketches with small inversion bias},
  author={Derezi\'nski, Micha{\l} and Liao, Zhenyu and Dobriban, Edgar and Mahoney, Michael},
  booktitle={Conference on Learning Theory},
  pages={1467--1510},
  year={2021},
  organization={PMLR}
}

@inproceedings{gaussianization,
  title={Algorithmic gaussianization through sketching: Converting data into sub-gaussian random designs},
  author={Derezi{\'n}ski, Micha{\l}},
  booktitle={The Thirty Sixth Annual Conference on Learning Theory},
  pages={3137--3172},
  year={2023},
  organization={PMLR}
}

@article{drineas2012fast,
  title={Fast approximation of matrix coherence and statistical leverage},
  author={Drineas, Petros and Magdon-Ismail, Malik and Mahoney, Michael W and Woodruff, David P},
  journal={The Journal of Machine Learning Research},
  volume={13},
  number={1},
  pages={3475--3506},
  year={2012},
  publisher={JMLR. org}
}

@inproceedings{meng2013low,
  title={Low-distortion subspace embeddings in input-sparsity time and applications to robust linear regression},
  author={Meng, Xiangrui and Mahoney, Michael W},
  booktitle={Proceedings of the forty-fifth annual ACM symposium on Theory of computing},
  pages={91--100},
  year={2013}
}

@inproceedings{cohen2016nearly,
  title={Nearly tight oblivious subspace embeddings by trace inequalities},
  author={Cohen, Michael B},
  booktitle={Proc. of the 27th annual ACM-SIAM Symposium on Discrete Algorithms},
  pages={278--287},
  year={2016},
  organization={SIAM}
}

@inproceedings{cohen2016optimal,
  title={Optimal approximate matrix product in terms of stable rank},
  author={Cohen, Michael B and Nelson, Jelani and Woodruff, David P},
  booktitle={International Colloquium on Automata, Languages, and Programming},
  year={2016},
  organization={Schloss Dagstuhl-Leibniz-Zentrum fur Informatik GmbH, Dagstuhl Publishing}
}

@article{woodruff2014sketching,
  title={Sketching as a tool for numerical linear algebra},
  author={Woodruff, David P},
  journal={Foundations and Trends{\textregistered} in Theoretical Computer Science},
  volume={10},
  number={1--2},
  pages={1--157},
  year={2014},
  publisher={Now Publishers, Inc.}
}

@article{martinsson2020randomized,
  title={Randomized numerical linear algebra: Foundations and algorithms},
  author={Martinsson, Per-Gunnar and Tropp, Joel A},
  journal={Acta Numerica},
  volume={29},
  pages={403--572},
  year={2020},
  publisher={Cambridge University Press}
}

@article{drineas2016randnla,
  title={RandNLA: randomized numerical linear algebra},
  author={Drineas, Petros and Mahoney, Michael W},
  journal={Communications of the ACM},
  volume={59},
  number={6},
  pages={80--90},
  year={2016},
  publisher={ACM New York, NY, USA}
}

@inproceedings{dasgupta2010sparse,
  title={A sparse johnson: Lindenstrauss transform},
  author={Dasgupta, Anirban and Kumar, Ravi and Sarl{\'o}s, Tam{\'a}s},
  booktitle={Proceedings of the forty-second ACM symposium on Theory of computing},
  pages={341--350},
  year={2010}
}

@article{kane2014sparser,
  title={Sparser johnson-lindenstrauss transforms},
  author={Kane, Daniel M and Nelson, Jelani},
  journal={Journal of the ACM (JACM)},
  volume={61},
  number={1},
  pages={1--23},
  year={2014},
  publisher={ACM New York, NY, USA}
}

@article{ailon2009fast,
  title={The fast Johnson--Lindenstrauss transform and approximate nearest neighbors},
  author={Ailon, Nir and Chazelle, Bernard},
  journal={SIAM Journal on computing},
  volume={39},
  number={1},
  pages={302--322},
  year={2009},
  publisher={SIAM}
}

@article{tropp2011improved,
  title={Improved analysis of the subsampled randomized Hadamard transform},
  author={Tropp, Joel A},
  journal={Advances in Adaptive Data Analysis},
  volume={3},
  number={01n02},
  pages={115--126},
  year={2011},
  publisher={World Scientific}
}

@inproceedings{nelson2013osnap,
  title={OSNAP: Faster numerical linear algebra algorithms via sparser subspace embeddings},
  author={Nelson, Jelani and Nguy{\^e}n, Huy L},
  booktitle={2013 ieee 54th annual symposium on foundations of computer science},
  pages={117--126},
  year={2013},
  organization={IEEE}
}

@article{achlioptas2003database,
  title={Database-friendly random projections: Johnson-Lindenstrauss with binary coins},
  author={Achlioptas, Dimitris},
  journal={Journal of computer and System Sciences},
  volume={66},
  number={4},
  pages={671--687},
  year={2003},
  publisher={Elsevier}
}

@article{brailovskaya2022universality,
  title={Universality and sharp matrix concentration inequalities},
  author={Brailovskaya, Tatiana and van Handel, Ramon},
  journal={Geometric and Functional Analysis},
  pages={1--105},
  year={2024},
  publisher={Springer}
}

@book{boucheron2013concentration,
  title={Concentration inequalities: A nonasymptotic theory of independence},
  author={Boucheron, St{\'e}phane and Lugosi, G{\'a}bor and Massart, Pascal},
  year={2013},
  publisher={Oxford university press}
}

@inproceedings{rudelson2010non,
  title={Non-asymptotic theory of random matrices: extreme singular values},
  author={Rudelson, Mark and Vershynin, Roman},
  booktitle={Proceedings of the International Congress of Mathematicians 2010 (ICM 2010) (In 4 Volumes) Vol. I: Plenary Lectures and Ceremonies Vols. II--IV: Invited Lectures},
  pages={1576--1602},
  year={2010},
  organization={World Scientific}
}

@inproceedings{li2022lower,
  title={Lower bounds for sparse oblivious subspace embeddings},
  author={Li, Yi and Liu, Mingmou},
  booktitle={Proceedings of the 41st ACM SIGMOD-SIGACT-SIGAI Symposium on Principles of Database Systems},
  pages={251--260},
  year={2022}
}

@inproceedings{nelson2014lower,
  title={Lower bounds for oblivious subspace embeddings},
  author={Nelson, Jelani and Nguy{\^e}n, Huy L},
  booktitle={International Colloquium on Automata, Languages, and Programming},
  pages={883--894},
  year={2014},
  organization={Springer}
}

@inproceedings{chepurko2022near,
  title={Near-optimal algorithms for linear algebra in the current matrix multiplication time},
  author={Chepurko, Nadiia and Clarkson, Kenneth L and Kacham, Praneeth and Woodruff, David P},
  booktitle={Proceedings of the 2022 Annual ACM-SIAM Symposium on Discrete Algorithms (SODA)},
  pages={3043--3068},
  year={2022},
  organization={SIAM}
}

@inproceedings{cherapanamjeri2023optimal,
  title={Optimal algorithms for linear algebra in the current matrix multiplication time},
  author={Cherapanamjeri, Yeshwanth and Silwal, Sandeep and Woodruff, David P and Zhou, Samson},
  booktitle={Proceedings of the 2023 Annual ACM-SIAM Symposium on Discrete Algorithms (SODA)},
  pages={4026--4049},
  year={2023},
  organization={SIAM}
}

@inproceedings{bourgain2015toward,
  title={Toward a unified theory of sparse dimensionality reduction in euclidean space},
  author={Bourgain, Jean and Dirksen, Sjoerd and Nelson, Jelani},
  booktitle={Proceedings of the forty-seventh annual ACM symposium on Theory of Computing},
  pages={499--508},
  year={2015}
}

@article{cartis2021hashing,
  title={Hashing embeddings of optimal dimension, with applications to linear least squares},
  author={Cartis, Coralia and Fiala, Jan and Shao, Zhen},
  journal={arXiv preprint arXiv:2105.11815},
  year={2021}
}

@book{vershynin2018high,
  title={High-dimensional probability: An introduction with applications in data science},
  author={Vershynin, Roman},
  volume={47},
  year={2018},
  publisher={Cambridge university press}
}

@article{bandeira2016sharp,
  title={Sharp nonasymptotic bounds on the norm of random matrices with independent entries},
  author={Bandeira, Afonso S and van Handel, Ramon},
  journal={Annals of Probability},
  volume={44},
  number={4},
  pages={2479--2506},
  year={2016},
  publisher={Institute of Mathematical Statistics}
}

@book{kallenberg2021foundations,
  title={Foundations of modern probability Third Edition},
  author={Kallenberg, Olav and Kallenberg, Olav},
  volume={2},
  year={2021},
  publisher={Springer}
}

@book{folland2013real,
  title={Real analysis: modern techniques and their applications},
  author={Folland, Gerald B},
  year={2013},
  publisher={John Wiley \& Sons}
}

@article{troppmatrixconc,
author = {Tropp, Joel A.},
title = {An Introduction to Matrix Concentration Inequalities},
year = {2015},
issue_date = {05 2015},
publisher = {Now Publishers Inc.},
address = {Hanover, MA, USA},
volume = {8},
number = {1–2},
issn = {1935-8237},
journal = {Found. Trends Mach. Learn.},
month = {may},
pages = {1–230},
numpages = {230}
}

@inproceedings{chenakkod2024optimal,
  title={Optimal embedding dimension for sparse subspace embeddings},
  author={Chenakkod, Shabarish and Derezi{\'n}ski, Micha{\l} and Dong, Xiaoyu and Rudelson, Mark},
  booktitle={Proceedings of the 56th Annual ACM Symposium on Theory of Computing},
  pages={1106--1117},
  year={2024}
}

@book{bergh2012interpolation,
  title={Interpolation spaces: an introduction},
  author={Bergh, J{\"o}ran and L{\"o}fstr{\"o}m, J{\"o}rgen},
  volume={223},
  year={2012},
  publisher={Springer Science \& Business Media}
}

@incollection{pisier2003non,
  title={Non-commutative Lp-spaces},
  author={Pisier, Gilles and Xu, Quanhua},
  booktitle={Handbook of the geometry of Banach spaces},
  volume={2},
  pages={1459--1517},
  year={2003},
  publisher={Elsevier}
}

@article{carlen2010trace,
  title={Trace inequalities and quantum entropy: an introductory course},
  author={Carlen, Eric},
  journal={Entropy and the quantum},
  volume={529},
  pages={73--140},
  year={2010}
}


\end{document}